\documentclass[11pt,orivec]{llncs}  
\usepackage{amsmath,amssymb,url,ifthen,fullpage}

\DeclareMathAlphabet{\mathpzc}{OT1}{pzc}{m}{it}

\newcommand{\A}{\mathcal{A}}
\newcommand{\B}{\mathcal{B}}

\newcommand{\G}{\mathbb{G}}
\newcommand{\Ghat}{\hat{\mathbb{G}}}
\newcommand{\ghat}{\hat{g}}
\newcommand{\hhat}{\hat{h}}
\newcommand{\vhat}{\hat{v}}
\newcommand{\what}{\hat{w}}
\newcommand{\Z}{\mathbb{Z}}
\newcommand{\N}{\mathbb{N}}

\newcommand{\PR}{\mathrm{Pr}}

\newcommand{\id}{\mathsf{id}}

\newcommand{\Sid}{\Sigma_{\mathsf{ID}}}

\newcommand{\HF}{\mathsf{HF}}
\newcommand{\HFSetup}{\textsf{HF.Setup}}
\newcommand{\HFMKg}{\textsf{HF.MKg}}
\newcommand{\HFKg}{\textsf{HF.Kg}}
\newcommand{\HFEv}{\textsf{HF.Eval}}
\newcommand{\HFInv}{\textsf{HF.Inv}}
\newcommand{\HFDel}{\textsf{HF.Del}}

\newcommand{\LHF}{\mathsf{LHF}}

\newcommand{\LHFMKg}{\textsf{LHF.MKg}}
\newcommand{\LHFKg}{\textsf{LHF.Kg}}

\newcommand{\InpSp}{\mathsf{InpSp}}
\newcommand{\IdSp}{\mathsf{IdSp}}
\newcommand{\AxSp}{\mathsf{AuxSp}}

\newcommand{\barw}{\bar{w}}

\newcommand{\checkeq}{\stackrel{?}{=}}

\newcommand{\CC}{\mathbf{CT}}
\newcommand{\DD}{\mathbf{D}}
\newcommand{\KK}{\mathbf{K}}

\newcommand{\JJ}{\mathbf{J}}
\newcommand{\PP}{\mathsf{PP}}

\newcommand{\hh}{\mathbf{h}}

\newcommand{\err}{\mathbf{r}}

\newcommand{\valpha}{\boldsymbol{\alpha}}

\newcommand{\vzeta}{\boldsymbol{\zeta}}
\newcommand{\vgamma}{\boldsymbol{\gamma}}

\newcommand{\ww}{\mathbf{w}}
\newcommand{\CCC}{\mathbf{C}}

\newcommand{\yy}{\mathbf{y}}
\newcommand{\ess}{\mathbf{s}}
\newcommand{\REAL}{\mathsf{REAL}}
\newcommand{\LOSSY}{\mathsf{LOSSY}}

\newcommand{\DE}{\mathsf{HIB}\textrm{-}\mathsf{DE}}
\newcommand{\DESetup}{\textsf{HIB-DE.Setup}}
\newcommand{\DEMKg}{\textsf{HIB-DE.MKg}}
\newcommand{\DEKg}{\textsf{HIB-DE.Kg}}
\newcommand{\DEEnc}{\textsf{HIB-DE.Enc}}
\newcommand{\DEDec}{\textsf{HIB-DE.Dec}}
\newcommand{\DEDel}{\textsf{HIB-DE.Del}}

\def \sample { \stackrel{\begin{footnotesize} _{R}
\end{footnotesize}}{\leftarrow} }

\title{Hierarchical Identity-Based (Lossy) Trapdoor Functions}

\author{Alex Escala \inst{1} \and Javier Herranz\inst{1} \and Beno\^it Libert\inst{2}    \and Carla
R\`afols\inst{3} }

\institute{Universitat Polit\`ecnica de Catalunya, Dept.
Matem\`atica Aplicada IV
 (Spain) \\
{\small e-mail: {\tt $\{$alex.escala,jherranz$\}$@ma4.upc.edu}  }
\and Universit\'e catholique de Louvain, ICTEAM Institute -- Crypto
Group
(Belgium) \\
{\small e-mail: {\tt benoit.libert@uclouvain.be}  } \and
Ruhr-Universit\"at Bochum, Horst G\"ortz Institut f\"ur IT-Sicherheit
 (Germany)  \\
{\small e-mail: {\tt Carla.Rafols@rub.de}  } }

\begin{document}

\maketitle 

\begin{abstract}

Lossy trapdoor functions, introduced by Peikert and Waters
(STOC'08), have received a lot of attention in the last years,
because of their wide range of applications in theoretical
cryptography. The notion  has been recently extended to the
identity-based setting by Bellare \textit{et al.} (Eurocrypt'12). We
provide one more step in this direction, by considering the notion
of hierarchical identity-based (lossy) trapdoor functions
(HIB-TDFs).
Hierarchical identity-based cryptography has proved very useful
both for practical applications and to establish theoretical
relations with other cryptographic primitives.

The notion of security for IB-TDFs put forward by Bellare \textit{et
al.}   easily extends to the hierarchical scenario, but an (H)IB-TDF
secure in this sense is not known to generically imply  other
related primitives with security against adaptive-id adversaries,
not even IND-ID-CPA secure encryption. Our first contribution is to
define a new security property for (H)IB-TDFs. We show that
functions satisfying this property imply secure cryptographic
primitives in the adaptive  identity-based setting: these include
encryption schemes with semantic security under chosen-plaintext
attacks, deterministic encryption schemes, and (non-adaptive) hedged
encryption schemes that maintain some security when messages are
encrypted using randomness of poor quality. We emphasize that some
of these primitives were unrealized in the (H)IB setting previous to
this work.

As a  second contribution, we describe  the first pairing-based
HIB-TDF realization. This  is also the first example of hierarchical
trapdoor function based on traditional number theoretic assumptions:
 so far, constructions were only known under lattice
assumptions.
 Our  HIB-TDF construction is based on   techniques that differ from those of Bellare {\it et al.} in that it uses  a
hierarchical predicate encryption scheme as a key ingredient. The
resulting HIB-TDF is proved to satisfy the new security definition,
against
 either selective or, for hierarchies of constant depth, adaptive adversaries. In either case,
 we only need
 the underlying predicate encryption system to be selectively secure.
\smallskip

\end{abstract}



\section{Introduction}

\subsection{(Identity-Based) Lossy Trapdoor Functions}

Lossy trapdoor functions, as introduced by Peikert and Waters in
\cite{PW08}, have been proved very powerful in theoretical
cryptography and received a lot of attention in the recent years
(see, e.g., \cite{FGKRS10,HO11,MY10,BW10,Hof12,Wee12}). Roughly
speaking, a lossy trapdoor function is a family of functions that
can be instantiated in two different modes. In the injective mode,
the function is injective and can be inverted using the
corresponding trapdoor. In lossy mode, the function is (highly)
non-injective since its image size is much smaller than the size of
the domain. The key point is that lossy instantiations of the
function must be indistinguishable from injective instantiations. \\
\indent In their seminal paper  \cite{PW08}, Peikert and Waters
showed that lossy trapdoor functions   provide  black-box
constructions of chosen-ciphertext secure (IND-CCA)  public-key
encryption schemes as well as universal one-way and
collision-resistant hash functions. Later on, other applications of
lossy trapdoor functions were discovered: they gave rise to
deterministic encryption schemes \cite{BBO07} in the standard model
\cite{BFO08}, public-key encryption hedged schemes maintaining some
security in the absence of reliable encryption coins \cite{BBN09}
and even public-key encryption with selective-opening security
\cite{BHY09} ({\it i.e.}, which offer certain
security guarantees in case of sender corruption). \\
\indent Recently, Bellare, Kiltz, Peikert and Waters \cite{BKPW}
introduced the notion of identity-based (lossy) trapdoor function (IB-TDF),
which is the analogue of lossy trapdoor functions in the setting of
identity-based cryptography \cite{Sha}.  In the identity-based
scenario, users' public keys are directly derived from their
identities, whereas secret keys are delivered by a trusted master
entity. In this way, the need for digital certificates, which
usually bind public keys to users in traditional public-key
cryptography, is drastically reduced. Throughout the last decade, several generalizations of identity-based cryptography were put forth,
 including hierarchical identity-based cryptography \cite{GS0}, attribute-based cryptography \cite{SW05,GPSW} or
 predicate-based cryptography \cite{BW07,KSW08}. In the setting of hierarchical
 identity-based cryptography, identities are organized in a hierarchical way, so that a user who holds the
 secret key of an identity $\id$ can generate, use and distribute valid secret keys for any identity that is a descendant of $\id$ in the hierarchy.
 Hierarchical identity-based encryption (HIBE) is of great interest due to both practical and theoretical reasons.
 On the practical side, many organizations and systems that may need (identity-based) cryptographic solutions are organized in a hierarchical way.
  On the theoretical side,  generic constructions \cite{CHK03,CHK} are known to
  transform a weakly secure HIBE scheme  ({\it i.e.}, IND-CPA security against selective adversaries) into  (public-key)
  encryption schemes with strong security properties,
 like chosen-ciphertext security \cite{CHK} or  forward-security \cite{And,CHK03}, where private keys are updated in such a way that past encryptions remain
 safe after a private key exposure.  \\
\indent Bellare \textit{et al.} \cite{BKPW} proposed instantiations
of identity-based lossy trapdoor functions  based on bilinear
 maps and on   lattices (as noted in \cite{BKPW}, almost all IBE
schemes belong to these families). The former makes clever use of an
anonymous IBE system (where the ciphertext hides the receiver's
identity) with pseudorandom ciphertexts whereas the latter relies on
lossiness properties of   learning-with-error-based cryptosystems.
Moreover, they show that their definition of partial-lossiness for identity-based
 trapdoor functions
leads to the same cryptographic results as lossy trapdoor functions,
but in the selective identity-based setting only, where the attacker
must choose the target identity upfront in the attack game. Namely,
in the case of selective adversaries, IB-TDFs satisfying their
definition imply identity-based encryption with semantic security,
identity-based deterministic encryption and identity-based hedged
encryption. In \cite{BKPW}, it was left  as an open problem to prove
that the same results hold in the case of adaptive adversaries.

\subsection{Our Contributions}

This paper  extends to the hierarchical setting the notion of
identity-based (lossy) trapdoor function. \\ \vspace{-0.3 cm}

\noindent \textsc{New Definition of Partial Lossiness.} From a
theoretical standpoint, we first define a new security property for
hierarchical identity-based trapdoor functions (HIB-TDFs). We show
that a HIB-TDF which satisfies this property can be used to obtain
the same kind of results that are derived from standard lossy
trapdoor functions \cite{PW08}. Namely, they lead to standard
encryption schemes, to deterministic encryption schemes, and to
non-adaptive hedged encryption schemes, which are secure in the
hierarchical identity-based setting, in front of adaptive-id
adversaries. Since HIB-TDFs contain IB-TDFs as a particular case,
our results for adaptive adversaries solve an open problem in
\cite{BKPW}. Interestingly,  the pairing-based IB-TDF of Bellare
{\it et al.} \cite{BKPW} can be proved to also satisfy the new
security property. As a consequence, it provides  adaptively secure
deterministic and hedged IBE schemes. Recently, Xie {\it et al.}
\cite{XXZ}  designed an adaptively secure D-IBE system using
lattices. The   construction of \cite{BKPW} is thus the first adaptively secure pairing-based realization.  \\
\vspace{-0.3 cm}

\noindent \textsc{Construction of a Pairing-Based Hierarchical
  Trapdoor Function.}  On the constructive side,
we  focus on pairing-based systems where, as already mentioned in
\cite{BKPW}, greater challenges are faced.
Indeed,     in the hierarchical scenario, anonymity -- which was
used as an important ingredient by Bellare {\it et al.} \cite{BKPW}
-- has been significantly harder to obtain in the world of bilinear
maps than with lattices \cite{CHKP,ABB10a,ABB10b}: for example, the
first pairing-based anonymous HIBE system \cite{BW06} appeared four
years after the first collusion-resistant HIBE \cite{GS0}. Moreover,
very few
anonymous IBE systems seem amenable for constructing IB-TDFs, as noted in \cite{BKPW} where a new scheme was specially designed for this purpose.  \\
\indent
 Using bilinear
maps, we thus construct 
a HIB-TDF and prove that it satisfies our new definition of partial
lossiness under
 mild hardness assumptions in groups of prime order. As an
intermediate step, we design 
a hierarchical predicate encryption (HPE) system \cite{ShiWa08,OT09}
with suitable anonymity  properties, which may be of independent
interest. Perhaps surprisingly, although this scheme is proved
secure only against weak \textit{selective} adversaries (who select
their target attribute set before seeing the public parameters), we
are able to turn it into a HIB-TDF providing security (namely, our
new version of partial lossiness) against \textit{adaptive}
adversaries for hierarchies of constant depth. To the best of our
knowledge, our HIB-TDF    gives rise to
 the first   hierarchy of trapdoor functions which does not rely on lattices: realizing such
 a hierarchy
 using   traditional number theoretic techniques was identified as an open problem by Cash {\it et al.} \cite{CHKP}. \\
 \indent Beyond its hierarchical nature, our construction
brings out an alternative design principle for (H)IB-TDFs. The idea
 is to rely on hierarchical predicate encryption (HPE) to deal with hierarchies. Namely, public parameters consist
of  a matrix of
 HPE encryptions and, when the function has to be evaluated,
  the latter matrix is turned into a matrix of (anonymous) HIBE ciphertexts. The  homomorphic properties of the underlying HIBE then make it
 possible to evaluate the function  while guaranteeing a sufficient amount of lossiness in lossy mode.  \\ \indent While
the pairing-based construction of Bellare {\it et al.} \cite{BKPW}
builds on an adaptively secure anonymous IBE, our (hierarchical)
IB-TDF is obtained from a {\it selectively} weakly attribute-hiding
 HPE system. This result is
somewhat incomparable with \cite{BKPW}:  on one hand, we start from
a more powerful primitive -- because predicate encryption implies
anonymous IBE -- but, on the other hand, we need a weaker security
level to begin with. Both (H)IB-TDF constructions rely on specific
algebraic properties in the underlying IBE/HPE and neither is
generic.  It would be interesting to see if a more general approach
exists for
building such functions. \\

\subsection{Discussion on the Implications}

  Combining our HIB-TDF with the theoretical implications of
our new security property, we obtain: (1) a semantically  secure
HIBE encryption scheme under adaptive adversaries, (2) the first
secure deterministic HIBE scheme\footnote{See \cite{XXZ} for a
recent and independent construction, in the (non-hierarchical) IBE case.}, (3)
the first HIBE scheme that (non-adaptively) hedges against bad
randomness, as advocated by Bellare {\it et al.} \cite{BBN09}. \\
\indent
  In the case of adaptive adversaries, these results only hold
for hierarchies of constant depth (said otherwise, we do not provide
full security). However,  using our definition of partial lossiness
or that of Bellare {\it et al.} \cite{BKPW}, this appears very
difficult to avoid. The reason is that both definitions seem
inherently bound to the partitioning paradigm. Namely, they assume
the existence of alternative public parameters, called {\it lossy
parameters}, where the identity space is partitioned into subsets of
injective  and lossy identities. The definition of \cite{BKPW}
intuitively captures that a fraction $\delta$ of identities are
lossy in the case of lossy parameters. In the hierarchical setting,
the analogy with HIBE schemes suggests that all ancestors of a lossy
identity   be lossy themselves.   Hence, unless one can make sure
that certain lossy identities only have lossy descendants, the
fraction
 $\delta$ seems  doomed to exponentially decline with the depth of the hierarchy.    \\
\indent Finally, due to the results of Canetti, Halevi and Katz
\cite{CHK03}, our construction also implies the first forward-secure
deterministic and hedged public-key encryption schemes (note that
selective security suffices to   give forward-secure public-key
cryptosystems). Although our scheme is not practical due to very
large ciphertexts and key sizes, it provides the first feasibility
results in these directions.


\section{Hierarchical Identity-Based (Lossy) Trapdoor Functions, and Applications}
\label{HIB-TDF-general}

 In this section we extend to the
hierarchical scenario the definitions for identity-based (lossy)
trapdoor functions given in \cite{BKPW}. \\ \vspace{-0.3 cm}

\noindent \textsc{Syntax.} A hierarchical identity-based trapdoor
function (HIB-TDF) is a tuple of efficient 
algorithms $\HF=(\HFSetup , \HFMKg , \HFKg, \HFDel , \HFEv , \HFInv
)$. The setup algorithm $\HFSetup$ takes as input a security
parameter $\varrho \in \mathbb{N}$, the (constant) number of levels
in the hierarchy $d\in\N$, the length of the identities $\mu\in
\mathsf{poly}(\varrho)$ and the length of the function inputs
$n\in\mathsf{poly}(\varrho)$, and outputs a set of global public
parameters $\mathsf{pms}$, which specifies an input space $\InpSp$,
an identity space $\IdSp$ and the necessary mathematical objects and
hash functions. The master key generation algorithm $\HFMKg$ takes
as input $\mathsf{pms}$ and outputs a master public key
$\mathsf{mpk}$ and a master secret key $\mathsf{msk}$. The key
generation algorithm $\HFKg$ takes as input $\mathsf{pms}$,
$\mathsf{msk}$ and a hierarchical identity
$(\id_1,\ldots,\id_{\ell})\in \IdSp$, for some $\ell \geq 1$ and
outputs a secret key $\mathbf{SK}_{(\id_1,\ldots,\id_\ell)}$. The
delegation algorithm $\HFDel$ takes as input $\mathsf{pms}$,
$\mathsf{msk}$, a hierarchical identity $(\id_1,\ldots,\id_{\ell})$,
a secret key $\mathbf{SK}_{(\id_1,\ldots,\id_\ell)}$ for it, and an
additional identity $\id_{\ell + 1}$; the output is a secret key
$\mathbf{SK}_{(\id_1,\ldots,\id_\ell,\id_{\ell + 1})}$ for the
hierarchical identity $(\id_1,\ldots,\id_\ell,\id_{\ell + 1})$ iff
$(\id_1,\ldots,\id_\ell,\id_{\ell + 1})\in \IdSp$. The evaluation
algorithm $\HFEv$ takes as input $\mathsf{pms}$, $\mathsf{msk}$, an
 identity $\id=(\id_1,\ldots,\id_{\ell})$ and a value $X
\in \InpSp$; the result of the evaluation is denoted as $C$.
Finally, the inversion algorithm $\HFInv$ takes as input
$\mathsf{pms}$, $\mathsf{msk}$, a hierarchical identity
$\id=(\id_1,\ldots,\id_{\ell})$, a secret key $\mathbf{SK}_{\id}$
for it and an evaluation $C$, and outputs a value $\tilde{X} \in \InpSp$. \\
 \indent
A HIB-TDF satisfies the property of correctness if
$$
\HFInv \big(\mathsf{pms},\mathsf{mpk},\id,\mathbf{SK}_{\id}, \HFEv
\big(\mathsf{pms},\mathsf{mpk},\id=(\id_1,\ldots,\id_{\ell}),X \big)
\big) \ =\ X ,
$$
for any $X \in \InpSp$, any
$\mathsf{pms},(\mathsf{mpk},\mathsf{msk})$ generated by  $\HFSetup$
and $\HFMKg$, any hierarchical identity $(\id_1,\ldots,\id_\ell)\in
\IdSp$ and any secret key $\mathbf{SK}_{(\id_1,\ldots,\id_\ell)}$
generated either by running
$\HFKg\big(\mathsf{pms},\mathsf{msk},(\id_1,\ldots,\id_{\ell})
\big)$ or by applying   the delegation algorithm
$\HFDel$ to secret keys of shorter hierarchical identities. \\
\indent Before formalizing the new definition of partial lossiness
for a HIB-TDF, let us recall the notion of \emph{lossiness}: if $f$
is a function with domain Dom$(f)$ and image Im$(f)= \{ f(x) \ :\ x
\in \textrm{Dom}(f)\}$, we say that $f$ is $\omega$-lossy if
$\lambda(f) \geq \omega$, where $\lambda(f) = \log
\frac{|\textrm{Dom}(f)|}{|\textrm{Im}(f)|}$. \\
\indent To properly define lossiness in the IB setting, it will also
be useful to consider extended HIB-TDFs, which differ from standard
HIB-TDFs in that, in the latter, the algorithm $\HFSetup$ specifies
in $\mathsf{pms}$ an auxiliary input space $\AxSp$, and $\HFMKg$
takes as additional auxiliary input $aux \in \AxSp$. Also, given
some HIB-TDF $\HF=(\HFSetup , \HFMKg , \HFKg, \HFDel , \HFEv ,
\HFInv )$, a \emph{sibling} for $\HF$ is an extended HIB-TDF
$\LHF=(\HFSetup , \LHFMKg , \LHFKg,$ $ \HFDel , \HFEv , \HFInv )$
whose  delegation, evaluation and inversion algorithms are those of
$\HF$, and where an auxiliary space $\mathsf{AuxSp}$ is contained in
$\mathsf{pms} \leftarrow \HFSetup(\varrho)$, so that $\IdSp  \subset
\mathsf{AuxSp}$.\\ \indent Looking ahead, we will define, as in
\cite{BKPW}, two different experiments: one corresponding to the
standard setup and one corresponding to the lossy setup, in one of
them the experiment will interact with a standard HIB-TDF, in the
other one with a sibling in which some identities lead to lossy
evaluation functions. The notion of extended HIB-TDF will serve to
construct both of these functions as an extended HIB-TDF but with
different auxiliary inputs $\vec{y}^{(0)}, \vec{y}^{(1)}$.

\subsection{A New Security Definition for HIB-TDFs}\label{sec:new_sec_def_HIBTDF}

The basic security property of a trapdoor function is
\emph{one-wayness}, which  means that the function is hard to invert
without the suitable secret key.  In the identity-based setting,
one-wayness is required to hold even when the adversary has oracle
access to the secret keys for some identities. \emph{Partial
lossiness} for identity-based trapdoor functions was introduced in
\cite{BKPW}, where it was  proved to imply one-wayness. Roughly
speaking, partial lossiness requires that the weighted difference of
the probability that any adversary outputs $1$ in the lossy or in
the real setup is negligible. For the selective case, the weights
can simply be set to $1$ and it can be proved that an IB-TDF
satisfying their notion of partial lossiness in the selective
scenario can be used to build: (1) identity-based encryption (IBE)
schemes with selective IND-CPA security, (2) selectively secure
deterministic IBE schemes, (3) selectively secure hedged IBE
schemes. However, these results are not known to be true in the
adaptive setting, in particular, the definition is not even known to
yield an IND-ID-CPA secure encryption scheme. \\
\indent To address this question, we propose  an alternative
definition for the partial lossiness of (hierarchical)
identity-based trapdoor functions --- in particular, they also
result in a new definition when the hierarchy depth is equal to $1$,
the case considered by Bellare \textit{et al.} \cite{BKPW}.  We will
show that a HIB-TDF satisfying this new definition gives a secure
construction of the same primitives we mentioned for the selective
case.\\ \indent
As in \cite{BKPW}, we define two different experiments, a lossy
experiment and a real experiment. For any adversary $\A$ against a
HIB-TDF, and any $\zeta \in (0,1)$, the
$\REAL_{\HF,\LHF,\mathcal{P},\omega,\zeta}^\A$ experiment and the
$\LOSSY_{\HF,\LHF,\mathcal{P},\omega,\zeta}^\A$ experiment are
parameterized by the security parameter $\varrho$ (which we will
usually omit in the notation) and values
$\zeta(\varrho),\omega(\varrho)$. The value $\zeta$ will be
important to show that our construction implies other cryptographic
primitives. Intuitively,  $\zeta$ will be the advantage of an
adversary against a cryptographic scheme the security of which is
reduced to the security of the HIB-TDF. The experiment also takes as
input the specification of some efficient algorithm $\mathcal{P}$
which takes as input
$\zeta,\mathsf{pms},\mathsf{mpk}_1,\mathsf{msk}_1,$ $IS,$
$\id^\star$, and outputs a bit $d_2$. This procedure $\mathcal{P}$
can be, in general, any probabilistic and polynomial-time algorithm.
 In the security analysis of the selectively secure version of our
HIB-TDF, $\mathcal{P}$ will be the trivial algorithm that always
outputs $d_2=1$. In other cases, $\mathcal{P}$ could be a more
complicated algorithm; for instance, in order to prove the security
of the adaptive-id version of our HIB-TDF, we will take as
pre-output stage $\mathcal{P}$ Waters' artificial abort step.
Actually, procedure $\mathcal{P}$ is a way to relax the security
requirements in Denition \ref{HIB-TDF-sec-def}, in order to allow
the possibility of building cryptographic schemes from secure
HIB-TDFs, in a black-box way.

Finally, the value $\omega$ is
related to the lossiness of the trapdoor function. To simplify
notation, we will simply write $\REAL$ instead of
$\REAL_{\HF,\LHF,\mathcal{P},\omega,\zeta}^\A$ and $\LOSSY$ instead
of $\LOSSY_{\HF,\LHF,\mathcal{P},\omega,\zeta}^\A$.

%
%

For compactness, we present the two experiments as a single
experiment depending of a bit $\beta$: the  challenger
$\mathcal{C}$, who interacts with the adversary $\A$, runs either
$\REAL$ if $\beta=0$ or $\LOSSY$ if $\beta=1$. Also, some
instructions of both experiments depend on whether \emph{selective}
 or \emph{adaptive} security is being considered. We say
that   a hierarchical identity $\id = (\id_1,\ldots,\id_{\ell})$ is
a \emph{prefix} of another one $\id^\star
=(\id^\star_1,\ldots,\id^\star_{\ell^\star})$ if $\ell \leq
\ell^\star$ and $\id_i = \id^\star_i$ for every $i =1,\ldots,\ell$.
We denote it by $\id\leq\id^\star$.

\begin{itemize}
\item[0.] $\mathcal{C}$ chooses  global
parameters $\mathsf{pms}$ by running $\HFSetup$. The parameters
$\mathsf{pms}$ are given to $\A$, who replies by choosing a
hierarchical identity $\id^\dagger
=(\id_1^\dagger,\ldots,\id_{\ell^\dagger}^\dagger)$, for some
$\ell^\dagger \leq d$.

\item[1.] $\mathcal{C}$ runs $(\mathsf{mpk}_0,\mathsf{msk}_0) \leftarrow \HFMKg(\mathsf{pms})$ and
 $(\mathsf{mpk}_1,\mathsf{msk}_1) \leftarrow \LHFMKg(\mathsf{pms},aux=\id^\dagger)$. The adversary $\A$ receives $\mathsf{mpk}_\beta$ and
 lists $IS \leftarrow \emptyset,\ QS\leftarrow \emptyset$ are initialized.

\item[2.]   $\A$ can make adaptive queries for hierarchical identities $\id = (\id_1,\ldots,\id_{\ell})$ and identities $\id_{\ell+1}$.
\\ \vspace{-0.3 cm}
\begin{itemize}
\item[-] {\bf Create-key}: $\A$ provides $\id$
and $\mathcal{C}$ creates a private key $SK_\id$. If $\beta=0$,
$SK_\id$ is created by running
$\HFKg(\mathsf{pms},\mathsf{msk}_0,\id )$. If $\beta=1$, it is
created by running $\LHFKg(\mathsf{pms},\mathsf{msk}_1,\id )$. The
list $QS$ is updated as $QS=QS\cup \{\id\}$.
\item[-] {\bf Create-delegated-key}: $\A$ provides $\id=(\id_1,\ldots,\id_{\ell})$ and $\id_{\ell+1}$ such that $\id\in QS$.
The challenger $\mathcal{C}$ then computes $SK_{\id'}$ for
$\id'=(\id_1,\dots,\id_{\ell+1})$ by running the delegation
algorithm
$\HFDel\big(\mathsf{pms},\mathsf{mpk}_\beta,SK_\id,\id_{\ell+1}\big)$.
The list $QS$ is updated as $QS=QS\cup \{\id'\}$.
\item[-] {\bf Reveal-key}: $\A$ provides $\id$ with the
restriction that if $\A$ is selective, then
$\id\not\leq\id^\dagger$. $\mathcal{C}$ returns $\perp$ if
$\id\not\in QS$. Otherwise, $SK_{\id}$ is returned to $\A$ and the
list $IS$ is updated   as $IS=IS
\cup \{\id\}$.\\
\vspace{-0.3 cm}
\end{itemize}

\item[3.] The adversary $\A$ outputs a hierarchical identity
$\id^\star=(\id_1^\star,\ldots,\id_{\ell^\star}^\star)$ and a   bit
$d_\A \in \{0,1\}$. If $\A$ is selective, then
$\id^\star=\id^\dagger$. In the adaptive case, no element of $IS$
can be a prefix of $\id^\star$. Let $d_1$ be the bit
$d_1:=\big(\forall\ \id\in IS, \lambda\left( \HFEv
(\mathsf{pms},\mathsf{mpk}_1,\id,\cdot ) \right)=0 \big)\wedge\big(
\lambda\left( \HFEv (\mathsf{pms},\mathsf{mpk}_1,\id^\star,\cdot )
\right)\geq\omega\big) $.

\item[4.] $\mathcal{C}$ sets $d_2 $ to be the output of the pre-output stage $\mathcal{P}$ with input
$\zeta,\mathsf{pms},\mathsf{mpk}_1,\mathsf{msk}_1,$ $IS,$
$\id^\star$.

\item[5.] The final output of the experiment consists of $\{d_\A,\ d_{\neg abort}^\A\}$, where $d_{\neg abort}^\A=d_1\wedge d_2 \in \{0,1\}$.
\end{itemize}
For notational convenience, from now on, let us define
$d_{exp}^\A=d_\A\wedge d_{\neg abort}^\A$.

Note that, in the lossy experiment, some identities may lead to
lossy functions, which can be detected by $\A$ if it queries the
secret key for such an identity. This causes an asymmetry when
comparing the real and the lossy experiments. For this reason,
Bellare \textit{et al.} defined the advantage of a distinguisher
among the lossy and real experiments as the \emph{weighted}
difference of the probability of outputting $1$ in the real case
minus the same probability in the lossy case. Our solution is
different: we force the experiment to adopt the same behavior in the
real and lossy settings: if a query would force the $\LOSSY$
experiment to set the bits $d_1$ or $d_2$ to 0 (and, as a
consequence, setting $d_{exp}^\A=0$), then it also forces the
$\REAL$ experiment to set $d_1$ or $d_2$ to $0$, respectively.
Therefore, the difference of probabilities in our definition (see
condition (i) below) is not weighted.

\begin{definition}\label{HIB-TDF-sec-def}
A HIB-TDF is $(\omega,\delta)$-partially lossy if it admits a sibling and an efficient pre-output stage $\mathcal{P}$ such that for all PPT adversaries $\A$ and for all non-negligible $\zeta$, there exist two non-negligible values $\epsilon_1,\epsilon_2$  with $\delta=\epsilon_1 \epsilon_2$ such that the following three conditions hold:
\begin{itemize}
\item[(i)] the following advantage function is negligible in the security parameter $\varrho$:
\begin{eqnarray}\label{cond1lossy}
\mathbf{Adv}_{\HF,\LHF,\mathcal{P},\omega,\zeta}^{\mathrm{lossy}}(\A)= |\Pr[d_{exp}^\A=1 |\ \REAL]\ -\
\Pr[d_{exp}^\A=1 |\ \LOSSY]|
\end{eqnarray}
\medskip
\item[(ii)] $\Pr[d_{\neg abort}^{\mathcal{A}}=1\ |\
\REAL] \geq \epsilon_1.$
\medskip
\item[(iii)] if $\A$ is such that $\Pr[d_\A=1\ |\ \REAL]\ -\ \frac{1}{2}>\zeta$, then
\begin{equation}\label{cond3lossy}
\Pr[d_{\mathcal{A}}=1\ |\ \REAL \wedge d_{\neg
abort}^{\mathcal{A}}=1 ] -\ \frac{1}{2} >\epsilon_2 \cdot \zeta,
\end{equation}
where $\delta$ may be a function of $q$ the maximal number of secret key queries of $\A$.
\end{itemize}

\end{definition}

As   said above, condition (i) is  a simple modification of the
definition of partial lossiness of Bellare {\it et al.} \cite{BKPW}.
We add condition (ii) to rule out some cases in which the definition
would be trivial to satisfy, like the case where the procedure
$\mathcal{P}$ aborts with overwhelming probability or the case where
the sibling admits only lossy identities:   in any of these
scenarios, we would have $d_{\neg abort}=1$ with negligible
probability, which would
  render the scheme useless with the sole condition
(i). Then, we add a third condition (iii) which allows reducing
HIB-TDFs to other primitives. Roughly speaking, this condition
guarantees that the probability of aborting is somewhat independent
of the behavior of any computationally bounded adversary.
Interestingly, it is possible to prove (by proceeding exactly as in
the proof of our HIB-TDF) that the pairing-based IB-TDF described by
Bellare {\it et al.} \cite{BKPW} satisfies our new partial-lossiness
definition.

\subsection{Implications of Lossy (H)IB-TDFs: the Example of  (H)IBE}\label{sec:app_IND_CPA}

Using the same argument as in \cite{BKPW}, it is quite easy to prove
that a HIB-TDF which enjoys the new version of the partial lossiness
property is already one-way, in both the selective and adaptive
settings.
In this section we prove that a HIB-TDF which satisfies our new
security definition can be used to build other primitives in the
hierarchical identity-based security, with security against adaptive
adversaries. We detail the example of hierarchical identity-based
encryption (HIBE) with IND-CPA security\footnote{The cases of
deterministic HIBE and hedged HIBE are discussed in Appendices
\ref{deterministic-scary-part} and \ref{hedged-appendix}.}. The construction is the
direct adaptation of the Peikert-Waters construction \cite{PW08} in
the public-key setting.

Let $\HF$ be a HIB-TDF with message space
$\{0,1\}^n$ and lossiness $\omega$, and $\mathcal{H}$ a family of
pairwise independent hash functions from $\{0,1\}^n$ to $\{0,1\}^l$
where $l\leq \omega-2\lg(1/\epsilon_{LHL})$ for some negligible
$\epsilon_{LHL}$. The HIBE scheme has message space $\{0,1\}^l$. Its setup, key generation and key delegation algorithms are basically the same ones as those
for $\HF$, the rest are as follows:
\begin{center}
  \begin{tabular}{ l | c | c| r }
   \textbf{MKGen}($\mathsf{pms}$)  &   \textbf{Enc}($\mathsf{pms}, \mathsf{mpk},m,\id$)  &   \textbf{Dec}($\mathsf{pms},  \mathsf{mpk},\mathbf{SK}_{\id}, C, \id$)   \\ \hline
    $(\mathsf{mpk}',\mathsf{msk})  \leftarrow \HFMKg(1^k)~$ & $x\gets\{0,1\}^n$ & $~x = \HFInv(\mathsf{pms},  \mathsf{mpk},\mathbf{SK}_{\id}, c_1, \id)$  \\
   $h\gets \mathcal{H}$ & $~c_1=\HFEv(\mathsf{pms},\mathsf{mpk},\id,x)~$ & $m= c_2 \oplus h(x)$ \\
        $\mathsf{mpk}=(\mathsf{mpk}',h)$   &$ c_2=h(x)\oplus m$& Return $m$ \\
         Return $\mathsf{mpk}$   &  Return $C=(c_1,c_2)$  &
            \end{tabular}
\end{center}
We prove the following theorem.
\begin{theorem}\label{cpa-construction-secure}If $\HF$ is $(\omega,\delta)$-partially lossy for some non-negligible value of $\delta$, then the HIBE scheme $\Pi$ described is IND-ID-CPA secure. In particular, for every IND-ID-CPA adversary $\B$ against $\Pi$ there exists a PPT adversary $\A$ against $\HF$  such that
$$\mathbf{Adv}_{\HF,\LHF,\mathcal{P},\omega,\zeta}^{\mathrm{lossy}}(\A) \geq \frac{2 }{3} \cdot \delta  \cdot  \mathbf{Adv}^{\mathrm{ind-id-cpa}}(\B)- \nu(\varrho)$$
for some  negligible  function $\nu$; both adversaries $\A$ and $\B$
run in comparable times.
\end{theorem}

\begin{proof}
Let us assume that an adversary $\B$  has advantage at least $\zeta$
in breaking the IND-ID-CPA security of the HIBE scheme $\Pi$, for
some non-negligible $\zeta$.  We
  build an adversary $\A$ that breaks the condition (i) of
Definition \ref{HIB-TDF-sec-def} assuming that conditions (ii) and
(iii) are satisfied.
Our adversary $\A$, who interacts with a challenger that runs either
the experiment $\REAL$ or the experiment $\LOSSY$, proceeds to
simulate the challenger in the IND-ID-CPA game with $\B$ as follows.

$\A$ forwards an identity $\id^\dagger$ to its challenger, which is some random identity in the adaptive case or corresponds to the challenge identity chosen by
$\B$ in the selective case.  When the challenger
runs the setup and gives the output to $\A$, $\A$ forwards this
information to $\B$ together with a hash function $h \leftarrow \mathcal{H}$. When $\B$ asks for a secret key for a
hierarchical identity $\id$, $\A$ forwards the query to the
experiment and forwards the reply to $\B$. At some point, $\B$ outputs $(m_0, m_1,\id^\star)$, with $\id^\dagger=\id^*$ in the selective case.
Adversary $\A$ then forwards $\id^\star$ to its challenger, chooses $\gamma\gets\{0,1\}$ at random and encrypts $m_\gamma$ under the identity $\id^\star$.
After some more secret key queries, $\B$ outputs a guess $\gamma'$ and $\A$ outputs $d_\A=1$ if $\gamma=\gamma'$ and
$d_\A=0$ otherwise.


In the $\REAL$ setting, we will have  $$ \Pr[\gamma'=\gamma|\
\REAL]-\frac{1}{2}=\Pr[d_\A=1|\ \REAL]-\frac{1}{2}\geq \zeta,$$
since $\A$ perfectly simulated the IND-ID-CPA game with $\B$. This
inequality can be combined with conditions (ii) and (iii) of the
definition of $(\omega,\delta)$-partial lossiness (which we assume
to be satisfied by $\HF$), and we obtain
\begin{eqnarray}\label{telopeto}
\Pr[d_{\neg abort}^{\mathcal{A}}=1\ |\ \REAL] \cdot \left(\ \Pr[d_{\mathcal{A}}=1\ |\ \REAL \wedge d_{\neg abort}^{\mathcal{A}}=1 ]
-\ \dfrac{1}{2}\ \right)    >\epsilon_1 \epsilon_2 \zeta.
\end{eqnarray}
On the other hand, as proved in \cite{BKPW}, in the $\LOSSY$ setting
when $\id^\star$ is lossy, the advantage of $\B$ in guessing
$\gamma$ is negligible. Indeed, since we are using a pairwise
independent hash function,  the Leftover Hash Lemma \cite{HILL}
(more precisely, its variant proved in \cite{DRS04}) implies that
the distribution of $c_2$
  given $c_1$ is statistically close to the uniform
distribution. We thus have $\Pr[d_\A=1|\ \LOSSY\wedge d_{\neg
abort}^\A=1]\leq \dfrac{1}{2}+\epsilon_{LHL}$, for some negligible
function $\epsilon_{LHL}$. Since $d_{exp}^\A=d_{\neg abort}^\A\wedge
d_\A$, we find
\begin{eqnarray} \nonumber
\Pr[d_{exp}^\A=1\ |\ \LOSSY]&=& \Pr[d_\A=1\ |\ \LOSSY\wedge d_{\neg
abort}^\A=1]\Pr[d_{\neg abort}^\A=1\ |\ \LOSSY]\\
\nonumber & \leq& \big(\dfrac{1}{2} + \epsilon_{LHL} \big) \cdot
\Pr[d_{\neg abort}^\A=1\ |\ \LOSSY] \\ \label{ineq-1} &\leq&
\dfrac{1}{2} \cdot \big( \Pr[d_{\neg abort}^\A=1\ |\ \REAL] +
\mathbf{Adv}_{\HF,\LHF,\mathcal{P},\omega,\zeta}^{\mathrm{lossy}}(\A)
\big)+\nu,
\end{eqnarray}
for some negligible function $\nu \in \mathsf{negl}(\varrho)$.  The
last equality follows from the fact that we can assume that
$\Pr[d_{\neg abort}^\A=1|\ \LOSSY] - \Pr[d_{\neg abort}^\A=1|\
\REAL] \leq \mathbf{Adv}_{\HF,\LHF,\mathcal{P},\omega,\zeta}^{\mathrm{lossy}}(\A)$: otherwise, we can easily build a
distinguisher\footnote{This distinguisher $\A_1$ is obtained from
$\A$ by ignoring  $d_{\A} \in \{0,1\}$ and replacing it by a $1$, so
that $d_{\neg abort}^{\A}=d_{exp}^{\A}$.}  against condition (i) of
the partial lossiness definition. If we plug (\ref{ineq-1}) into the
definition of
$\mathbf{Adv}_{\HF,\LHF,\mathcal{P},\omega,\zeta}^{\mathrm{lossy}}(\A)$,
we obtain
\begin{eqnarray}  \nonumber
\mathbf{Adv}_{\HF,\LHF,\mathcal{P},\omega,\zeta}^{\mathrm{lossy}}(\A)&=&
\left|\Pr[d_{exp}^\A=1\ |\ \REAL]-\Pr[d_{exp}^\A=1\ |\
\LOSSY]\right| \\ \label{contrad} & \geq &\left| \Pr[d_{\neg
abort}^\A=1\ |\ \REAL] \cdot \left(\Pr[d_\A=1\ |\ \REAL \wedge
d_{\neg abort}^\A=1]-\dfrac{1}{2}\right)   \right| \\ \nonumber & &
\qquad - \frac{1}{2} \cdot
\mathbf{Adv}_{\HF,\LHF,\mathcal{P},\omega,\zeta}^{\mathrm{lossy}}(\A) -\nu,
\end{eqnarray}
so that there exists $\tilde{\nu} \in \mathsf{negl}(\varrho)$ such
that
\begin{eqnarray*}  \nonumber
\mathbf{Adv}_{\HF,\LHF,\mathcal{P},\omega,\zeta}^{\mathrm{lossy}}(\A)& \geq &
\frac{2}{3} \cdot
 \left| \Pr[d_{\neg abort}^\A=1\ |\ \REAL] \cdot \left(\Pr[d_\A=1\ |\
\REAL \wedge d_{\neg abort}^\A=1]-\dfrac{1}{2}\right)   \right|
-\tilde{\nu}.
\end{eqnarray*}
This means that
$\mathbf{Adv}_{\HF,\LHF,\mathcal{P},\omega,\zeta}^{\mathrm{lossy}}(\A)$ is
non-negligible since $\delta=\epsilon_1 \epsilon_2$ and the
right-hand-side member of the above expression is at least $(2/3)
\cdot\delta \cdot \zeta -\nu$. In other words, any adversary
guessing $\gamma'=\gamma$ with non-negligible advantage $\zeta$ in
$\mathsf{REAL}$ necessarily contradicts condition (i). \qed
\end{proof}

\section{Interlude: Hierarchical Predicate Encryption}\label{sec:background}

In this section we propose a new hierarchical predicate encryption scheme with the attribute-hiding property, which will be used as an ingredient to build, in the next section, a hierarchical identity-based (lossy) trapdoor function. The syntax and security model for hierarchical predicate encryption schemes are recalled in Appendix \ref{app:syntax-HIPE}.

\subsection{Some Complexity Assumptions}

We consider  groups $(\G,\Ghat,\G_T)$ of prime order $p$ for which
an asymmetric bilinear map $e:\G \times \Ghat \rightarrow \G_T$ is
efficiently computable.  We will assume that the DDH assumption
holds in both $\G$ and $\Ghat$, which implies that  no isomorphism
is efficiently computable  between $\G$ and $\Ghat$. The assumptions
that we need are sometimes somewhat stronger than DDH. However, they
have {\it constant} size ({\it i.e.}, we de not rely on $q$-type
assumptions) and were previously used in \cite{Duc10}.

\begin{description}
 \item[The Bilinear Diffie Hellman Assumption (BDH):] in  bilinear
  groups $(\G,\Ghat,\G_T)$ of prime order $p$,
  the
  distribution
$D_1 =\{(g,~g^a,~g^c,~\ghat,~\ghat^a,~\ghat^b, ~ e(g,\ghat)^{abc} )\
~|~
 a,b,c \sample \Z_p \},$ is computationally indistinguishable from
 $D_2  =\{ (g,~g^a,~g^c,~\ghat,~\ghat^a,~\ghat^b, ~ e(g,\ghat)^{z} )\
~|~
 a,b,c,z \sample \Z_p  \}.$
 \item[The $\mathcal{P}$-BDH$_1$ Assumption: ]
  in  asymmetric bilinear groups $(\G,\Ghat,\G_T)$ of prime order $p$,
  the
  distribution
   $ D_1=\{(g ,  g^b, g^{ab},g^c, \ghat,\ghat^a,\ghat^b,
 g^{abc }     )~ |~
a,b,c \sample \Z_p  \}$ is computationally   indistinguishable from
 $D_2 =\{ (g ,  g^b,  g^{ab}, g^c,  \ghat, \ghat^a, \ghat^b,
  g^z     ) ~ | ~
   a,b,c,z \sample \Z_p  \}$.
 \item[The  DDH$_2$ Assumption: ]
  in  asymmetric bilinear groups  $(\G,\Ghat,\G_T)$ of prime order
  $p$, the distribution $D_1 = \{(g ,  ~\ghat,~\ghat^a, ~\ghat^b,
 ~\ghat^{ab  }     )\ |~
a,b  \sample \Z_p  \}$ is computationally indistinguishable
  from the distribution $D_2  = \{ (g ,
~\ghat,~\ghat^a, ~ \ghat^b,
 ~\ghat^{z }     ) \ |~
  a,b ,z \sample \Z_p  \}$.
\end{description}

\subsection{A Selectively Secure Weakly Attribute-Hiding Hierarchical
Predicate Encryption Scheme}\label{sec:new-HIPE}

The construction is inspired by the Shi-Waters
 delegatable predicate encryption scheme \cite{ShiWa08} (at a high-level, it also bears similarities with the lattice-based scheme of \cite{AFV11}). However,
we have to turn it into a predicate encryption scheme for inner
product relations (like the one of Okamoto and Takashima
\cite{OT09}) instead of a hidden vector encryption \cite{BW07}.
Another  difficulty to solve is that we cannot use composite order
groups as in \cite{ShiWa08} because, in our  HIB-TDF of Section
\ref{hierarchical-ltdf}, one of the subgroups would eventually leak
information on the input in lossy mode (this is actually what
happened with our initial attempt). For this reason, we chose to
work with prime-order groups and used asymmetric pairing
configurations to anonymize ciphertexts. As a benefit, we obtain a
better efficiency than by using the techniques of \cite{BW07} by
reducing the number of pairing evaluations.

\begin{description}
\item[Setup$(\varrho,d,\mu)$:] given a security parameter $\varrho \in
\mathbb{N}$, the (constant) desired number of levels in the hierarchy $d\in\N$ and the desired length $\mu$ of the attribute vectors at
each level, choose asymmetric bilinear groups $(\G,\Ghat,\G_T)$ of
order $p$, where $p >2^\varrho$. Choose $g \sample \G$, $\ghat
\sample \Ghat$. Then, pick $\alpha,\alpha_v,\alpha_w \sample \Z_p^*$
and set $v=g^{\alpha_v}$, $\hat{v}=\hat{g}^{\alpha_v}$,
$w=g^{\alpha_w}$ and $\hat{w}=\ghat^{\alpha_w}$.  For
$i_1=1,\ldots,d$ and $i_2=0,\ldots
  ,\mu$, choose $\alpha_{i_1,i_2} \sample \Z_p^*$ and compute
 $ h_{i_1,i_2}=g^{\alpha_{i_1,i_2}}  \in \G $ and $ \hhat_{i_1,i_2}=\ghat^{\alpha_{i_1,i_2}}  \in \G $  .  The master   public
key is defined to be
\begin{eqnarray*}
 \mathsf{mpk} & := &  \Bigl( v ,~w, ~e(g,\hat{v})^{\alpha} , ~\{ h_{i_1,i_2}   \}_{i_1 \in \{1,\ldots,d\},
~ i_2 \in \{0,\ldots,
  \mu \}}
\Bigr)
\end{eqnarray*}
while the master secret key is
$\mathsf{msk}:=\big(\hat{g},\hat{g}^{\alpha},\hat{v} ,\hat{w} ,
 \{ \hat{h}_{i_1,i_2} \}_{i_1 \in \{1,\ldots,d\}, ~i_2 \in \{0,\ldots, \mu \}} \big)$.
\\ \vspace{-0.3 cm}

\item[Keygen$\big(\mathsf{msk},(\vec{X}_1,\ldots,\vec{X}_{\ell}) \big)$:] to generate a private key for
 vectors $(\vec{X}_1,\ldots,\vec{X}_{\ell})$ with $\ell \leq d$, parse  $\mathsf{msk}$ as $
\big(\ghat,\hat{v} ,\hat{w},
 \{ \hat{h}_{i_1,i_2} \}_{i_1 \in \{1,\ldots,d\}, ~i_2 \in \{0,\ldots,\mu \}} \big)$.
 For $i_1=1$ to $\ell$, parse  $\vec{X}_{i_1}$ as $(x_{i_1,1},\ldots,x_{i_1,\mu}) \in \Z_p^{\mu}$.
 Choose  $r_w \sample \Z_p^*$ and
$r_1,\ldots,r_{\ell}  \sample \Z_p^*$,
  for $i_1 \in \{ 1,\ldots,\ell\}$. Then, compute the decryption component
$SK_D=(D,D_w,  \{D_{i_1}\}_{i_1=1}^\ell)$ of the key as
$$
  D = \hat{g}^{\alpha} \cdot  \prod_{i_1=1}^{\ell} \big(   \prod_{ i_2 =1}^{\mu} \hat{h}_{i_1,i_2}^{x_{i_1,i_2}}
\big)^{r_{i_1} } \cdot \hat{w}^{r_{w}}   ,
\qquad
 D_w  =  \hat{v}^{r_{w}}   ,  \qquad
   D_{i_1} = \hat{v}^{r_{i_1}}.$$
  \vspace{-0.3 cm}
To define the elements of its delegation component $SK_{DL}$
$$
   \big(\{K_{j,k},L_{j}
,L_{j,k,i_1}, L_{w,j,k} \}_{j \in \{\ell+1,\ldots,d\}, ~k \in
\{1,\ldots,\mu\},~i_1 \in \{1,\ldots,\ell\} }
  \big)
$$

pick $s_{j}  \sample \Z_p^*$, $s_{j,k,i_1} \sample \Z_p^*$,
$s_{w,j,k} \sample \Z_p^*$  for $i_1 \in \{1,\ldots,\ell\},\ i_2 \in
\{1,\ldots,\mu\}$ and set:
 $$ K_{j,k}  =  \prod_{i_1=1}^{\ell} \big(   \prod_{ i_2 =1}^{\mu}
\hat{h}_{i_1,i_2}^{x_{i_1,i_2}} \big)^{s_{j,k,i_1} } \cdot
\hat{h}_{j,k}^{s_{j}} \cdot \hat{w}^{s_{w,j,k}},\ \ L_j
=\hat{v}^{s_j},\ \ L_{j,k,i_1} = \hat{v}^{s_{j,k,i_1}}, \ \
L_{w,j,k} = \hat{v}^{s_{w,j,k}}.$$

 Output the private key $SK_{(\vec{X}_1,\ldots,\vec{X}_\ell)}=\bigl(
 SK_D,SK_{DL}\big).$ \\ \vspace{-0.2 cm}
 \item[Delegate$\big(\mathsf{mpk},(\vec{X}_1,\ldots,\vec{X}_\ell),SK_{(\vec{X}_1,\ldots,\vec{X}_\ell)},\vec{X}_{\ell+1}\big)$:]
parse the   key $SK_{(\vec{X}_1,\ldots,\vec{X}_\ell)}$ as
$(SK_D,SK_{DL})$. Given,
 $\vec{X}_{\ell+1}=(x_{\ell+1,1},\ldots,x_{\ell+1,\mu}) \in
\Z_p^{\mu}$,  do the following. \\ \vspace{-0.3 cm}
\begin{itemize}
\item[1.] Randomize $SK_{DL}$  by raising all its component to some $z \sample \Z_p^*$. Call this new key $\widehat{SK}_{DL}$ and write its elements with a hat (e.g.,
$\widehat{K}_{j,k}=K_{j,k}^z$).
\item[2.] Compute a \emph{partial decryption key}
$$
K_{\ell+1} = \prod_{k=1}^{\mu} \widehat
K_{\ell+1,k}^{x_{\ell+1,k}}
  =   \prod_{i_1=1}^{\ell} (\prod_{i_2=1}^{\mu}
\hat{h}_{i_1,i_2}^{x_{i_1,i_2}})^{s_{\ell+1,i_1}}  \cdot
\big(\prod_{k=1}^{\mu} \hat{h}_{\ell+1,k}^{x_{\ell+1,k}} \big)^{
s_{\ell+1}} \cdot \hat{w}^{s_{w,\ell+1}} ,\ \
L_{\ell+1,\ell+1}=\widehat L_{\ell+1},$$
$$
 L_{\ell+1,i_1} = \prod_{k=1}^{\mu} \widehat L_{\ell+1,k,i_1}^{x_{\ell+1,k}} = \hat{v}^{s_{\ell+1,i_1}}\ \ \textrm{for }   ~i_1 \in \{1,\ldots,\ell\},\ \
  \ \   L_{w,\ell+1} = \prod_{k=1}^{\mu} \widehat
L_{w,\ell+1,k}^{x_{\ell+1,k}} = \hat{v}^{s_{w,\ell+1}}$$ where we
define the exponents $ s_{\ell+1,i_1}  = z \cdot \sum_{k=1}^{\mu}
s_{\ell+1,k,i_1} \cdot x_{\ell+1,k}$  for   $i_1 \in
\{1,\ldots,\ell\}$, and $s_{w,\ell+1} = z \cdot \sum_{k=1}^{\mu}
s_{w,\ell+1,k} \cdot x_{\ell+1,k}$. \\ \vspace{-0.3 cm}
\item[3.] \label{rerandjk} For all $j\in \{\ell+2,\ldots,d\}$, $k\in \{1,\ldots,\mu\}$, compute re-randomized
versions of the partial decryption key by raising the partial
decryption key to a random power $\tau_{j,k} \sample \Z_p^*$.
$$
K_{\ell+1}^{(j,k)}=K_{\ell+1}^{\tau_{j,k}},  \qquad
L_{w,\ell+1}^{(j,k)}  =   L_{w,\ell+1}^{\tau_{j,k}}, \qquad
 \{ L_{\ell+1,i_1}^{(j,k)} =  L_{\ell+1,i_1}^{\tau_{j,k}} \}_{ i_1=1}^{\ell+1}.
$$
These values will be used to compute the delegation component of the
 new key at step 5. \\ \vspace{-0.3 cm}

\item[4.] Compute a decryption component $SK_D'=(D',D_w', \{D_{i_1}'\}_{i_1=1}^{\ell+1})$  for the delegated key by setting
  $D'=D \cdot K_{\ell+1}  $,  $D_w'=D_w \cdot L_{w,\ell+1} $. Then, define $D_{\ell+1}'=L_{\ell+1,\ell+1} $
and, for each $i_1 \in \{1,\ldots,\ell\}$, set $D_{i_1}'=D_{i_1}
\cdot L_{\ell+1,i_1} $.  \\
\vspace{-0.3 cm}

\item[5.] Compute a delegation component for the delegated key. For each $j \in \{\ell+2,\ldots,d\}$, set
 $L_j'=\widehat L_j  $. Then, for $k=1$ to $\mu$ and $i_1=1$ to $\ell+1$, set 
 \begin{eqnarray*}
K_{j,k}'  =  \widehat K_{j,k} \cdot K_{\ell+1}^{(j,k)} , \quad  ~
L_{w,j,k}' = \widehat L_{w,j,k} \cdot L_{w,\ell+1}^{(j,k)}  \quad ~
L_{j,k,i_1}'  =  \widehat L_{j,k,i_1} \cdot L_{\ell+1,i_1}^{(j,k)} ,
\end{eqnarray*}
where   $\widehat L_{j,k,\ell+1}=1$ for all $j,k$. The new
delegation component $SK_{DL}'$ is
$$
   \big(   \{K_{j,k}' , L_{j}', L_{j,k,i_1}',  L_{w,j,k}'\}_{j \in \{\ell+2,\ldots,d\}, ~k
\in \{1,\ldots,\mu\},~i_1 \in \{1,\ldots,\ell\} }
    \big)
$$
Return the delegated private key
$SK_{(\vec{X}_1,\ldots,\vec{X}_{\ell+1})}=(SK_D',SK_{DL}')$.
\end{itemize}

\item[Encrypt$\big(\mathsf{mpk},(\vec{Y}_1,\ldots,\vec{Y}_\kappa),M \big)$:] given $\mathsf{mpk}$,  a plaintext $M \in \G_T$ as well as a hierarchy of  vectors  $
\vec{Y}_1=(y_{1,1},\ldots,y_{1,\mu})$, \ldots,
$\vec{Y}_\kappa=(y_{\kappa,1},\ldots,y_{\kappa,\mu})
 $, choose $s \sample \Z_p^*$ and compute
\begin{eqnarray*}
C_0=M \cdot e(g,\hat{v})^{\alpha \cdot s}, \quad ~~~    C_v = v^{s}
, \quad ~~~ C_w= w^s  , \quad~~
   \{ C_{i_1,i_2}= \bigl( h_{i_1,0}^{
y_{i_1,i_2}} \cdot h_{i_1,i_2} \bigr)^s   \}_{ i_1 \in
\{1,\ldots,\kappa \},~i_2\in \{1,\ldots,\mu \}},
\end{eqnarray*}
The ciphertext is
$C=\bigl(C_0,C_v,C_{w}, \{C_{i_1,i_2}  \}_{i_1 \in
\{1,\ldots,\kappa \},~i_2 \in \{1,\ldots,\mu \} }  \bigr).$ \\
 \vspace{-0.3 cm}
\item[Decrypt$\big(\mathsf{mpk},(\vec{X}_1,\ldots,\vec{X}_\ell),SK_{(\vec{X}_1,\ldots,\vec{X}_\ell)},C\big)$:]
parse the private key $SK_{(\vec{X}_1,\ldots,\vec{X}_\ell)}$ as
$\bigl( SK_D,SK_{DL}\big)$, where $SK_D= (D,D_w,
\{D_{i_1}\}_{i_1=1}^\ell) $ and  the ciphertext $C$ as
$\bigl(C_0,C_v,C_{w}, \{C_{i_1,i_2} \}_{i_1 \in \{1,\ldots,\kappa
\},~i_2 \in \{1,\ldots,\mu \} } \bigr)$. 
\begin{itemize}
\item[1.] For each $i_1 \in \{1,\ldots,\ell \}$, compute
 $C_{i_1}=    \prod_{i_2=1}^{\mu}  C_{i_1,i_2}^{x_{i_1,i_2}}=\big(      h_{i_1,0}^{    \vec{X}_{i_1} \cdot \vec{Y}_{i_1}} \cdot \prod_{i_2=1}^{\mu} h_{i_1,i_2}^{x_{i_1,i_2}}
\big)^s  .$
\item[2.]  Return $M$ if
 $   M =  C_0 \cdot  e(C_v,D) ^{-1} \cdot e(C_{w},D_{w})  \cdot  \prod_{i_1=1}^{\ell} e(C_{i_1},D_{i_1})   $
 is in the appropriate subspace\footnote{As in \cite{BW07,KSW08}, the plaintext space is restricted to have a size much smaller than $|\G_T|$ to
 make sure that the decryption algorithm returns $\perp$ if an unauthorized key is used to decrypt. } of $\G_T$. Otherwise, return $\perp$.
\end{itemize}
\end{description}

We show that our scheme is correct in Appendix
\ref{proof-correctness}. Our HIB-TDF uses the predicate-only variant
of the above scheme,
 obtained by  discarding the ciphertext component $C_0$
(which contains the payload) and the factor $\ghat^{\alpha}$ from
the private key component $D$.

%
%
%

We remark that the number of pairing evaluations only depends on the
depth $\ell$ of the predicate $(\vec{X}_1,\ldots,\vec{X}_{\ell})$
encoded in the private key and not on the dimension $\mu $ of
vectors at each level. Except a recent fully secure construction
\cite{OT11}, all previous schemes required $O(\ell \cdot \mu)$
pairing evaluations to decrypt. In our HIB-TDF of Section
\ref{hierarchical-ltdf}, this will make it possible to invert a
function admitting $n$-bit inputs by computing $O(\ell \cdot n)$
pairings, instead of  $O(\ell \cdot n \cdot \mu)$.


%
%


%
%
%

 The new HPE  scheme is selectively weakly attribute-hiding under the
BDH, $\mathcal{P}$-BDH$_1$ and DDH$_2$ assumptions, as established
by Theorem \ref{thm-HIPE}. The security of its predicate-only
variant (which is the one used as a key ingredient in the design of
our HIB-TDF) relies only on the latter two assumptions.

\begin{theorem} \label{thm-HIPE}
The HPE scheme    is selectively weakly attribute-hiding (in the
sense of Definition \ref{sec-def} in Appendix \ref{app:syntax-HIPE})
if the BDH, $\mathcal{P}$-BDH$_1$ and DDH$_2$ assumptions hold in
$(\G,\Ghat,\G_T)$. \textnormal{(The proof is given in Appendix
\ref{proof-HPE}).}
\end{theorem}

\section{A Hierarchical Identity-Based (Lossy) Trapdoor Function} \label{hierarchical-ltdf}

From the HPE scheme of Section \ref{sec:new-HIPE}, our lossy
function is obtained by including a $n \times n$ matrix of HPE
ciphertexts in the master public parameters. As in the DDH-based
 function of  \cite{PW08}, each row of the
matrix is associated with an encryption exponent, which is re-used
throughout the entire row. Each column corresponds to a different
set of public parameters in the HPE system.
\\ \indent  The HIB-TDF that we construct is actually an extended
HIB-TDF, and so the master key generation protocol takes an
auxiliary input. Depending on the value of this auxiliary input, we
obtain the trapdoor (injective) function or a partially lossy
function, used in the security proofs. Actually, all  HPE ciphertexts in the above-mentioned matrix
  correspond to different hierarchical vectors
$(\yy_1,\ldots,\yy_d) \in \Z_p^{d \cdot \mu}$, depending on the auxiliary input. The selective weak
attribute-hiding property of the HPE scheme   guarantees that the
two setups are computationally indistinguishable. \\
\indent In order to evaluate a function for some hierarchical
identity $\id=(\id_1,\ldots,\id_{\ell})$, the first step of the
evaluation algorithm  computes a transformation on HPE ciphertexts
so as to obtain a matrix of  HIBE ciphertexts. During this
transformation, a set of inner products   $\{\langle \yy_{i_1}
,\id_{i_1} \rangle \}_{i_1=1}^{\ell}$ is calculated in the exponent
in the diagonal entries of the matrix. The transformation provides a
$n \times n$ matrix (\ref{HCT-2}) of anonymous HIBE ciphertexts that
are always well-formed in non-diagonal entries. As for diagonal
entries, they contain ``perturbed'' HIBE ciphertexts: at each level,
one ciphertext component contains a perturbation factor of the form
$\langle \yy_{i_1} ,\id_{i_1} \rangle $. In this matrix of HIBE
ciphertexts, random
encryption exponents are again re-used in all positions at each row. \\
\indent The  function evaluation is then carried out as in
\cite{PW08}, by computing a matrix-vector product in the exponent
and taking advantage of homomorphic properties of the HIBE scheme
over the randomness space. The function output can be seen as a set
of $n$ anonymous HIBE ciphertexts -- one for each input bit -- which
are
 well-formed  ciphertexts if and only if the corresponding input
bit is $0$ ({\it i.e.}, if and only if the perturbation factors
$\{\langle \yy_{i_1} ,\id_{i_1} \rangle \}_{i_1=1}^{\ell} $ are left
out when computing the matrix-vector product in the exponent). The
function is thus inverted by testing the well-formedness of each
HIBE ciphertext using the private key.

\subsection{Description}

\begin{description}
\item[$\HFSetup(\varrho,d,n,\mu)$:] given a security parameter $\varrho \in
\mathbb{N}$, the (constant) desired number of levels in the
hierarchy $d\in \mathbb{N}$ and integers $\mu,n \in
\mathsf{poly}(\varrho)$  specifying the length of identities and
that of function inputs, respectively, choose asymmetric bilinear
groups $(\G,\Ghat,\G_T)$ of prime order $p>2^\varrho$.  Define
$\InpSp = \{0,1\}^n$, $\Sid=\{(1,\mathbf{x}): \mathbf{x}\in
\Z_p^{\mu-1} \}$, $\IdSp = \Sid^{(\leq d)}$ and $\mathsf{AuxSp} =
\Z_p^{d \cdot \mu}$. The public parameters are $\mathsf{pms} =
\big(p,(\G,\Ghat,\G_T),d,n,\mu,\InpSp,\IdSp,\mathsf{AuxSp} \big)$.
\end{description}
 Since $\HF$ is an extended HIB-TDF,   the master key generation algorithm of our HIB-TDF
receives an auxiliary input $\mathbf{y} \in \mathsf{AuxSp}$. Here, it is seen as a
concatenation of $d$ row vectors $\yy_1,\ldots,\yy_d \in \Z_p^{\mu}$.
\begin{description}
\item[$\HFMKg(\mathsf{pms},\mathbf{y})$:]
 parse the auxiliary input as $\yy=[\yy_1|\ldots |\yy_d] \in \Z_p^{d \cdot
\mu}$, and proceed as follows. \\ \vspace{-0.3 cm}
\begin{itemize}
\item[1.]
 Choose  $  \alpha_v   \sample \Z_p^*$,     $\valpha_w  \sample (\Z_p^*)^n$,
 and $\valpha_h \sample (\Z_p^*)^{d \times (\mu+1) \times n}$. Define
 $v=g^{\alpha_v}$, $\hat{v}=\hat{g}^{\alpha_v}$,
 $\mathbf{w}=g^{\valpha_w} \in \G^n$ and  $\hat{\mathbf{w}}=\ghat^{\valpha_w} \in
 \Ghat^n$. Likewise, set up vectors
 $\hh =g^{\valpha_h}   \in \G^{d \times (\mu+1) \times n }$ and $\hat{\hh} =\ghat^{\valpha_h}   \in \Ghat^{d \times (\mu+1) \times n }$.  Define
\begin{eqnarray*}
  \PP_{core} &:= & \Bigl( v
 ,~ \{\mathbf{w}[l_1] \}_{l_1=1}^n ,~
    \{ \hh[i_1,i_2,l_1]  \}_{i_1 \in \{1,\ldots,d\}, i_2 \in \{0,\ldots,
  \mu \},~l_1\in \{1,\ldots,n\}}
\Bigr)
\end{eqnarray*}
   \item[2.]
For $i_1=1$ to $d$, parse $ \yy_{i_1}$ as $(\yy_{i_1}[1],\ldots
,\yy_{i_1}[\mu]) \in \Z_p^{\mu}$. For $l_2=1$ to $n$, do the
following. \\ \vspace{-0.2 cm}
\begin{itemize}
\item[a.]
 Choose $\ess[l_2]   \sample \Z_p^*$ and  compute $\JJ[l_2]=v^{\ess[l_2]} $ as well as
\begin{eqnarray*}
 \CCC_w[l_2,l_1]=\ww[l_1]^{\ess[l_2]}, \ \
   \CCC[i_1,i_2,l_2,l_1]= \bigl(
  \hh[i_1,0,l_1]^{ \yy_{i_1}[i_2] \cdot
\Delta (l_2,l_1)} \cdot \hh[i_1,i_2,l_1] \bigr)^{\ess[l_2]}
\end{eqnarray*}
for each $ i_1 \in \{1,\ldots,d \}$, $i_2\in \{1,\ldots,\mu\}$,
$l_1\in \{1,\ldots,n\}$. \\ \vspace{-0.3 cm}
\item[b.] Define a $n \times n$ matrix $\{\CC[l_2,l_1] \}_{l_2,l_1 \in \{1,\ldots,n\}}$ of HPE ciphertexts
\begin{eqnarray}\label{HCT}
\hspace{-0.7 cm} \CC[l_2,l_1]=\big(\JJ[l_2],\CCC_w[l_2,l_1], \{
\CCC[i_1,i_2,l_2,l_1] \}_{i_1 \in \{1,\ldots,d \},~i_2 \in
\{1,\ldots,\mu \} } \big). \quad
\end{eqnarray}
\end{itemize}
The master public key consists of $\mathsf{mpk}:=\big( \PP_{core},
\{\CC[l_2,l_1] \}_{l_2,l_1 \in \{1,\ldots,n\}}  \big)$ while the
master secret key is $\mathsf{msk}:=\big(\vhat ,\hat{\mathbf{w}},
   \hat{\hh}  \big)$. For each $l_1 \in \{1,\ldots,n\}$, it will be convenient to view $(\mathsf{PP}_{core},\mathsf{msk})$ as a vector of HPE
   master key pairs
$(\mathsf{mpk}[l_1],\mathsf{msk}[l_1])$, with
\begin{eqnarray*}
 \mathsf{mpk}[l_1] &=& \big(v,\ww[l_1], \{ \hh[i_1,i_2,l_1]
\}_{i_1 \in \{1,\ldots,d\},i_2 \in \{0,\ldots,\mu\}}  \big) \\
\mathsf{msk}[l_1] &=& \big(\hat{v},\hat{\ww}[l_1], \{
\hat{\hh}[i_1,i_2,l_1] \}_{i_1 \in \{1,\ldots,d\},i_2 \in
\{0,\ldots,\mu\}}  \big).
 \vspace{-0.3 cm}
 \end{eqnarray*}
\end{itemize}

\item[$\HFKg\big(\mathsf{pms},\mathsf{msk},(\id_1,\ldots,\id_{\ell}) \big)$:] to generate a  key for
 an identity   $(\id_1,\ldots,\id_{\ell}) \in \IdSp,$ parse  $\mathsf{msk}$ as $
\big( \vhat,\hat{\mathbf{w}},
   \hat{\hh}  \big)$ and  $\id_{i_1}$ as $\id_{i_1}[1]\ldots \id_{i_1}[\mu]$ for  $i_1=1$ to $\ell$. Choose
$\mathbf{r}_w,\mathbf{r}_1,\ldots,\mathbf{r}_{\ell} \sample
(\Z_p^*)^n$. For each $l_1 \in \{1,\ldots,n\}$, compute the
decryption component
$\mathbf{SK}_D=(\DD,\DD_w,\{\DD_{i_1}\}_{i_1=1}^{\ell})$ of the key
as

\begin{eqnarray}
  \DD[l_1]  =      \prod_{i_1=1}^{\ell} \big(  \prod_{ i_2 =1}^\mu \hat{\hh}[i_1,i_2,l_1]^{\id_{i_1}[i_2]}
\big)^{\err_{i_1}[l_1] }    \cdot
\hat{\mathbf{w}}[l_1]^{\mathbf{r}_w[l_1] }   , \  \label{dec-comp}
\DD_{w }[l_1] = \vhat^{\err_w[l_1]}   ,\   \DD_{i_1}[l_1]=
\vhat^{\err_{i_1}[l_1]}
\end{eqnarray}
and the delegation component  $\mathbf{SK}_{DL}=\bigl(
\{\KK[j,k,l_1] \}_{j,k,l_1},~\{\mathbf{L}[j,l_1]\}_{j,l_1},  \ \{
\mathbf{L}[j,k,i_1,l_1] \}_{j,k,i_1,l_1},$ $\{ \mathbf{L}_w[j,k,l_1]
  \}_{j,k,l_1} \bigr),$
with $j \in \{\ell+1,\ldots,d\},$ $k \in \{1,\ldots,\mu\}$ and $i_1 \in \{1,\ldots,\ell\}$ as
$$
    \KK[j,k,l_1] =  \prod_{i_1=1}^{\ell} \Big( \prod_{i_2=1}^{\mu} \hat{\hh}[i_1,i_2,l_1]^{\id_{i_1}[i_2]} \Big)^{\mathbf{s}[j,k,i_1,l_1]} \cdot \hat{\hh}[j,k,l_1]^{\mathbf{s}'[j,l_1]} \cdot \hat{\mathbf{w}}[l_1]^{\mathbf{s}_{w}[j,k,l_1]}   ,
  $$
$$ \mathbf{L}[j,l_1] =  \vhat^{\mathbf{s}'[j,l_1]},    \quad \quad  \mathbf{L}[j,k,i_1,l_1]  =  \vhat^{\mathbf{s}[j,k,i_1,l_1]}       \quad \mathrm{and}  \  \mathbf{L}_w[j,k,l_1] = \vhat^{\mathbf{s}_w[j,k,l_1]} . $$
Output  $\mathbf{SK}_{(\id_1,\ldots,\id_\ell)}   =\big(
\mathbf{SK}_D,\mathbf{SK}_{DL} \big).$ \smallskip

 \item[$\HFDel\big(\mathsf{pms},\mathsf{mpk},(\id_1,\ldots,\id_\ell),\mathbf{SK}_{(\id_1,\ldots,\id_\ell)},\id_{\ell+1}\big)$:]
parse $\mathbf{SK}_{(\id_1,\ldots,\id_\ell)}$ as a $\mathrm{HF}$
private key of the form $(\mathbf{SK}_D,\mathbf{SK}_{DL})$, and
$\id_{\ell+1}$ as a  string $\id_{\ell+1}[1]\ldots \id_{\ell+1}
[\mu] \in \Sid$.
\\ \vspace{-0.3 cm}
\begin{itemize}
\item[1.] For $l_1=1$ to $n$, 
define $\ell$-th level HPE  keys
$\mathbf{SK}_{(\id_1,\ldots,\id_\ell)}[l_1]=(\mathbf{SK}_D[l_1],\mathbf{SK}_{DL}[l_1])$
where
\begin{eqnarray*}
 \mathbf{SK}_D[l_1]&=&
 (\mathbf{D}[l_1],\mathbf{D}_w[l_1], \{\mathbf{D}_{i_1}[l_1]\}_{i_1=1}^{\ell})
 \\
\mathbf{SK}_{DL}[l_1] &=& \bigl( \{\KK[j,k,l_1]
\}_{j,k},~\{\mathbf{L}[j,l_1]\}_{j},  ~\{ \mathbf{L}[j,k,i_1,l_1]
\}_{j,k,i_1}, \{ \mathbf{L}_w[j,k,l_1]  \}_{j,k} \bigr).
\end{eqnarray*}
\item[2.]
 For $l_1=1$ to $n$, run 
$\mathbf{Delegate}(\mathsf{mpk}[l_1],(\id_1,\ldots,\id_\ell),\mathbf{SK}_{(\id_1,\ldots,\id_\ell)
}[l_1],\id_{\ell+1})$ (as specified in Section \ref{sec:new-HIPE})
to get
$\mathbf{SK}_{(\id_1,\ldots,\id_\ell,\id_{\ell+1})}[l_1]=(\mathbf{SK}_D'[l_1],\mathbf{SK}_{DL}'[l_1])$.
\\ \vspace{-0.3 cm}
\end{itemize}
Finally return
$\{\mathbf{SK}_{(\id_1,\ldots,\id_\ell,\id_{\ell+1})}[l_1]\}_{l_1=1}^{n}$.
\\ \vspace{-0.3 cm}

\item[$\HFEv \big(\mathsf{pms},\mathsf{mpk},(\id_1,\ldots,\id_{\ell}),X
\big)$:]  Given a $n$-bit input $X= x_1 \ldots x_n \in \{0,1\}^n$,
for  $i_1=1$ to $\ell$, parse $\id_{i_1}$ as $\id_{i_1}[1]\ldots
\id_{i_1}[\mu]$. For $l_1=1$ to $n$, do the following.
\\ \vspace{-0.3 cm}
\begin{itemize}
\item[1.] For each $l_2\in \{1,\ldots,n\}$, compute modified HPE ciphertexts  by
defining
\begin{eqnarray*}
   \CCC_{\id}[i_1,l_2, l_1] &=&  \prod_{i_2=1}^{\mu} \CCC[i_1,i_2,l_2,l_1]^{\id_{i_1}[i_2]} =  \Bigl(
  \hh[i_1,0,l_1]^{   \langle
\yy_{i_1},\id_{i_1} \rangle  \cdot \Delta(l_2,l_1) } \cdot
 \prod_{i_2=1}^{\mu}
\hh[i_1,i_2,l_1]^{\id_{i_1}[i_2] } \Bigr)^{\ess[l_2]}
\end{eqnarray*}
for each $ i_1 \in \{1,\ldots,\ell \}$, $l_1,l_2 \in
\{1,\ldots,n\}$. The modified ciphertexts are
\begin{eqnarray}\label{HCT-2}
\CC_{\id}[l_2,l_1]=\big( \JJ[l_2],
\{\CC_{\id}[i_1,l_2,l_1]\}_{i_1=1}^{\ell} \big) \in \G^{\ell+1}
  .
\end{eqnarray}
The resulting $\{\CC_{id}[l_2,l_1]\}_{l_2,l_1 \in \{1,\ldots,n\}}$
thus form a $n \times n$ matrix of anonymous HIBE ciphertexts for
the identity $\id=(\id_1,\ldots,\id_{\ell})$. \\ \vspace{-0.3 cm}
\item[2.] Compute $ C_{\id,v} = \prod_{l_2=1}^{n} \JJ[l_2]^{x_{l_2}}=v^{\langle
\ess ,X \rangle }  $, $
   \CC_{\id,w}[l_1]=\prod_{l_2=1}^n \CCC_w[l_2,l_1]^{x_{l_2}}= \ww[l_1]^{\langle s,X
 \rangle}
$ and
 \begin{eqnarray}
   \CC_{\id}[i_1,l_1] &=& \prod_{l_2=1}^n
\CCC_{\id}[i_1,l_2, l_1]^{x_{l_2}}  \label{ct-hibe}    =
\hh[i_1,0,l_1]^{ \ess[l_1] \cdot x_{l_1} \cdot   \langle
\yy_{i_1},\id_{i_1} \rangle   } \cdot \Bigl( \prod_{i_2=1}^{\mu}
\hh[i_1,i_2,l_1]^{\id_{i_1}[i_2] } \Bigr)^{ \langle \ess, X \rangle
} \qquad
\end{eqnarray}
\begin{align} \label{output}
 \hspace{-1.3 cm} \textrm{Then, output } C=\big(C_{\id,v},\{ \CC_{\id,w}[l_1] \}_{l_1=1}^n,
 \{\CC_{\id}[i_1,l_1]\}_{i_1 \in \{1,\ldots,\ell \}, l_1 \in
\{1,\ldots,n\}}\big) \in \G^{ n+1+ n \times \ell} .
\end{align}
\end{itemize}

\item[$\HFInv \big(\mathsf{pms},\mathsf{mpk},(\id_1,\ldots,\id_\ell),\mathbf{SK}_{(\id_1,\ldots,\id_\ell)},C\big)$:] parse  the
decryption component $\mathbf{SK}_D$ of the private key
   as a tuple of the form $(\DD,\DD_w,\DD_{\barw},\{\DD_{i_1}\}_{i_1=1}^{\ell})$
 and the output $C$ as per
 (\ref{output}). Then, for $l_1=1$ to $n$, set $x_{l_1}=0$  if
 \begin{eqnarray} \label{inversion}
 \hspace{-0.9 cm}    e(C_{\id,v},\mathbf{D}[l_1]) \cdot
e(\CC_{\id,w}[l_1],\mathbf{D}_{w}[l_1])^{-1}   \cdot
\prod_{i_1=1}^{\ell} e(\CC_{\id}[i_1,l_1]
,\mathbf{D}_{i_1}[l_1])^{-1}  =  1_{\G_T}.
\end{eqnarray} Otherwise,
set $x_{l_1}=1$. Eventually, return $X=x_1\ldots x_n \in \{0,1\}^n$.
\end{description}
From (\ref{ct-hibe}), we notice that, with overwhelming probability, if there exists some $i_1 \in \{1,\ldots,d\}$
 such that $\langle \mathbf{y}_{i_1}, \id_{i_1}\rangle \neq 0$,
 relation (\ref{inversion}) is
satisfied if and only if $x_{l_1}=0$. Indeed, in this case, the
output  (\ref{output}) is  distributed as  a vector of $n$
Boneh-Boyen HIBE ciphertexts  (in their anonymous variant considered
 in \cite{Duc10}). These ciphertexts correspond to the same
encryption exponent $\langle \mathbf{s},X \rangle$ and are generated
under $n$ distinct master public keys   sharing the same
component $v \in \G$. \\
\indent  When the function is implemented in injective mode, the
auxiliary input consists of a vector
$\yy^{(0)}=[(1,0,\ldots,0)|\ldots|(1,0,\ldots,0)] \in \Z_p^{d \cdot
\mu}$. Since $\id_{i_1}[1]=1$ for each $i_1 $, this guarantees
injectivity since $\langle \yy_{i_1}^{(0)},\id_{i_1} \rangle \neq 0$
for each $i_1 $. In the partially lossy mode, we   have $\langle
\yy_{i_1} ,\id_{i_1} \rangle = 0$ for each $i_1 \in
\{1,\ldots,\ell\}$ with non-negligible probability, which leads to
high non-injectivity.

\subsection{Security Analysis}
\label{security-analysis}

To analyze the security of the scheme (w.r.t. the new definition, in Section \ref{sec:new_sec_def_HIBTDF}), in both the adaptive and the
selective cases, we define two experiments, $RL_0$ and $RL_n$. In
both of them, $\HFSetup$ is run and the public parameters are given
to the adversary $\A$. Algorithm $\HFMKg$ is run with auxiliary input
$\yy^{(0)}=[(1,0,\ldots,0)|\ldots | (1,0,\ldots,0)]$ in  $RL_0$, and
 auxiliary input $\yy^{(1)}=[\yy_1^{(1)}|\ldots |\yy_d^{(1)}] $
in $RL_n$,  where $\yy^{(1)}$ is produced by an
auxiliary input generator $\mathsf{Aux}(\id)$ taking as input a
special  hierarchical identity $\id=(\id_1,\ldots,\id_{\ell})$. The
master public key $\mathsf{mpk}$ is given to the $\A$. $\A$ can request secret keys for identities $\id$, which will be
answered using $\HFKg$ and $\HFDel$ and will be added to $IS$, initialized to $IS=\{\emptyset\}$. Also, $\A$ will
output a hierarchical identity $\id^{\star}$. Finally, $\A$
will output a guess $d_\A$. Both in $RL_0$ and $RL_n$, the
experiment will halt and output $d'=0$ if: a) for any
$\id=(\id_1,\ldots,\id_{\ell}) \in IS$ we have that $\langle
y_{i_1}^{(1)},\vec \id_{i_1} \rangle =0$ for each $i_1
 \in \{1,\ldots,\ell\}$, or b) if $\langle
y_{i_1}^{(1)},\vec \id^{\star}_{i_1} \rangle \neq 0$ for some $i_1
 \in \{1,\ldots,\ell^{\star}\}$. If the experiment has not aborted, it will output the bit
 $d'=d_\A$.
The result in the following lemma will be used to prove that $\HF$ enjoys partial lossiness.

\begin{lemma} \label{indist-setup}
Under the $\mathcal{P}$-BDH$_1$ and DDH$_2$ assumptions, the
experiments $RL_0$ and $RL_n$ return $1$ with nearly identical
probabilities. Namely, there exist  PPT algorithms $\B_1$ and $\B_2$
such that
$$| \PR[RL_0 \Rightarrow 1] - \PR[RL_n \Rightarrow 1]| \leq n \cdot
 \bigl( (d \cdot \mu +1) \cdot \mathbf{Adv}^{\mathcal{P}\textrm{-}\mathrm{BDH}_1}(\B_1) + q \cdot \mathbf{Adv}^{\mathcal{P}\textrm{-}\mathrm{DDH}_2}(\B_2) \bigr),
$$
where $q$ is the number of ``Reveal-key'' queries made by
$\A$.\textnormal{(The proof is  in Appendix \ref{hybrid-par}).}
\end{lemma}

\subsubsection{Adaptive-id Security.}\label{adaptive-id-sec}
 The details on the security analysis of $\HF$ with respect to
selective adversaries are given in Appendix
\ref{app:selective-security}. For adaptive security we consider our
construction of $\HF$ with a restricted identity space, namely
taking
    $\Sigma_{ID}=\{(1,\mathbf{x}): \mathbf{x} \in \{0,1\}^{\mu-1} \}$. Define a sibling
$\LHF$ with $\AxSp=\Z_p^{\mu \cdot d}$, where the auxiliary input
$\yy^{(1)}=[\yy_1^{(1)}|\ldots |\yy_d^{(1)}]$, for any $i_1$ from
$1$ to $d$, is defined as
\begin{eqnarray*}
  y_{i_1}' & \sample &  \{0,\dots,2q-1\}, \  \xi_{i_1}
\sample\{0,\dots,\mu+1\}, \ \yy_{i_1}^{(1)}[1]=y_{i_1}'-2\xi_{i_1}q,
\ \{\yy_{i_1}^{(1)}[i] \sample \{0,\dots,2q-1\}\}_{i=2}^{\mu}.
\end{eqnarray*}

The pre-output stage $\mathcal{P}$ associated to the scheme $\HF$ in the
adaptive case will be the artificial abort used in \cite{Wat05}: if $IS$ is
the set of queried identities (only taking into account
``Reveal key'' queries) and $\id^\star$ is the challenge identity, let
$E(IS,\id^\star)$ be the event that the $\REAL$ or the $\LOSSY$
experiments set $d_1$ to be 1, and let
$\eta(IS,\id^\star)=\Pr[E(IS,\id^\star)]$. We prove in Appendix
\ref{proof-lower-bound}:

\begin{lemma}\label{lemma-lower-bound}
$\eta_{low}=1/\left(2\cdot(2 q\mu)^{d}\right)\leq
\eta(IS,\id^\star)$
\end{lemma}

$\mathcal{P}$ computes an approximation $\eta'(IS,\id^\star)$
of $\eta(IS,id^\star)$, by using
$O(\zeta^{-2}\ln(\zeta^{-1})\eta_{low}^{-1}\ln(\eta_{low}^{-1}))$
samples, for a non-negligible $\zeta$, which is part of the input to
$\mathcal{P}$. If $\eta'(IS,\id^\star) \leq \eta_{low}$, then $d_2$ is always set to 1.
If $\eta'(IS,\id^\star)>\eta_{low}$, then
$d_2=0$ with probability
$1-\eta_{low}/\eta'(IS,\id^\star)$ and $d_2=1$ with probability
$\eta_{low}/\eta'(IS,\id^\star)$.
Adaptive security of $\HF$ is established in the following theorem.

\begin{theorem}\label{theor-adapt-loss}
Let $n>\log p$ and let $\omega=n-\log p$. Let $\HF$ be the HIB-TDF
with parameters $n,d,\mu, \Sigma_{ID}=\{(1,\mathbf{x}): \mathbf{x} \in
\{0,1\}^{\mu-1} \}, \IdSp=\Sigma_{ID}^{(\leq d)}$. Let $\LHF$ be the
sibling with auxiliary input as above, let the pre-output stage
associated to $\HF$ be the one specified above. In front of
adaptive-id adversaries that make a maximal number of $q\leq
p/(2\mu)$ queries, $\HF$ is $(\omega,\delta)$-partially lossy, for
$\delta = 13/\left(32\cdot(2 q\mu)^{d}\right)$.

Specifically regarding condition (i), for any such adaptive-id adversary $\A$ there exist algorithms $\B_1$ and $\B_2$ such that
$\mathbf{Adv}_{\HF,\LHF,\omega,\zeta}^{\mathrm{lossy}}(\A)\leq 2n \cdot
\bigl( (d \cdot \mu +1) \cdot \mathbf{Adv}^{\mathcal{P}\textrm{-}\mathrm{BDH}_1}(\B_1) + q \cdot \mathbf{Adv}^{\mathcal{P}\textrm{-}\mathrm{DDH}_2}(\B_2) \bigr)
$,
where the running time of $\B_1$ and $\B_2$ is that of $\A$ plus a
$O(\zeta^{-2}\ln(\zeta^{-1})\eta_{low}^{-1}\ln(\eta_{low}^{-1}))$
overhead, with $\eta_{low}=1/\left(2\cdot(2 q\mu)^{d}\right)$.
 \textnormal{(The proof is given in Appendix \ref{proof-adaptive-loss}).}
    \end{theorem}

The fact that $\eta_{low}^{-1}$ depends exponentially on $d$ is what
makes our adaptive-id security result valid only for hierarchies of
constant depth $d$. Whether this restriction can be avoided at all
is an interesting open question.

%
%
%

\appendix

\section{Deferred Definitions and Proofs for Hierarchical Predicate Encryption}

\subsection{Definitions for Hierarchical Predicate Encryption}\label{app:syntax-HIPE}

A tuple of integers $\vec{\nu}=(\mu_1,d;\nu_1,\ldots,\nu_d)$ such
that $\nu_0 =0 \leq \nu_1 < \nu_2 < \cdots,< \nu_d =\mu_1$  is
called a {\it format hierarchy} of depth $d$. Such a hierarchy is
associated with attribute spaces $\{ \Sigma_{i}  \}_{i=0}^{d}$. In
our setting, we set $\Sigma_{i }=\Z_p^{\nu_i - \nu_{i-1}} \backslash
\{ \vec{0} \}$ for $i=1$ to $d$, for some $p \in \mathbb{N}$, and
also define the universe of hierarchical attributes as $\Sigma :=
\cup_{i=1}^d (\Sigma_1 \times \cdots \times \Sigma_i )$. For vectors
$\{\vec{X}_i \in \Sigma_i \}_{  i \in \{1,\ldots,d\} }$, we will
consider the (inner-product) hierarchical predicate
$f_{(\vec{X}_1,\ldots,\vec{X}_{\ell})}(\vec{Y}_1,\ldots,\vec{Y}_\kappa)=1$
iff $\ell \leq \kappa $ and $\vec{X}_i \cdot \vec{Y}_i =0$ for all
$i \in \{1,\ldots,\ell\}$. The space of hierarchical predicates is
defined to be $$\mathcal{F}=\{ f_{(\vec{X}_1,\ldots,\vec{X}_\ell)} |
\{\vec{X}_i \in \Sigma_i \}_{i \in \{1,\ldots,\ell \} } \}$$ and the
integer
$\kappa $ (resp. $\ell$) is called  the depth (resp. the level)  of $(\vec{Y}_1,\ldots,\vec{Y}_\kappa)$ (resp. $(\vec{X}_1,\ldots,\vec{X}_\ell)$).
Most of the time we will assume that $\nu_i/\nu_{i-1}=\mu$ for any $i$, and to specify the format hierarchy we will only have to give $\mu$ and $d$. \\
\indent Let $\vec{\nu}=(\mu_1,d;\nu_1,\ldots,\nu_d)$ be a format
hierarchy. A hierarchical predicate encryption (HPE) scheme for a
predicate family $\mathcal{F}$ consists of these  algorithms.
\begin{description}
\item[Setup$(\varrho,\vec{\nu})$:] takes as input a security parameter
$\varrho \in \mathbb{N}$ and a format hierarchy
$\vec{\nu}=(n_1,d;\nu_1,\ldots,\nu_d)$. It outputs a master secret
key $\mathsf{msk}$ and a master public key $\mathsf{mpk}$ that
includes the description of a hierarchical attribute space $\Sigma$.
\item[Keygen$\big(\mathsf{msk},(\vec{X}_1,\ldots,\vec{X}_{\ell}) \big)$:] takes as input
predicate vectors  $(\vec{X}_1,\ldots,\vec{X}_\ell) \in  {\Sigma}_1
\times \cdots \times \Sigma_{\ell}$ and the master secret key
$\mathsf{msk}$. It outputs a private key
$SK_{(\vec{X}_1,\ldots,\vec{X}_{\ell})}$.
\item[Encrypt$\big(\mathsf{mpk},(\vec{Y}_1,\ldots,\vec{Y}_\kappa),M \big):$] takes as input    attribute vectors
$(\vec{Y}_1,\ldots,\vec{Y}_\kappa) \in  {\Sigma}_1 \times \cdots
\times \Sigma_{k}$,  the master public key $\mathsf{mpk}$ and a
message $M$. It outputs a ciphertext $C$.
\item[Decrypt$\big(\mathsf{mpk},(\vec{X}_1,\ldots,\vec{X}_{\ell}), SK_{(\vec{X}_1,\ldots,\vec{X}_{\ell})} ,C \big)$:] takes
in a private key $SK_{(\vec{X}_1,\ldots,\vec{X}_{\ell})}$ for the
 vectors $(\vec{X}_1,\ldots,\vec{X}_{\ell})$,
 the master
 public key $\mathsf{mpk}$ and a
ciphertext $C$. It outputs   a plaintext $M$ or $\perp$.
\item[Delegate$\big(\mathsf{mpk},(\vec{X}_1,\ldots,\vec{X}_{\ell}), SK_{(\vec{X}_1,\ldots,\vec{X}_{\ell})},
\vec{X}_{\ell+1} \big)$:] takes in a $\ell$-th level private key
$SK_{(\vec{X}_1,\ldots,\vec{X}_{\ell})}$, the corresponding vectors
$(\vec{X}_1,\ldots,\vec{X}_{\ell})$ and a vector $\vec{X}_{\ell+1}$.
It outputs a $(\ell+1)$-th level private key $SK_{
(\vec{X}_1,\ldots,\vec{X}_{\ell},\vec{X}_{\ell+1}) }$.
\end{description}
Correctness   mandates that, for any message $M$ and hierarchical
vectors $(\vec{X}_1,\ldots,\vec{X}_\ell)$,
$(\vec{Y}_1,\ldots,\vec{Y}_\kappa) $, we have
  $\mathbf{Decrypt}\big(\mathsf{mpk},(\vec{X}_1,\ldots,\vec{X}_\ell),SK_{
(\vec{X}_1,\ldots,\vec{X}_\ell)},
\mathbf{Encrypt}(\mathsf{mpk},(\vec{Y}_1,\ldots,\vec{Y}_\kappa,M)
\big)=M$  whenever the condition
$f_{(\vec{X}_1,\ldots,\vec{X}_{\ell})}(\vec{Y}_1,\ldots,\vec{Y}_\kappa)=1$ is satisfied.\\
\indent Following \cite{OT09}, we write $f' \leq f$ to express that
the predicate vector of $f$ is a prefix of that for $f'$, meaning
that $f'$ is at least as constraining as $f$.

\begin{definition} \label{sec-def}
A  Hierarchical Predicate Encryption scheme is {\bf
selectively weakly attribute-hiding} if no PPT adversary has
non-negligible advantage in the following game:
\begin{itemize}
\item[1.] The adversary $\A$ chooses a format hierarchy $\vec{\nu}=(n_1,d;\nu_1,\ldots,\nu_d)$ and
    vectors $(\vec{Y}_1^0,\ldots,\vec{Y}_{d^\star}^0)$,
$(\vec{Y}_1^1,\ldots,\vec{Y}_{d^\star}^1) $, for some $d^\star \leq
d$. The challenger generates a master  key pair
$(\mathsf{msk},\mathsf{mpk}) \leftarrow
\mathbf{Setup}(\varrho,\vec{\nu})$ and $\mathsf{mpk}$ is given to
$\A$.
\item[2.]   $\A$ is allowed to make a number of adaptive    queries.
\\ \vspace{-0.3 cm}
\begin{itemize}
\item[-] {\bf Create-key}: $\A$ provides a predicate $f \in
\mathcal{F}$ and the challenger creates a private key $SK_f$ for $f$
{\it without} revealing it to $\A$.
\item[-] {\bf Create-delegated-key}: $\A$ chooses a private key that was
previously created for some predicate $f$ and also specifies another
predicate $f' \leq f$ that $f$ is a prefix of. The challenger then
computes a delegated key $SK_{f'}$ for $f'$ without revealing it to
$\A$.
\item[-] {\bf Reveal-key}: $\A$ asks the challenger to give out a
previously created key. \\ \vspace{-0.3 cm}
\end{itemize}  For each
   Reveal-key  query $(\vec{X}_1,\ldots,\vec{X}_{\ell})$, it is required that
 $$
f_{(\vec{X}_1,\ldots,\vec{X}_{\ell})}(\vec{Y}_1^0,\ldots,\vec{Y}_{d^\star}^0)=
f_{(\vec{X}_1,\ldots,\vec{X}_{\ell})}(\vec{Y}_1^1,\ldots,\vec{Y}_{d^\star}^1)=0.$$
\item[3.]   $\A$ outputs   messages
$M_0,M_1$.
  Then, the challenger chooses $\beta
\sample \{0,1\}$ and computes a challenge ciphertext
$C^\star=\mathbf{Encrypt}\big(\mathsf{mpk},(\vec{Y}_1^{\beta},\ldots,\vec{Y}_{d^\star}^{\beta}),M_{\beta}
\big)$, which is sent   to $\A$.
\item[4.] $\A$ makes further private key queries for hierarchical vectors $(\vec{X}_1,\ldots,\vec{X}_\ell)$ under the same restriction as above.
\item[5.] $\A$ outputs a bit $\beta' \in \{0,1\}$ and wins if
$\beta'=\beta$.
\end{itemize}
  $\A$'s advantage is quantified as the distance
$\mathbf{Adv}(\A)=|\Pr[\beta'=\beta]-1/2|$.
\end{definition}

\subsection{Correctness of the HIPE scheme} \label{proof-correctness}

\begin{lemma} \label{thm-correct-HIPE}
The HIPE scheme of Section \ref{sec:new-HIPE} is correct. Namely,
for any message $M$ and any vectors $\vec{X}_1,\ldots,\vec{X}_\ell$,
$\vec{Y}_1,\ldots,\vec{Y}_\kappa $, we have
 $$\mathbf{Decrypt}\big(\mathsf{mpk},(\vec{X}_1,\ldots,\vec{X}_\ell),SK_{
(\vec{X}_1,\ldots,\vec{X}_\ell)},
\mathbf{Encrypt}(\mathsf{mpk},(\vec{Y}_1,\ldots,\vec{Y}_\kappa),M
\big)=M$$  whenever
$f_{(\vec{X}_1,\ldots,\vec{X}_{\ell})}(\vec{Y}_1,\ldots,\vec{Y}_\kappa)=1$.
This fact does not depend on whether the key $SK_{
(\vec{X}_1,\ldots,\vec{X}_\ell)}$ was created using
$\mathbf{Delegate}$ or $\mathbf{Keygen}$.
\end{lemma}

\begin{proof}

Let $SK_{(\vec{X}_1,\ldots,\vec{X}_\ell)}=\bigl(
 SK_D,SK_{DL}\big)$ be obtained by running   $\mathbf{Keygen}\big(\mathsf{msk},(\vec{X}_1,\ldots,\vec{X}_{\ell}) \big)$, where each attribute vector is
 $\vec{X}_{i_1}=(x_{i_1,1},\ldots,x_{i_1,\mu}) \in \Z_p^{\mu}$, for each $i_1 \in \{1,\ldots,\ell \}$.
  Let us write the decryption component of the key as $SK_D=(D,D_w,  \{D_{i_1}\}_{i_1=1}^\ell)$.

 Let $C=\bigl(C_0,C_v,C_{w}, \{C_{i_1,i_2}  \}_{i_1 \in
\{1,\ldots,\kappa \},~i_2 \in \{1,\ldots,\mu \} }  \bigr)$    be the
output of the encryption algorithm
$\mathbf{Encrypt}\big(\mathsf{mpk},(\vec{Y}_1,\ldots,\vec{Y}_\kappa),M
\big)$, for vectors $ \vec{Y}_1=(y_{1,1},\ldots,y_{1,\mu})$, \ldots,
$\vec{Y}_\kappa=(y_{\kappa,1},\ldots,y_{\kappa,\mu})$.

Since
$f_{(\vec{X}_1,\ldots,\vec{X}_{\ell})}(\vec{Y}_1,\ldots,\vec{Y}_\kappa)=1$,
and by the definition of hierarchical inner-product predicates, we
know that $\ell \leq \kappa $ and $\vec{X}_i \cdot \vec{Y}_i =0$ for
all $i \in \{1,\ldots,\ell\}$. Therefore, when the decryption
protocol
$\mathbf{Decrypt}\big(\mathsf{mpk},(\vec{X}_1,\ldots,\vec{X}_\ell),SK_{(\vec{X}_1,\ldots,\vec{X}_\ell)},C\big)$
computes $C_{i_1}$, for each $i_1 \in \{1,\ldots,\ell \}$, the
obtained value equals $C_{i_1} = \big(  \prod_{i_2=1}^{\mu}
h_{i_1,i_2}^{x_{i_1,i_2}} \big)^s $. The decryption protocol
computes then the pairing of this element $C_{i_1}$ with $D_{i_1} =
\hat{v}^{r_{i_1}}$. Multiplying all these pairings, for indices $i_1
\in \{1,\ldots,\ell \}$, one obtains $e\big( v^s\ ,\
\prod_{i_1=1}^{\ell} \big(   \prod_{ i_2 =1}^{\mu}
\hat{h}_{i_1,i_2}^{x_{i_1,i_2}} \big)^{r_{i_1} } \big)$. This value
is canceled out with one of the factors of $e(C_v,D)$, when
computing the final decryption operation $C_0 \cdot  e(C_v,D)^{-1}
\cdot e(C_{w},D_{w})  \cdot  \prod_{i_1=1}^{\ell}
e(C_{i_1},D_{i_1})$. The other two factors of $e(C_v,D)$ are
$e(v^s,\hat{g}^{\alpha})$ and $e(v^s,\hat{w}^{r_{w}})$, which cancel
out the factor $e(g,\hat{v})^{\alpha s}$, contained in $C_0= M \cdot
e(g,\hat{v})^{\alpha s}$, and the factor $e(C_{w},D_{w}) =
e(w^s,\hat{v}^{r_w})$, respectively. Therefore, the final
computation of the decryption protocol results in the plaintext $M$
contained in $C_0$.

In this way, we have proved that the encryption and decryption
protocols work correctly when the original secret keys (resulting
from Keygen) are used. The fact that the decryption protocol works
fine also with delegated secret keys (resulting from Delegate) is a
consequence of Lemma \ref{lemma-original-delegated}, inside the
proof of Theorem \ref{thm-HIPE}. If a delegated secret key could
lead to an incorrect decryption, then this fact could be used to
distinguish original secret keys from delegated ones, which would
contradict the statement of Lemma
\ref{lemma-original-delegated}.\end{proof}

\subsection{Proof of Theorem \ref{thm-HIPE}} \label{proof-HPE}

The proof considers a sequence of games starting with the real game
and ending with a game where the adversary has no advantage and wins
with probability exactly $1/2$. \\ \indent
   For each $i$, we denote by $S_i$ the event that the
adversary wins in Game$_i$. In the whole sequence of games, we call
$d^\star$ the depth of the challenge hierarchical vectors
$(\vec{Y}_1^0,\ldots,\vec{Y}_{d^\star}^0)$ and $(\vec{Y}_1^1,\ldots,
\vec{Y}_{d^\star}^1)$ and
$$C^\star=\big(C_0^\star,C_v^\star,C_w^\star,\{C_{i_1,i_2}^\star
\}_{i_1 \in \{1,\ldots,d^\star \}, i_2 \in \{1,\ldots,\mu\}}
 \big)$$
  denotes the challenge ciphertext.
\begin{description}
\item[Game$_{0}$:] is the real attack game at the end of which
the challenger outputs $1$ in the event, called $S_0$, that   the
adversary $\A$ manages to output $\beta' \in \{0,1\}$ such that
$\beta'=\beta$, where $\beta \in \{0,1\}$ is the challenger's hidden
bit in the challenge phase. If   $\beta'\neq \beta$, the challenger
outputs $0$.
\item[Game$_{1}$:] is identical to Game$_{0}$ with the difference that
the challenger always answers private key queries by returning fresh
private keys ({\it i.e.}, keys produced by $\mathbf{Keygen}$)
instead of deriving those keys using the delegation algorithm.
\item[Game$_2$:] is like Game$_{1}$ but the challenge ciphertext $C^\star$ is
now an encryption under
$(\vec{Y}_1^{\beta},\ldots,\vec{Y}_{d^\star}^{\beta})$ of a random
plaintext $M \sample \G_T$, which is chosen independently of $M_0$
and $M_1$.
\item[Game$_{3}$:] is identical to Game$_2$ with the difference that,
in the challenge ciphertext, $C_w^\star$ is replaced by a random
group element chosen uniformly and independently in $\G$.
\item[Game$_{4,i,j}$ $(1 \leq i \leq d^\star,~ 1\leq j \leq \mu)$:] is identical to Game $3$
with the difference that, in the challenge ciphertext,
$C_{i_1,i_2}^\star$ are replaced by random elements of $\G $ if $i_1
<i$ or $(i=i_1) \wedge (i_2 \leq j)$. Other group elements ({\it
i.e.}, for which $i>i_1$ or $(i=i_1) \wedge (i_2 >j)$)  are still
computed as in a normal challenge ciphertext.
\end{description}
In Game$_{4,d^\star,\mu}$, it is easy to see that the adversary $\A$
cannot guess $\beta \in \{0,1\}$ with higher probability than
$\Pr[S_{4,d^\star,\mu}]=1/2$ since the challenge ciphertext
$C^\star$ is completely independent of $\beta \in \{0,1\}$. \qed

\begin{lemma}\label{lemma-original-delegated}
Game$_0$ and Game$_1$ are computationally indistinguishable if the
DDH$_2$  assumption  holds in $(\G, \Ghat ) $.
\end{lemma}
\begin{proof}
The lemma will be proved by a hybrid argument. We define
Game$_{0,i}$ for all $0\leq i\leq q$. Game$_{0,i}$ differs from
Game$_0$ in the fact that, when the adversary issues the first $i$
delegation queries, instead of generating the delegated keys
faithfully using the Delegate algorithm, the challenger calls the
Keygen algorithm to generate these delegated keys. For all the
remaining queries, the challenger computes keys and responds
faithfully as in Game$_0$. Under the above definition, Game$_{0,0}$
is the same as Game$_0$ and Game$_{0,q}$ is the same as Game$_1$. We
will prove that Game$_{0,\kappa}$ is indistinguishable from
Game$_{0,\kappa+1}$ for all $0\leq \kappa\leq q-1$. To this end, we
will proceed similarly to \cite{ShiWa08} and rely on a generalized
version (called GDDH hereafter) of the DDH problem in $\Ghat$.

Given a group generator $\textsf{GG}$, define the following
distribution $P(\varrho)$:

\begin{align*}
&\big(p,(\G,\Ghat,\G_T),e \big)\xleftarrow{R} \textsf{GG} (\varrho,1),\\
&g \xleftarrow{R}\G,\ \ \ghat \xleftarrow{R}\Ghat, \\
&\hat{h}_1,\hat{h_2},\dots,\hat{h_\ell} \xleftarrow{R} \Ghat \\
&\tau \xleftarrow{R}\Z_{p}\\
&X\gets \bigl( \big(p,(\G,\Ghat,\G_T),e \big),\ g,\ghat,\hat{h_1},\hat{h_2},\dots,\hat{h_\ell} \bigr)\\
&Q\gets \big(\hat{h_1^\tau} ,\hat{h_2^\tau } , \dots,
\hat{h_\ell^\tau} \big), \\  &\textrm{Output } (X,Q)
\end{align*}

For an algorithm $\A$, define $\A$'s advantage in solving the above
problem:

$$\ell \textrm{-}\textsf{GDDH Adv}_{\textsf{GG},\A}(\varrho):=\left| \Pr \left[\A(X,Q)=1\right]-\Pr\left[\A(X,R)\right]
\right|$$ where $(X,Q)\gets P(\varrho)$ and $R \gets \Ghat^{\ell}$.
It is immediate\footnote{The straightforward reduction computes a
GDDH instance from a DDH$_2$ instance
$\big(g,\ghat,\ghat^a,\ghat^b,\hat{\eta}\checkeq \ghat^{ab} \big)$
by setting $\hat{h_i}=\ghat^{\alpha_i} \cdot (\ghat^b)^{\beta_i}$
and $Q_i=(\ghat^a)^{\alpha_i} \cdot \hat{\eta}^{\beta_i}$ for $i=1$
to $\ell$ with $\alpha_1,\ldots,\alpha_\ell \sample \Z_p$,
$\beta_1,\ldots,\beta_{\ell} \sample \Z_p$. If
$\hat{\eta}=\ghat^{ab}$, we have
$Q=(\hat{h}_1^a,\ldots,\hat{h}_{\ell}^a)$ whereas, if $\eta \in_R
\Ghat$, $Q$ is a random vector of $\Ghat^\ell$.} that GDDH is not
easier than DDH$_2$ and  that the latter advantage function is
negligible if the DDH$_2$ assumption holds in $(\G,\Ghat)$.

To prove that Game$_{0,\kappa}$ is indistinguishable from
Game$_{0,\kappa+1}$ we will use another hybrid argument. We define
Game$_{0,\kappa}'$, which differs from Game$_{0,\kappa}$ in that,
for the $(\kappa+1)$-th delegation query,  $\widehat{SK}_{DL}$ is
the delegation component of a fresh key, instead of a delegation
component obtained by raising every element in $SK_{DL}$ to the same
random power $z \in_R \Z_p$. We show that a PPT adversary cannot
distinguish between the two games.

\indent  We also define Game$_{0,\kappa}''$, which differs from
Game$_{0,\kappa}'$ in that, for the $(\kappa+1)$-th delegation
query, instead of re-randomizing the components of the partial
decryption key with the same exponent $\tau_{j.k}$, $ \{
L_{\ell+1,i_1}^{(j,k)} \}_{i_1 =1}^{\ell+1}, L_{w,\ell+1}^{(j,k)}$
are randomized with different independently chosen exponents, while
$K_{j,k}'$ is chosen in such a way that the resulting key is still
valid. We also prove that no PPT adversary can notice the
difference.

\indent We will argue that  Game$_{0,\kappa}''=$Game$_{0,\kappa+1}$.
Indeed, in the first step, we change $\widehat{SK}_{DL}$ so that, in
step 2 of the Delegate algorithm, we obtain a randomized decryption
key (except for the $g^\alpha$ term). When multiplied by $SK_D$, it
gives a randomized decryption key for $(\vec X_1,\dots, \vec
X_{\ell+1})$. On the other hand, in step 2 of the hybrid proof we
change the partial decryption keys so that they also are randomized
keys except the $g^\alpha$ term.
\begin{claim}
Game$_{0,\kappa}$ is computationally indistinguishable from
Game$_{0,\kappa}'$.
\end{claim}
\begin{proof}
Let $q_0$ denote the maximum number of secret key queries (taking
into account both the ``Create-key'' and ``Create-delegated-key''
queries) made by the adversary. We build a simulator $\B$ that uses
$\A$ to break the following $(q_0 d \mu (d+1))$-GDDH assumption.
\begin{align*}
&(p,\G,\Ghat,\G_T,e)\xleftarrow{R} \textsf{GG} (\varrho,1),\\
&g  \xleftarrow{R}\G,\ \ \ghat  \xleftarrow{R}\Ghat \\
& \{\hat{v}_{i,i_1,i_2,k}\gets \Ghat\}_{i\in \{1,\ldots, q_0 \},
i_1\in \{1,\ldots, d \}, i_2\in \{ 0, \ldots, \mu \}, k\in
\{1,\ldots, d+1 \}}\\
&\tau \xleftarrow{R}\Z_{p}\\
&X\gets \big((p,\G,\Ghat,\G_T,e),\ g,\ghat,\ \
\{\hat{v}_{i,i_1,i_2,k}\}_{i\in \{1,\ldots, q_0 \}, i_1\in
\{1,\ldots, d \}, i_2\in \{1, \ldots, \mu \}, k\in
\{1,\ldots, d+1 \}} \big)\\
&Q\gets (\{\hat{v}_{i,i_1,i_2,k}^\tau\}_{i\in \{1,\ldots, q_0 \},
i_1\in \{1,\ldots, d \}, i_2\in \{ 1, \ldots, \mu \}, k\in
\{1,\ldots, d+1 \}})
\end{align*}

Then, the challenger randomly decides to give $(X,Q'=Q)$ or
$(X,Q'=R)$, where $R$ is a random vector of elements in $\Ghat$ of
the size of $Q$. The simulator will use $\A$ as a subroutine to
break the above problem. \\ \vspace{-0.3 cm}

\noindent \emph{Init and Setup}. At the beginning of the security
game, the adversary commits to two hierarchical vectors
$(\vec{Y_1}^0,\ldots,\vec{Y}_{d^\star}^0)$ and
$(\vec{Y_1}^1,\ldots,\vec{Y}_{d^\star}^1)$, where $d^\star \leq d$.
The simulator chooses the public and the secret key as usual
according to the Setup algorithm. Let $c=\log_v(w)$ and
$a_{i_1,i_2}=\log_v(h_{i_1,i_2})$. Note that these values are known
to the simulator since they are easily computable from
$\mathsf{msk}$.

\vspace{0.3cm}

\noindent \emph{Secret key queries}. We distinguish three cases:
\begin{itemize}
\item [$\bullet$] When a ``create-key'' query or one of the first $\kappa$ delegated secret key queries is made,
the simulator computes and saves a private key, which is given to
$\A$ when a ``reveal-key'' query is made. To compute this secret
key, the simulator uses the elements from the GDDH instance, in such
a way that the exponents are distributed at random. In particular,
if it is the $i$-th query, the simulator defines the components of
the decryption component of the key as:
\begin{eqnarray*}
D &=& g^\alpha\displaystyle\prod_{i_1=1}^\ell
D_{i_1}^{(\sum_{i_2=1}^\mu a_{i_1,i_2}x_{i_1,i_2})} \cdot D_w^c,
\qquad D_w= \hat{v}_{i,0,1,d+1}, \qquad \\
D_{i_1} &=& \hat{v}_{i,0,1,i_1} \ \  i_1\in \{1,\ldots,\ell\}.
\end{eqnarray*}
\noindent For the delegation component of the key, for all $j\in
\{\ell+1\ldots d\}$, all $k\in \{1,\ldots, \mu\}$, the simulator
lets:
\begin{align*}
L_j=\hat{v}_{i,j,1,\ell+1}, \qquad L_{j,k,i_1}=\hat{v}_{i,j,k,i_1}\
\ i_1\in \{1,\ldots, \ell\}, \qquad  L_{w,j,k}=\hat{v}_{i,j,k,d+1}.
\end{align*}
\noindent As the simulator knows the discrete logarithms
$c=\log_{\hat{v}}(\hat{w})$ and
$a_{i_1,i_2}=\log_{\hat{v}}(\hat{h}_{i_1,i_2})$, for each $j \in
\{\ell+1,\ldots d\}$ and all $k \in \{1,\ldots,\mu\}$, it can
compute the remaining components of the key as follows:
\begin{eqnarray}\label{expr:kjk}
K_{j,k}=\displaystyle\prod_{i_1=1}^\ell L_{j,k,i_1}^{(\sum_{i_2=1}^k
a_{i_1,i_2} \cdot x_{i_1,i_2})}  \cdot  L_j^{a_{j,k}} \cdot
L_{w,j,k}^c.
\end{eqnarray}

\item [$\bullet$] When the adversary makes the $(\kappa+1)$-{th} delegation query, it specifies a parent key and asks to fix the
level of the hierarchy to some vector $\vec X_{\ell+1}$. In
particular, assume that the parent key was created in the $i$-th
query. When performing Step 1 of the Delegate algorithm, for all
$j\in \{\ell+2,\ldots d\},\ k\in \{1,\ldots, \mu \}$, the simulator
sets
\begin{eqnarray*}
\widehat L_j &=& Q_{i,j,1,\ell+1}', \qquad \\
 \widehat L_{j,k,i_1} &=& Q_{i,j,k,i_1}' \qquad \qquad  \ i_1\in \{1,\ldots, \ell\},\qquad \\
 \widehat L_{w,j,k} &=& Q_{i,j,k,d+1}'
\end{eqnarray*}
and computes $K_{j,k}$, for  each $j \in \{\ell+2,\ldots d\}$,  $k
\in \{1,\ldots,\mu\}$ exactly in the same way as in expression
(\ref{expr:kjk}).
\item [$\bullet$] For all the remaining queries, the simulator responds
faithfully as in the real game.
\end{itemize}

\noindent Clearly, if $Q'=Q$ in the GDDH instance, then the above
simulation is identical to  Game$_{0,\kappa}$. Otherwise, it is
identical to  Game$_{0,\kappa}'$ since, in the the $(\kappa+1)$-{th}
delegation query, $Q=R$ implicitly defines a set of fresh random
values for $s_j,s_{j,k,i_1}, s_{w,j,k}$, for the appropriate values
of $j,k,i_1$.

\vspace{0.3cm}

\noindent \emph{Challenge} The simulator generates the challenge
ciphertext as normal.

\vspace{0.3cm}

 \noindent \emph{Guess} If the adversary has a
difference of $\epsilon$ in its advantage in Game$_{0,\kappa}$ and
Game$_{0,\kappa}'$, the simulator has a comparable advantage in
solving the GDDH instance. \hfill $\Box$
\end{proof}

\begin{claim}
Game$_{0,\kappa}'$ is computationally indistinguishable from
Game$_{0,\kappa}''$.
\end{claim}
\begin{proof}
To prove this claim, we will appeal to a nested hybrid argument. Let
Game$_{0,\kappa,0,0}'=$Game$_{0,\kappa}'$, and for $1\leq \eta\leq
(d-\ell-1)$, $1\leq \nu\leq \mu$, define Game$_{0,\kappa,\eta,\nu}'$
as the game that differs from Game$_{0,\kappa}'$ in the following:
in the step of the delegation algorithm where the components
$\{L_{\ell+1,i_1}^{(j,k)}\}_{i_1\in \{1,\ldots,\ell+1\}},
L_{w,\ell+1}^{(j,k)}$ are created by re-randomizing the partial
decryption key with some exponent $\tau_{j,k}$, we re-randomize
instead each of these components with a different exponent chosen
uniformly and independently at random whenever $(j,k)\leq
(\eta+\ell+1,\nu)$ (in lexicographic order). Observe that, by
definition, Game$_{0,\kappa,d-\ell-1,\mu}'=$Game$_{0,\kappa}''.$
\noindent   We will show that an adversary cannot distinguish
between one game and the next. That is, if we define and
Game$_{0,\kappa,\eta,\mu+1}'=$Game$_{0,\kappa,\eta+1,0}'$ to
simplify the notation, what we we will show is that no polynomial
time adversary can distinguish between Game$_{0,\kappa,\eta,\nu}'$
and Game$_{0,\kappa,\eta,\nu+1}'$.

\noindent The simulator tries to solve the following
$(\mu(d+1))$-GDDH instance.
\begin{align*}
&(p,\G,\Ghat,\G_T,e)\xleftarrow{R} \textsf{GG} (\varrho,1),\\
&g\xleftarrow{R}\G,\ \ghat \xleftarrow{R}\Ghat\\
& \{ \hat{v}_{i,k}\gets \Ghat\}_{i\in \{1,\ldots, \mu\}, k\in
\{1,\ldots,d+1\}}\\
&\tau \xleftarrow{R}\Z_{p}\\
&X\gets ((n,\G,\Ghat,\G_T,e),\ g,\ghat,\ \ \{ \hat{v}_{i,k}\}_{i\in
\{1,\ldots, \mu\}, k\in
\{1,\ldots,d+1\}})\\
&Q\gets (\{\hat{v}_{i,k}^\tau \}_{i\in \{1,\ldots, \mu\}, k\in
\{1,\ldots,d+1\}})
\end{align*}
The simulator tries to distinguish between $(X,Q'=Q)$ and $(X,Q'=R)$
where $R$ is a random vector from $\Ghat$. The simulator uses as a
subroutine an adversary $\A$ who can distinguish between
$\textsf{Game}'_{0,\kappa,\eta,\nu}$ and
$\textsf{Game}_{0,\kappa,\eta,\nu+1}'$. \vspace{0.3cm}

The simulator runs the Setup algorithm as usual in such a way that
it knows the discrete logarithms
$c=\log_v(w)=\log_{\hat{v}}(\hat{w})$ and $a_{i_1
i_2}=\log_v(h_{i_1,i_2})=\log_{\hat{v}}(\hat{h}_{i_1,i_2})$ for any
$1\leq i_1 \leq d$, $0\leq i_2\leq \mu$.

\noindent To answer secret key queries, the simulator proceeds as
follows:
\begin{itemize}
\item [$\bullet$] For the first $\kappa$-th
delegation queries and all of the ``create-key'' queries, the
simulator computes the keys freshly at random.
\item [$\bullet$]At the $(\kappa+1)$-th delegation query,
the adversary specifies a parent key and requests to fix the
$(\ell+1)$-th level of the hierarchy to $\vec X_{\ell+1}$. To answer
this query, the simulator first generates some components of
$\widehat{SK}_{DL}$ simply by choosing at random the values
 $\widehat L_j,\widehat L_{w,j,k}$ for all $j \in \{\ell+2,\ldots,d\}, k\in \{1,\ldots,\mu\},
i_1\in\{1,\ldots,\ell\}$. To compute the remaining components and
the decryption key component, the simulator sets
\begin{eqnarray*}
L_{\ell+1,i_1} &=& \prod_{k=1}^{\mu} \widehat L_{\ell+1,k,i_1}^{x_{\ell+1,k}}=\prod_{k=1}^{\mu} (\hat{v}_{k,i_1})^{x_{\ell+1,k}} \qquad \qquad i_1 \in \{1,\ldots,\ell \}\\
L_{\ell+1,\ell+1}&=&\widehat
L_{\ell+1}=\hat{v}_{1,\ell+1}\\
L_{w,\ell+1} &=& \prod_{k=1}^{\mu} \widehat
L_{w,\ell+1,k}^{x_{\ell+1,k}} = \prod_{k=1}^{\mu}
(\hat{v}_{k,d+1})^{x_{\ell+1,k}}.
\end{eqnarray*}
Since  the simulator knows the discrete logarithms
$c=\log_{\hat{v}}(\hat{w})$ and
$a_{i_1,i_2}=\log_{\hat{v}}(\hat{h}_{i_1,i_2})$, the remaining
components of $\widehat{SK}_{DL}$ and those of the partial
decryption key   can be generated efficiently in the same way as in
the proof of indistinguishability of Game$_{0,\kappa}$ and
Game$_{0,\kappa}'$. In particular, the simulator can compute
\begin{eqnarray*}
 K_{\ell+1}=\prod_{i_1=1}^{\ell} \prod_{i_2=1}^{\mu}
L_{\ell+1,i_1}^{a_{i_1,i_2}} \cdot
L_{\ell+1,\ell+1}^{\sum_{k=1}^{\mu} a_{\ell+1,k}} \cdot
L_{w,\ell+1}^c.
\end{eqnarray*}

To create $K_{\ell+1}^{(j,k)},\{L_{\ell+1,i_1}^{(j,k)}\}_{ i_1\in
\{1,\ldots,\ell+1\}}, L_{w,\ell+1}^{(j,k)}$, if $(j,k)\leq
(\eta+\ell+1,\nu)$ (in lexicographic order), the values are chosen
as fresh random delegation keys. For the $(\eta+\ell+1,\nu+1)$
partial decryption key, the simulator lets
\begin{eqnarray*}
L_{\ell+1,i_1}^{(\eta+\ell+1,\nu+1)} &=& \prod_{k=1}^{\mu} (Q'_{k,i_1})^{x_{\ell+1,k}}     \qquad \qquad i_1 \in \{1,\ldots,\ell \} \\
L_{\ell+1,\ell+1}^{(\eta+\ell+1,\nu+1)}&=&Q'_{1,\ell+1}\\
L_{w,\ell+1}^{(\eta+\ell+1,\nu+1)} &=& \prod_{k=1}^{\mu}
(Q'_{k,d+1})^{x_{\ell+1,k}}.
\end{eqnarray*}
Again, as the simulator knows the discrete logarithm of
$\hat{w},\hat{h}_{i_1,i_2}$ w.r.t. the base $\hat{v}$, the remaining
terms -- including $K_{\ell+1}^{(\eta+\ell+1,\nu+1)} $ -- can be
generated efficiently. For the remaining partial decryption keys,
the simulator generates them faithfully using the Delegate
algorithm.
\item [$\bullet$] The remaining delegated key queries are generated faithfully.
\end{itemize}
\noindent Clearly, if $Q'=Q$ in the GDDH instance, then the above
simulation is identically distributed as
Game$_{0,\kappa,\eta,\nu}'$, otherwise it is identically distributed
as
Game$_{0,\kappa,\eta,\nu+1}'$.    \\
\vspace{0.3cm}

The simulator generates the challenge ciphertext as normal and sends
it to the adversary. If the adversary has $\epsilon$ difference in
its advantage in Game$_{0,\kappa,\eta,\nu}'$ and
Game$_{0,\kappa,\eta,\nu+1}'$, it is not hard to see that the
simulator has a comparable advantage in solving the GDDH
instance. 
$\hfill$ $\Box$
\end{proof}
\end{proof}

\begin{lemma}
Game$_{1}$ and Game$_2$ are computationally indistinguishable if the
BDH assumption  holds in $(\G,\Ghat,\G_T)$.
\end{lemma}

\begin{proof} Assume that there's an adversary $\A$ that can distinguish between
Game$_{1}$ and Game$_{2}$. We build an adversary $\B$ that uses $\A$
as a subroutine to break the BDH assumption. The simulator $\B$
receives as input $$(g,g^a,g^c,\hat{g},\hat{g}^a,\hat{g}^b,Q'),$$
where $Q'$ is either $e(g,\hat{g})^{abc}$ or an element chosen
uniformly at random in $\G_T$.

The adversary $\A$ commits to two vectors $(\vec Y_1^0,\dots,\vec
Y_{d^{\star}}^0)$ and $(\vec Y_1^1,\dots,\vec Y_{d^{\star}}^1)$. The
challenger $\B$ picks  $\beta \sample \{0,1\}$. If $d^{\star} < d$,
the simulator picks at random $d-d^{\star}$ vectors
$Y_{d^{\star}+1}^{\beta},\ldots, Y_{d}^{\beta} \in
\mathbb{Z}_p^{\mu}$. The simulator $\B$ runs the Setup algorithm as
usual except that it implicitly sets $\alpha$ to be $ab$ by defining
$e(g,\hat{v})^{\alpha}=e(g^a,\hat{g}^b)^{\alpha_v}$, and then, for
$i_1=1,\ldots,d$ and $i_2=1,\ldots ,\mu$, it defines
$h_{i_1,i_2}=(g^a)^{z_{i_1}y_{i_1,i_2}}g^{t_{i_1,i_2}}$ for
$i_1\in\{1\dots,d\}$, $i_2\in\{1,\dots,\mu\}$ and
$h_{i_1,0}=(g^a)^{-z_{i_1}}$, for some random exponents
$z_{i_1},t_{i_1,i_2}$. The values of $\hat{h}_{i_1,i_2}$ for
$i_1\in\{1\dots,d\}$, $i_2\in\{0,\dots,\mu\}$ are defined similarly
from the value $\hat{g}^a$ of the BDH instance.  Observe that
$\beta$ remains hidden from $\A$ and that the parameters are
correctly distributed.

For the secret key queries, in the last lemma we have just proven
that delegated keys are indistinguishable from freshly generated
ones. Therefore, $\B$ will generate the secret keys using algorithm
Keygen when a reveal query is made. Note that $\alpha$ is defined as
$ab$, which is not known to $\B$. However, $\B$ needs to simulate
secret keys for $(\vec X_1,\dots,\vec X_\ell)$ as long as $f_{(\vec
X_1,\dots,\vec X_\ell)}(\vec Y_1^\beta,\dots,\vec
Y_{d^{\star}}^\beta)=0$. As $\alpha$ does not appear in the
delegating component of the key, the delegation component of the
secret keys $SK_{DL}$ can be created using the parameters as usual.
Therefore, from now on we focus on how to create the decryption
component of the secret key. Denote by $\ell'$ the index of the
smallest element of the vector $(\vec X_1,\dots,\vec X_\ell)$ for
which $\vec X_{\ell'}\cdot \vec Y_{\ell'}^\beta\neq 0$. This value
$\ell'$ always exists if $\ell \leq d^{\star} $ by hypothesis and
exists with overwhelming probability if $\ell > d^{\star}$, since
the vectors $Y_{d^{\star}+1}^{\beta},\ldots, Y_{d}^{\beta} \in
\mathbb{Z}_p^{\mu}$ are completely hidden from $\A$'s view.

The simulator creates the decryption component of the secret key as
follows:
$$D=\hat{g}^{ab}\displaystyle\prod_{i_1=1}^\ell\left(\prod_{i_2=1}^\mu
\hat{h}_{i_1,i_2}^{x_{i_1,i_2}}\right)^{r_{i_1}}\hat{w}^{r_w}$$
where the term $\hat{g}^{ab}\left(\prod_{i_2=1}^\mu
\hat{h}_{\ell',i_2}^{x_{\ell',i_2}}\right)^{r_{\ell'}}$ is computed
as
$$\left(\prod_{i_2=1}^\mu \hat{h}_{\ell',i_2}^{x_{\ell',i_2}}\right)^{\hat
r_{\ell'}}(\hat{g}^b)^{-(\sum_{i_2=1}^\mu
x_{\ell',i_2}t_{\ell',i_2})/(z_{\ell'}\vec X_{\ell'}\cdot\vec
Y_{\ell'}^{\beta})},$$ while the other terms in the product are just
computed as usual. It is not hard to see that, if we define
$r_{\ell'}=\hat r_{\ell'}-b/(z_{\ell'}\vec X_{\ell'}\cdot\vec
Y_{\ell'}^{\beta})$, then the computation is correct. All the other
terms in the decryption component can be computed efficiently, since
the simulator knows all the parameters needed and it also knows
$\hat{g}^b$.

At the challenge step, the adversary $\A$ gives $\B$ two messages,
$M_0$ and $M_1$. $\B$ then computes
\begin{eqnarray*}
C_{0} &=& M_{\beta}\cdot (Q')^{\alpha_v} ,\qquad \qquad
C_v=(g^{c})^{\alpha_v},\qquad \qquad C_w=(g^c)^{\alpha_w}, \qquad  \\
\{C_{i_1,i_2} &=&
(g^c)^{t_{i_1,i_2}}\}_{i_1\in\{1,\dots,d^{\star}\},\
i_2\in\{1,\dots,\mu\}}
\end{eqnarray*}
It is not hard to check that the challenge ciphertext is correctly
distributed.

Finally, when the adversary $\A$ outputs a guess $\beta'$, if
$\beta=\beta'$, then $\B$ guesses that $Q'=e(g,\hat{g})^{abc}$ and
if $\beta\neq\beta'$ guesses that $Q'=R$. If $\A$ has $\epsilon$
advantage in distinguishing between the two cases, then $\B$ also
has $\epsilon$ advantage in solving the assumption G3DH instance,
except in the case that $\A$ managed to output some secret key query
for vector $(\vec{X_1},\ldots,\vec{X}_{\ell})$ for which $\vec{X}_i
\cdot \vec{Y}_i^{\beta}=0$ for all $i=d^{\star}+1,\ldots,d$ was zero
for all $i \in \{1,\ldots, d\}$, which occurs only with negligible
probability.
\end{proof}

\begin{lemma} \label{1-2}
Game$_{2}$ and Game$_{3}$ are computationally indistinguishable if
the $\mathcal{P}$-BDH$_1$ assumption  holds in $(\G,\Ghat)$.
\end{lemma}
\begin{proof}
Assuming that the adversary $\A$ outputs $\beta'=\beta$ with
noticeably different probabilities in Game$_2$ and Game$_3$, we
build an distinguisher $\B$ for the $\mathcal{P}$-BDH$_1$ assumption . \\
\indent Namely, algorithm $\B$ receives as input a tuple
$$\big(g,~~g^b,~g^{ab},~g^c,~\ghat,~\ghat^{a},~
~\ghat^{b} ,~g^{z}  \big),$$ where $a,b,c \sample \Z_p$. Its goal is to decide if $z=abc$ or $z\in_R \Z_p$. \\
\indent To this end, $\B$ interacts with $\A$ as follows. It first
receives the vectors $(\vec{Y}_{0}^0,\ldots,\vec{Y}_{d^\star}^0)$,
$(\vec{Y}_{0}^1,\ldots,\vec{Y}_{d^\star}^1)$, at some depth $d^\star
\leq d$, that $\A$ wishes to be challenged upon. We may assume
w.l.o.g. that $\B$ chooses its challenge bit $\beta \sample \{0,1\}$
at the beginning of the game. Also, if $d^\star < d$, for each $i
\in \{d^\star+1,\ldots,d\}$, $\B$ defines the vector
$\vec{Y}_{i}^{\beta}$ as a random vector of $\Z_p^{\mu}$. \\
\indent It defines the master public key by setting
$e(g,\hat{v})^{\alpha} = e(g,\ghat)^{\alpha \cdot \gamma_v}$ and
\begin{eqnarray} \nonumber
v &=& g^{^{\gamma_v}} ,  \\ \nonumber h_{i_1,0} &=&
(g^b)^{\gamma_{i_1,0}}  \qquad \qquad \qquad
\qquad   \quad    ~~~ i_1 \in \{1,\ldots,d\}  \\
\label{pub-key}
 h_{i_1,i_2} &=&  (g^b)^{-\gamma_{i_1,0} \cdot y_{i_1,i_2}^{\beta} } \cdot   g^{\gamma_{i_1,i_2}}
  \qquad   \qquad i_1 \in \{1,\ldots,d\},~ i_2
\in \{1,\ldots,\mu\} \\ \nonumber w  &=& g^{y} \cdot (g^{ab} )^x
 ,
\end{eqnarray}
 where $\alpha,\gamma_v,x,y  \sample \Z_p$ and $\gamma_{i_1,i_2} \sample \Z_p$,  for each $i_1 \in
\{1,\ldots,d\}$, $i_2 \in \{0,\ldots,\mu\}$. For each $i_1$, we also
define the vector
$\vec{\gamma}_{i_1}=(\gamma_{i_1,1},\ldots,\gamma_{i_1,\mu}) \in
\Z_p^{\mu}$. We observe that $\B$ does not know $\what$ (which
depends on the unavailable term $\ghat^{ab}$) but {\it can} compute
 $\{\hat{h}_{i_1,i_2}\}_{i_1 \in \{1,\ldots,d\},i_2 \in \{1,\ldots,\mu\}}$.\\
\indent When the adversary  $\A$ requests a  key for a hierarchical
vector $\mathbf{X}=(\vec{X}_1,\ldots,\vec{X}_{\ell})$, $\B$ parses
$\vec{X}_{i_1}$ as  $(x_{i_1,1},\ldots,x_{i_1,\mu}) \in \Z_p^{\mu}$
for
each $i_1 \in \{1,\ldots,\ell \}$. Then, $\B$ responds as follows. \\
\vspace{-0.3 cm}

\noindent $\bullet$ If $\ell \leq d^\star$, let $i \in
\{1,\ldots,\ell\}$ be the smallest index  such that $\vec{X}_i \cdot
\vec{Y}_i^{\beta} \neq 0$   (by hypothesis, this index must exist).
By choosing $ r_{w}  \sample \Z_p$ and $r_i \sample \Z_p$, $\B$
implicitly defines the exponent
\begin{eqnarray*}
\tilde{r_i} &=&  {r_i}  + \frac{a \cdot r_w \cdot x }{ \gamma_{i,0}
\cdot \vec{X}_i \cdot \vec{Y}_i^{\beta} }
\end{eqnarray*}
and can compute
\begin{eqnarray} \nonumber
( \prod_{i_2=1}^{\mu} \hat{h}_{i,i_2}^{x_{i,i_2}} )^{\tilde{r}_i}
\cdot \hat{w}^{r_w}    &=& \big( \prod_{i_2=1}^{\mu}
\hat{h}_{i,i_2}^{x_{i,i_2}} \big)^{ {r}_i} \cdot   \bigl(
(\hat{g}^{b})^{ -\gamma_{i,0} \cdot \vec{X}_i \cdot
\vec{Y}_i^{\beta} } \cdot \hat{g}^{\vec{\gamma_i} \cdot
\vec{Y}_i^{\beta}} \bigr)^{a \cdot \frac{ r_w \cdot x
}{\gamma_{i,0} \cdot \vec{X} \cdot \vec{Y}_i^{\beta}}} \\
\label{key-gen-1}  & & \qquad \cdot (\ghat^{ab})^{r_w \cdot x  } \cdot \ghat^{r_w \cdot y  } \\
\nonumber &=& \big( \prod_{i_2=1}^{\mu} \hhat_{i,i_2}^{x_{i,i_2}}
\big)^{ {r}_i} \cdot (\ghat^a)^{\vec{\gamma_i} \cdot
\vec{Y}_i^{\beta} \cdot \frac{r_w \cdot x }{\gamma_{i,0} \cdot
\vec{X} \cdot \vec{Y}_i^{\beta}}} \cdot \ghat^{r_w \cdot y  }
\end{eqnarray}
without knowing $\ghat^{ab}$. Similarly, it can   compute
  $$D_i=\hat{v}^{\tilde{r}_i} =\ghat^{\gamma_v \cdot r_i} \cdot (\ghat^{a})^{ \gamma_v \cdot r_w \cdot x
  /(\gamma_{i,0} \cdot \vec{X} \cdot \vec{Y}_i^{\beta})}  $$ as well as
$D_w=\vhat^{r_w} $. \\
\indent We now turn to indices $i_1 \in \{1,\ldots,\ell\} \backslash
\{i\}$, for which $\B$ can trivially compute $( \prod_{i_2=1}^{\mu}
\hhat_{i,i_2}^{x_{i,i_2}} )^{{r}_{i_1}}$ and
$D_{i_1}=\vhat^{r_{i_1}} \cdot S_{v,i_1}$ since
it knows $\{\hhat_{i_1,i_2}\}_{i_2=0}^{\mu}$ and $v$. This   suffices for computing the whole decryption component $SK_D$ of the private key.\\
\indent As for the delegation component $SK_{DL}$, 
$\B$ can  compute $\{K_{j,k}\}_{j\in \{\ell+1,\ldots,d\}, k \in
\{1,\ldots,\mu\}}$, $\{L_j\}_{j=\ell+1}^d$ and $\{L_{w,j,k}
\}_{j,k}$ by applying exactly the same procedure as for $D$, $D_i$
and $D_w$  (and taking advantage of the fact that $\vec{X}_i \cdot
\vec{Y}_i^{\beta} \neq 0$ for at least one $i\in
\{1,\ldots,\ell\}$). Remaining pieces of $SK_{DL}$ are then
trivially computable since $\B$ has
$\{\hhat_{i_1,i_2}\}_{i_2=0}^{\mu}$ and $\vhat$ at disposal. \\
\vspace{-0.3 cm}

\noindent $\bullet$ If $\ell > d^\star$, it can be the case that
$\vec{X}_i \cdot \vec{Y}_i^{\beta}=0$ for $i=1$ to $\ell$. However,
with overwhelming probability, there must exist $i \in
\{d^\star+1,\ldots,\ell\}$ such that $\vec{X}_i \cdot
\vec{Y}_i^{\beta} \neq 0$ since the vectors
$\vec{Y}_{d^\star+1},\ldots,\vec{Y}_d$ have been chosen at random
and, due to the generation of the public key as per (\ref{pub-key}),
they are completely independent of $\A$'s view. It comes that $\B$
can generate a private key in the same way as in the case $\ell
<d^\star$ (see
 equation (\ref{key-gen-1})).  \\ \vspace{-0.3 cm}

\indent When $\B$ has to construct the challenge ciphertext, $\B$
sets $C_0=M \cdot e(g^c,\ghat^{\gamma_v} )^{\alpha}$, where $M
\sample \G_T$, and
\begin{eqnarray*}
C_v=(g^c )^{\gamma_v} , \qquad \qquad \qquad C_w=(g^c )^y \cdot
(g^{z} )^x ,
\end{eqnarray*}
   \begin{eqnarray*}
    C_{i_1,i_2}=(g^c  )^{\gamma_{i_1,i_2}}  \qquad  \qquad \qquad i_1
    \in \{1,\ldots,d^\star\}, ~i_2 \in \{1,\ldots,\mu\}
\end{eqnarray*}
We observe that, if $z=abc$, $(C_0,C_v,C_w, \{C_{i_1,i_2}\}_{i_1 \in
\{1,\ldots,d^\star\},i_2 \in \{1,\ldots,\mu\}})$ corresponds to a
valid ciphertext with the encryption exponent $s=c$. In this situation, $\B$ is playing Game$_2$ with $\A$. \\
\indent In contrast, if $z \in_R \Z_p$, we have $z\neq c$ with
overwhelming probability. In this case,  we have $g^{z}=g^{abc +
\theta}$, for some $\theta \neq 0$, and we can thus write $C_w=w^c
\cdot g^{\theta \cdot x}  $. This means that  $C_w$  looks uniformly
random and independent from $\A$'s view. Indeed, until the challenge
phase, $\A$ has no information about $\theta \in \Z_{p}$ (recall
that public parameters do not depend on $c$) and the value $x \in
\Z_{p_1}$ is also independent of $\A$'s view.   We conclude that, if
$g^z  $ is such that $z \in_R \Z_{p}$, we are in Game$_3$. \qed
\end{proof}

\begin{lemma} \label{4-5-int}
For each $\delta_1 \in \{1,\ldots,d\}$ and each $\delta_2 \in
\{2,\ldots,\mu\}$, Game$_{4,\delta_1,\delta_2-1}$ and
Game$_{4,\delta_1,\delta_2}$ are computationally indistinguishable
if the $\mathcal{P}$-BDH$_1$ assumption holds in $(\G,\Ghat)$.
\end{lemma}
\begin{proof}
Towards a contradiction, we assume there exists $\delta_1,\delta_2$
such that the adversary $\A$ outputs $\beta'=\beta$ with
significantly different probabilities in
Game$_{4,\delta_1,\delta_2-1}$ and Game$_{4,\delta_1,\delta_2}$. We
show that
$\A$ implies  a distinguisher $\B$ against  $\mathcal{P}$-BDH$_1$. \\
\indent Our distinguisher $\B$ receives as input
 $\big(g, g^b,g^{ab},g^c,\ghat,\ghat^{a},\ghat^{b}
,g^{z}  \big),$    with $a,b,c \sample \Z_p$. It aims to decide if $z=abc$ or $z\in_R \Z_p$. \\
\indent To do this, $\B$ runs the adversary $\A$ as follows. It
first receives the challenge vectors
$(\vec{Y}_{0}^0,\ldots,\vec{Y}_{d^\star}^0)$,
$(\vec{Y}_{0}^1,\ldots,\vec{Y}_{d^\star}^1)$, at some depth $d^\star
\leq d$, that are chosen by $\A$. We  assume   that $\B$ chooses its
challenge bit $\beta \sample \{0,1\}$ at the outset of the game.
Also, if $d^\star < d$, for each $i \in \{d^\star+1,\ldots,d\}$,
$\B$ defines the vector
$\vec{Y}_{i}^{\beta}$ as a random vector of $\Z_p^{\mu}$. \\
\indent It defines the master public key by setting
$e(g,\vhat)^{\alpha} = e(g,\ghat)^{\alpha \cdot \gamma_v}$ and
\begin{eqnarray} \nonumber
v &=& g^{{\gamma_v}} ,  \\ \nonumber h_{i_1,0} &=&
g_{p_1}^{\gamma_{i_1,0}}   \quad \qquad \qquad \qquad
\quad    ~~\textrm{ for } i_1 \in \{1,\ldots,d\}  \\
\nonumber
 h_{i_1,i_2} &=&  g^{-\gamma_{i_1,0} \cdot y_{i_1,i_2}^{\beta} } \cdot   g^{\gamma_{i_1,i_2}}
  \qquad  ~   \textrm{ if } i_1 \in
\{1,\ldots,d\} \backslash \{\delta_1\} \textrm{ or }~ i_2 \in
\{1,\ldots,\mu\} \backslash \{\delta_2\} \\ \label{pub-key-2}
 h_{\delta_1,\delta_2} &=&  g^{-\gamma_{\delta_1,0} \cdot y_{\delta_1,\delta_2}^{\beta} } \cdot   (g^{ab} )^{\gamma_{\delta_1,\delta_2}}
  \\
\nonumber w &=&
g^{y} \cdot (g^{b} )^x , 
\end{eqnarray}
 where $\alpha,\gamma_v,y  \sample \Z_p$, $x \sample \Z_p^*$
 and $\gamma_{i_1,i_2} \sample \Z_p$,  for each $i_1 \in
\{1,\ldots,d\}$, $i_2 \in \{0,\ldots,\mu\}$. For each $i_1$, we also
define the vector
 $\vec{\gamma}_{i_1}=(\gamma_{i_1,1},\ldots,\gamma_{i_1,\mu}) \in \Z_p^{\mu}$. Note that, in the implicitly defined master secret key $\mathsf{msk}$,
 the distinguisher  $\B$ knows all the components but $\hat{h}_{\delta_1,\delta_2}$, which depends on the unknown term $\hat{g}^{ab}$. \\
\indent When the adversary  $\A$ requests a  key for a hierarchical
vector $\mathbf{X}=(\vec{X}_1,\ldots,\vec{X}_{\ell})$, $\B$ parses
$\vec{X}_{i_1}$ as  $(x_{i_1,1},\ldots,x_{i_1,\mu}) \in \Z_p^{\mu}$
for
each $i_1 \in \{1,\ldots,\ell \}$. Then, $\B$ responds as follows. \\
\vspace{-0.3 cm}

\noindent $\bullet$ If $\ell \geq \delta_1$,  $\B$ chooses  $ r_{w}'
\sample \Z_p$ and $r_{\delta_1} \sample \Z_p$, $\B$ implicitly
defines the exponent
\begin{eqnarray} \label{keygen-last-0}
{{r}_w} &=&  {r}_{w}'  - \frac{a \cdot r_{\delta_1} \cdot
\gamma_{\delta_1,\delta_2}}{  {x} }
\end{eqnarray}
and can compute the product
\begin{eqnarray} \label{keygen-last-1}
 \Bigl(\hhat_{\delta_1,0}^{y_{\delta_1,\delta_2}^{\beta}} \cdot
 \hhat_{\delta_1,\delta_2} \Bigr)^{r_{\delta_1}} \cdot \hat{w}^{r_w} &=&
 (\ghat^{ab})^{\gamma_{\delta_1,\delta_2} \cdot r_{\delta_1}}
 \cdot  \hat{w}^{{r}_w'} \cdot \big( \ghat^y (\ghat^b)^{x}
 \big)^{-a \cdot r_{\delta_1} \cdot \gamma_{\delta_1,\delta_2}
 /{x}} \\ \nonumber  &=& \hat{w}^{{r}_w'} \cdot (\ghat^a)^{-y \cdot r_{\delta_1} \cdot \gamma_{\delta_1,\delta_2}
 /{x}} ,
\end{eqnarray}
which is the only factor of $D$ that it cannot trivially compute
without knowing $\ghat^{ab}$. Similarly, it can compute
  $D_w =v^{r_w}= \ghat^{\gamma_v \cdot {r}_w'} \cdot (\ghat^{a})^{- \gamma_v \cdot r_{\delta_1} \cdot \gamma_{\delta_1,\delta_2}/ {x} }
   $. \\
 \ \indent To generate the delegation component $SK_{DL}$ of the key,  the reduction $\B$
 is able to
compute   $\{K_{j,k},L_{w,j,k} \}_{j\in \{\ell+1,\ldots,d\},k\in
\{1,\ldots,\mu\}}$ by repeating $(d-\ell)\cdot \mu$ times the same
procedure as for computing $D$ and $D_{w}$.
 \\
\vspace{-0.3 cm}

\noindent $\bullet$ If $\ell< \delta_1$,  $\B$ can directly compute
$(D,D_w,\{D_{i_1}\}_{i_1=1}^{\ell})$ since it knows
$\{\hhat_{i_1,i_2}\}_{i_1 \in \{1,\ldots,\ell\}, i_2\in
\{1,\ldots,\mu\}}$,  $\vhat$ and $\what$. The difficulty is to
compute the delegation components
 $\{K_{\delta_1,k} \}_{k=1}^{\mu}$ without knowing $\ghat^{ab}$. In
 fact, among these components,
  the only factor of $K_{\delta_1,\delta_2} $ that $\B$ cannot trivially compute is $\hhat_{\delta_1,\delta_2}^{s_{\delta_1}}$.  However,
similarly to (\ref{keygen-last-0})-(\ref{keygen-last-1}),
   it can choose $s_{\delta_1},{s}_{w,\delta_1,\delta_2}' \sample \Z_p$, define
   ${s}_{w,\delta_1,\delta_2}={s}_{w,\delta_1,\delta_2}'-
   \frac{
    a \cdot  s_{\delta_1} \cdot  \gamma_{\delta_1,\delta_2}}{{x}}$ and
   compute the product
\begin{eqnarray}  \nonumber
h_{\delta_1,\delta_2}^{s_{\delta_1}} \cdot
\hat{w}^{{s}_{w,\delta_1,\delta_2} } &=&
(\ghat^{-\gamma_{\delta_1,0} \cdot y_{\delta_1,\delta_2}^{\beta}
\cdot s_{\delta_1} } ) \cdot
(\ghat^{ab})^{\gamma_{\delta_1,\delta_2} \cdot s_{\delta_1}} \cdot
 \big( \ghat^y (\ghat^b)^{x} \big)^{-a \cdot s_{\delta_1} \cdot
 \gamma_{\delta_1,\delta_2}/{x}} \cdot \hat{w}^{{s}_{w,\delta_1,\delta_2}'
 } \\ \label{keygen-3}
 &=& (\ghat^{-\gamma_{\delta_1,0} \cdot y_{\delta_1,\delta_2}^{\beta}
\cdot s_{\delta_1} } )   \cdot \what^{{s}_{w,\delta_1,\delta_2}'}
\cdot (\ghat^a)^{-y \cdot s_{\delta_1} \cdot
 \gamma_{\delta_1,\delta_2}/{x}}
\end{eqnarray}
In the same way, $\B$ computes
$$L_{w,\delta_1,\delta_2}=\vhat^{s_{w,\delta_1,\delta_2}}   = \ghat^{\gamma_v \cdot
 {s}_{w,\delta_1,\delta_2}'} \cdot (\ghat^a)^{- \gamma_v \cdot
s_{\delta_1} \cdot \gamma_{\delta_1,\delta_2} / x }   .$$
 \indent Note that,
 for each $k \in \{1,\ldots,\mu\}\backslash \{\delta_2\}$,
$\B$ has to generate $K_{\delta_1,k}$ by  computing
$\hhat_{\delta_1,k}^{s_{\delta_1}}$ using the same random exponent
$s_{\delta_1}$ as in (\ref{keygen-3}). This is always possible since $\B$ knows that exponent.  \\
\vspace{-0.3 cm}

\indent When it comes to construct the challenge ciphertext,
algorithm  $\B$ first sets $C_0=M \cdot
e(g^c,\ghat^{\gamma_v})^{\alpha}$, where $M \sample \G_T$. It also
chooses $\Gamma_w \sample \G$ and computes
\begin{eqnarray*}
C_v=(g^c )^{\gamma_v} , \qquad \qquad \qquad \qquad C_{w}= \Gamma_w,
\end{eqnarray*}
as well as
   \begin{eqnarray*}
    C_{i_1,i_2} &=& (g^c)^{\gamma_{i_1,i_2}}  \qquad \qquad
    \qquad
    \textrm{ if }~~ (i_1 > \delta_1) \vee \big( (i_1 = \delta_1) \wedge (i_2 >
    \delta_2) \big) \\
    C_{\delta_1,\delta_2} &=& (g^{z}
     )^{\gamma_{\delta_1,\delta_2}}
\end{eqnarray*}
If $i_1 < \delta_1$ or $i_1=\delta_1$ and $i_2 < \delta_2$, then
$C_{i_1,i_2}$ is  chosen uniformly in $\G$.  \\ \indent
 We observe that, in the situation where $z=abc$,
$(C_0,C_v,C_w, \{C_{i_1,i_2}\}_{i_1 \in
\{1,\ldots,d^\star\},i_2 \in \{1,\ldots,\mu\}})$  is distributed in the same way as in Game$_{4,\delta_1,\delta_2-1}$.   \\
\indent In contrast, if $z \in_R \Z_p$, we have $z\neq c$ with
overwhelming probability. In this case, $C_{\delta_1,\delta_2}$
looks random to the adversary and $\B$ is thus playing
Game$_{4,\delta_1,\delta_2}$.  \qed
\end{proof}

\section{Deferred Proofs for the Security of $\HF$}
\label{deferred-lossy-selective}

\subsection{Proof of Lemma \ref{indist-setup}} \label{hybrid-par}

\hspace{0.1 cm} \vspace{0.15 cm}

\noindent We consider a sequence of $n+1$ hybrid experiments $RL_0$,
\ldots,
 $RL_n$. For each $k \in \{0,\ldots,n\}$, $RL_k$ is defined to be
an experiment where public parameters are generated as follows.
First, the simulator $\B$ chooses $ v \sample \G$, $\mathbf{w}
\sample \G^n,\hat{\mathbf{w}} \sample \Ghat^n$,
  $\hh \sample
\G^{d \times (\mu+1) \times n}$,  $\hat{\hh} \sample \Ghat^{d \times
(\mu+1) \times n}$ and computes $\mathsf{PP}_{core}$ in the same way
as in the real scheme.
\\ \indent In the second step of the setup procedure, the simulator
$\B$ chooses a vector  $\mathbf{s} \sample (\Z_p^*)^n$. For $l_1,l_2
\in \{1,\ldots,n\}$, it first computes
\begin{eqnarray*}
\mathbf{J}[l_2] &=& v^{\mathbf{s}[l_2]} , \\
\mathbf{C}_w[l_2,l_1] &=& \mathbf{w}[l_1]^{\mathbf{s}[l_2]}.
\end{eqnarray*}
Then, for each pair $(l_1,l_2)$ such that $l_1 \neq l_2$, $\B$ sets
\begin{eqnarray*}
\mathbf{C}[i_1,i_2,l_2,l_1]= \hh[i_1,i_2,l_1]^{\mathbf{s}[l_2]}  .
\end{eqnarray*}
Finally, for each $l \in \{1,\ldots,n\}$, $i_1 \in \{1,\ldots,d\}$
and $i_2 \in \{1,\ldots,\mu\}$, $\B$ defines
\begin{eqnarray*}
\mathbf{C}[i_1,i_2,l,l] &=& \Bigl(
\hh[i_1,0,l]^{\yy_{i_1}^{(0)}}[i_2] \cdot  \hh[i_1,i_2,l]
\Bigr)^{\mathbf{s}[l]}   \qquad \qquad \textrm{ if } l \leq k, \\
\mathbf{C}[i_1,i_2,l,l] &=& \Bigl(
\hh[i_1,0,l]^{\yy_{i_1}^{(1)}}[i_2] \cdot  \hh[i_1,i_2,l]
\Bigr)^{\mathbf{s}[l]}   \qquad \qquad \textrm{ if } l > k.
\end{eqnarray*}
Lemma \ref{plus-one} below demonstrates that, for each $k \in
\{1,\ldots,n\}$,  experiment $RL_k$ is computationally
indistinguishable from experiment
 $RL_{k-1}$. \\
\indent If we assume that the statement of  Lemma \ref{indist-setup}
is false, there must exist $k \in \{1,\ldots,n\}$ such that the
adversary can distinguish  $RL_k$ from $RL_{k-1}$ and we obtain a
contradiction. \qed

\begin{lemma} \label{plus-one}
 If the HPE scheme described in Section \ref{sec:new-HIPE} is
selectively weakly attribute-hiding, then Game $RL_k$ is
indistinguishable from Game $RL_{k-1}$ for each $k \in
\{1,\ldots,n\}$. Namely, for each $k$, there exist algorithms $\B_1$
and $\B_2$ such that  $$| \PR[RL_k \Rightarrow 1] - \PR[RL_{k-1}
\Rightarrow 1]| \leq (d \cdot \mu +1) \cdot
\mathbf{Adv}^{\mathcal{P}\textrm{-}\mathrm{BDH}_1}(\B_1) + q \cdot
\mathbf{Adv}^{\mathcal{P}\textrm{-}\mathrm{BDH}_2}(\B_2), $$ where
$q$ is the number of ``Reveal-key'' queries made by $\A$.
\end{lemma}

\begin{proof}

 For the sake of  contradiction, let us assume that there exist two
 auxiliary
hierarchical vector  $\yy^{(0)}=[\yy_1^{(0)} | \ldots |
\yy_{d}^{(0)}]$, $\yy^{(1)}=[\yy_1^{(1)} | \ldots | \yy_{d}^{(1)}]$
and an index
 $k \in
\{1,\ldots,n\}$ such that the adversary $\A$ has noticeably
different behaviors in experiments  $RL_k$ and  $RL_{k-1}$. Using
$\A$, we construct a  selective weakly attribute-hiding adversary
$\B$ against the   HPE scheme described in Section
\ref{sec:new-HIPE} (in   its predicate-only variant).
\\
\indent Our adversary
 $\B$ first declares   $\yy^{(0)}, \yy^{(1)} \in
\Z_p^{d \cdot \mu}$ as the vectors that it wishes to be challenged
upon. Then, the HPE challenger provides $\B$ with public parameters
$$
 \mathsf{mpk}_{\textrm{HPE}}  =   \Bigl(  v , ~ w ,   ~  \{ h_{i_1,i_2} \}_{i_1 \in \{1,\ldots,d\},
 i_2 \in \{0,\ldots,
  \mu \}}
\Bigr).
$$
   Then, $\B$   chooses a vector
  $\vzeta    \sample \Z_p^n$ and a matrix
   $\vgamma \in  \Z_{p}^{d \times (\mu+1) \times n}$, which it uses to compute
\begin{eqnarray*}
 \mathbf{w}[l_1] &=&  v^{\vzeta[l_1]}   \quad       \qquad ~ \textrm{
 for }
   l_1 \in \{1,\ldots,n\} \backslash \{k \} \\
 \hh[i_1,i_2,l_1] & =& v^{\vgamma[i_1,i_2,l_1]}
     \quad   ~ \textrm{ for } i_1 \in
\{1,\ldots,d\}  ,\ i_2 \in \{0,\ldots,
  \mu \},~l_1\in \{1,\ldots,n\} \backslash \{k \}.
\end{eqnarray*}
It also sets $\mathbf{w}[k]=w$  as well as
\begin{eqnarray*}
\hh[i_1,i_2,k]= h_{i_1,i_2} \qquad  \qquad \textrm{ for } i_1 \in
\{1,\ldots,d\},\ i_2 \in \{0,\ldots,\mu \}.
\end{eqnarray*}
Then,  $\B$ defines core public parameters
$$
\PP_{core} = \Bigl( v ,~\{\mathbf{w}[l_1] ,  ~~  \{
\hh[i_1,i_2,l_1]\}_{i_1 \in \{1,\ldots,d\}, i_2 \in \{0,\ldots,
  \mu \},~l_1\in \{1,\ldots,n\}}
\Bigr),
$$
that correspond  to the  master secret key
$\mathsf{msk}=(\vhat,\hat{\mathbf{w}},\hat{\mathbf{h}})$, which is
  not completely known to  $\B$ (specifically, $\vhat$, $\hat{\ww}[k]$ and $\hat{\hh}[.,.,k]$ are not available).
Then, $\B$ notifies its HPE challenger that it wishes to directly
enter the challenge phase without making any pre-challenge query.
The challenger replies with the challenge ciphertext
$$
C^\star=\bigl( C_v,C_{w},  \{C_{i_1,i_2}  \}_{i_1 \in \{1,\ldots,d
\},~i_2 \in \{1,\ldots,\mu \} }  \bigr),
$$
where
\begin{eqnarray*}
    C_v = v^{s}  , \qquad \qquad \qquad C_w= w^s  ,
\end{eqnarray*}
\vspace{-0.9 cm}
\begin{eqnarray*}
   \{\ C_{i_1,i_2}= \bigl( h_{i_1,0}^{
y^{(\beta)}_{i_1,i_2}} \cdot h_{i_1,i_2} \bigr)^s   \ \}_{ i_1 \in
\{1,\ldots,d \},~i_2\in \{1,\ldots,\mu \}},
\end{eqnarray*}
for a random element  $s \sample \Z_p^*$  and   a random bit $\beta
\in \{0,1\}$. Here we are using, for $\beta \in \{0,1\}$, the
notation $\yy{(\beta)}=[\yy^{(\beta)}_1|\ldots |\yy^{(\beta)}_d] \in
\Z_p^{d \cdot \mu}$, where each $\yy^{(\beta)}_{i_1} =
(y^{(\beta)}_{i_1,1},\ldots,y^{(\beta)}_{i_1,\mu}) \in \Z_p^{\mu}$,
for $i_1 \in \{1,\ldots,d\}$.\\
\indent At this point, $\B$  constructs the matrix
$\{\mathbf{CT}[i_1,i_2]\}_{l_1,l_2 \in \{1,\ldots,n\} }$ of HPE
ciphertexts by setting
\begin{eqnarray*}
\mathbf{J}[k] &=& C_v \\
\mathbf{C}_w[k,k] &=& C_w \\
\mathbf{C}[i_1,i_2,k,k] &=& C_{i_1,i_2} \qquad \textrm{ for } i_1
\in \{1,\ldots,d\},\ i_2 \in \{1,\ldots,\mu\}.
\end{eqnarray*}
and, for each $l_1 \in \{1,\ldots,n\} \backslash \{k\}$,
\begin{eqnarray*}
\mathbf{C}_w[k,l_1] &=&  C_v^{\vzeta[l_1]}   \\
\mathbf{C}[i_1,i_2,k,l_1] &=& C_v^{\vgamma[i_1,i_2,k,l_1]}   \qquad
\qquad \textrm{ for } i_1 \in \{1,\ldots,d\},\ i_2 \in
\{1,\ldots,\mu\}.
\end{eqnarray*}
  Note that this implicity sets
$\ess[k]=s$, where $s$ is the encryption exponent chosen by the HPE
challenger to compute $C^\star$. Then, for each $l_2 \in
\{1,\ldots,n\} \backslash \{k\}$, $\B$ chooses a random exponent
$\ess[l_2] \sample \Z_p^*$ and computes
\begin{eqnarray} \label{exp-use}
\mathbf{J}[l_2] &=& v^{\ess[l_2]}   \\ \nonumber
 \mathbf{C}_w[l_2,l_1] &=&  \mathbf{w}[l_1]^{\ess[l_2]}
 \qquad \qquad \qquad~ \textrm{ for } l_1 \in \{1,\ldots,n\} \backslash \{l_2\} \\ \nonumber
\mathbf{C}[i_1,i_2,l_2,l_1] &=&  \hh[i_1,i_2,l_1]^{\ess[l_2]} \qquad
\qquad \textrm{ for } i_1 \in \{1,\ldots,d\},\ i_2 \in
\{1,\ldots,\mu\},\\ & & \qquad  \qquad  \qquad  \qquad  \qquad \quad
~~~ l_1 \in \{1,\ldots,n\} \backslash \{l_2\}
\end{eqnarray}
As for   entries of the form $\{ \mathbf{C}[i_1,i_2,l
,l]\}_{i_1,i_2,l\neq k }$, $\B$ computes them    as
\begin{eqnarray*}
 \mathbf{C}[i_1,i_2,l] &=& \big( \hh[i_1,0,l]^{\yy_{i_1}^{(0)}}[i_2]
\cdot  \hh[i_1,i_2,l] \big)^{\ess[l]}   \qquad \qquad \textrm{ if } l < k \\
\mathbf{C}[i_1,i_2,l] &=& \big( \hh[i_1,0,l]^{\yy_{i_1}^{(1)}}[i_2]
\cdot \hh[i_1,i_2,l] \big)^{\ess[l]}  \qquad \qquad \textrm{ if } l
> k.
\end{eqnarray*}
 using the exponents $\ess[l] \in \Z_p^*$ that were chosen in
(\ref{exp-use}). Finally,
 our adversary $\B$ defines the $n \times n$ matrix $\{\CC[l_2,l_1]
\}_{l_2,l_1 \in \{1,\ldots,n\}}$ of HPE ciphertexts
$$
\CC[l_2,l_1]=\big( \JJ[l_2],\CCC_w[l_2,l_1], \{
\CCC[i_1,i_2,l_2,l_1] \}_{i_1 \in \{1,\ldots,d \},~i_2 \in
\{1,\ldots,\mu \} } \big).
$$
Finally, $\B$ defines $\mathsf{mpk}:=\big( \PP_{core},
\{\CC[l_2,l_1] \}_{l_2,l_1 \in \{1,\ldots,n\}}  \big)$ and sends
it to the adversary $\A$. \\
\indent When it comes to answer $\A$'s private key queries for
hierarchical identities $(\id_1,\ldots,\id_{\ell})$, $\B$ first
encodes each level's identity $\id_{i_1} \in \{0,1\}^{\mu}$ as a
$\mu$-vector $\vec{X}_{i_1}=(\id_{i_1}[1],\ldots,\id_{i_1}[\mu])$
for each $i_1 \in \{1,\ldots,\ell\}$. Although $\B$ does not
entirely know $\mathsf{msk}$, the decryption components
$(\mathbf{D}[l_1],\mathbf{D}_w[l_1], \mathbf{D}_{i_1}[l_1])$ of the
private key are always directly computable when $l_1\neq k$: namely,
$\B$ chooses $\mathbf{D}_w[l_1] \sample \Ghat$ and
$\mathbf{D}_{i_1}[l_1] \sample \Ghat$, for $i_1 =1$ to $\ell$, and
computes
\begin{eqnarray*}
\mathbf{D}[l_1] &=&  \prod_{i_1=1}^{\ell} \mathbf{D}_{i_1}[l_1]^{
\sum_{i_2=1}^{\mu} \vgamma[i_1,i_2,l_1] \cdot \id_{i_1}[i_2]  }
\cdot \mathbf{D}_w[l_1]^{\vzeta[l_1]}  .
\end{eqnarray*}
  It is easy to see that
$(\mathbf{D}[l_1],\mathbf{D}_w,
\{\mathbf{D}_{i_1}[l_1]\}_{i_1=1}^{\ell})$ forms a  decryption
component of the form (\ref{dec-comp}). Moreover, the delegation
components can be obtained exactly in the same way. \\ \indent  As
for the remaining coordinate $l_1=k$, the simulator $\B$ aborts if,
for any $\tilde{\gamma}\in\{0,1\}$, the obtained hierarchical vector
$(\vec{X}_1,\ldots,\vec{X}_{\ell})$ is one for which $\langle
y_{i_1}^{(\tilde{\gamma})},X_{i_1} \rangle =0$ for each $i_1 \in
\{1,\ldots,\ell\}$ (which translates into
$$f_{(\vec{X}_1,\ldots,\vec{X}_{\ell})}(\yy_1^{(\tilde{\gamma})},\ldots,\yy_d^{(\tilde{\gamma})})=1$$
in the predicate encryption language). Otherwise, $\B$ can obtain
the missing private key components by invoking its HPE challenger
 to obtain a complete private
key $\mathbf{SK}_{(\id_1,\ldots,\id_{\ell})}=(\mathbf{SK}_D,\mathbf{SK}_{DL})$. \\
 \indent It is easy to check that, if the HPE challenger's
bit is $\beta=0$, $\mathsf{mpk}$ is distributed as in Game $RL_{k}$.
In contrast, if $\beta=1$, $\mathsf{mpk}$ has the same distribution
as in Game $RL_{k-1}$.  \qed
\end{proof}


\subsection{Selective Security of $\HF$}\label{app:selective-security}

We consider our $\HF$ with $\mu=2$, $\Sid=\{(1,x): x \in
\Z_p^*\}$ and $\IdSp=(\Sid)^{(\leq d)}$. Let
$(\id_1^{\star},\dots,\id_{\ell^{\star}}^{\star}) \in \IdSp$ be the
identity chosen by the adversary. Define the sibling $\LHF$ with
$\AxSp=\Z_p^{2d}$ and auxiliary input $\yy^{(1)}=[\yy_1^{(1)}|\ldots
|\yy_d^{(1)}] \in \Z_p^{2d }$ where
$\yy_{i_1}^{(1)}=(-id_{i_1}^{\star},1)$ for all $i_1 \in \{1,\ldots,
\ell^{\star}\}$ and $\yy_{i_1}^{(1)} =(1,0)$ for $i_1 \in \{
\ell^{\star}+1,\ldots, d \}$. Selective security is established in
the following theorem. In this selective setting, we will omit $\zeta$ in the notation,
because the trivial pre-output stage that always outputs $d_2=1$ will be enough in this case and, as a consequence, the experiments do not depend
on $\zeta$.

\begin{theorem}\label{theor-select-loss}
Let $n>\log{p}$ and let $\omega=n-\log{p}$. Let $\HF$ be the
HIB-TDF with parameters $n,d,\mu=2$, $\Sid=\{(1,x): x \in \Z_p^*\}$
and $\IdSp=(\Sid)^{(\leq d)}$. Let $\LHF$ be the sibling associated
with it as above. Then $\HF$ is $(\omega,1)$-partially lossy against selective-id adversaries.

Specifically regarding condition (i), for any selective-id adversary $\A$ there
exist algorithms $\B_1$ and $\B_2$ such that
$$\mathbf{Adv}_{\HF,\LHF,\omega}^{\mathrm{lossy}}(\A)\leq n\cdot\left((2d+1)\cdot \mathbf{Adv}^{\mathcal{P}\textrm{-}\mathrm{BDH}_1}(\B_1)+q\cdot \mathbf{Adv}^{\mathcal{P}\textrm{-}\mathrm{DDH}_2}(\B_2)\right)$$
The running time of $\B_1$ and $\B_2$ are comparable to the running
time of $\A$.
\end{theorem}

\begin{proof} Let $RL_0$ and $RL_n$ be the games specified above. We claim that both $$\Pr[d_{exp}^\A=1\ |\
\REAL_{\HF,\LHF,\omega}^\A]=\Pr[RL_0\Rightarrow 1]$$ and
$$|\Pr[d_{exp}^\A=1\ |\
\LOSSY_{\HF,\LHF,\omega}^\A]-\Pr[RL_n\Rightarrow 1]| \in
\mathsf{negl}(\varrho),$$ which implies condition (i) of partial lossiness, when combined with Lemma \ref{indist-setup}. \\
\indent To prove that, we just need to justify that, when using
$\yy^{(0)}$ as the auxiliary input, the function $\HFEv$ is
injective for all identities $\id$. On the other hand, when using
$\yy^{(1)}$, the function $\HFEv$ will be lossy for identities
$\id=(\id_1,\ldots,\id_{\ell})$ such that $\langle
\yy_{i_1}^{(1)},\id_{i_1} \rangle =0$ for all $i_1 \in
\{1,\ldots,\ell\}$ and injective for all other identities.

To see that, note that all the terms of the output of $\HFEv$ are
determined by $\langle
 \mathbf{s},X\rangle$ except the term $\CC_{\id}[i_1,l_1]$, which is
 also determined by $\ess[l_1] \cdot x_{l_1} \cdot   \langle
\yy_{i_1},\id_{i_1} \rangle$. Therefore, when $\langle
\yy_{i_1},\id_{i_1}\rangle=0$ for all $i_1\in\{1,\dots,\ell\}$, the
function will be determined only by $\langle
 \mathbf{s},X\rangle$, which can take only $p$ values, so $\HFEv$ will be lossy with lossiness $\lambda\left( \HFEv
\big(\mathsf{pms},\mathsf{mpk},(\id_1,\ldots,\id_{\ell}),. \big)
\right) \geq n-\log{p} =\omega$. On the other hand, if there is at
least one index $i_1$ for which $\langle
\yy_{i_1},\id_{i_1}\rangle\neq0$, then $\HFEv$ will be injective
with overwhelming probability and $\HFInv$ will be its inverse, also
with overwhelming probability. Finally, observe that when the
auxiliary input is $\yy^{(0)}$, then for all identities
$(\id_1,\dots,\id_\ell)$ and for all $i_1\in\{1,\dots,\ell\}$,
$\langle \yy_{i_1}^{(0)},\id_{i_1}\rangle \equiv 1$ by construction.

With this, we conclude that $\yy^{(0)}$ makes $\HFEv$ injective for
all identities and $\yy^{(1)}$ makes $\HFEv$ lossy only for
identities such that $\langle \yy_{i_1},\id_{i_1}\rangle=0$ for all
$i_1\in\{1,\dots,\ell\}$, and in this case the lossiness is $n-\log
p$ (where when we claim that $\HFEv$ is injective, it holds with
overwhelming probability). This concludes the proof of condition (i).

Regarding conditions (ii) and (iii), when
the auxiliary input is $\yy^{(1)}$, then $\HFEv$ is injective for all
valid queried $\id$ in ``Reveal key'' queries and $\HFEv$ is lossy
for $\id^\star=(\id_1^\star,\dots,\id_{\ell^\star}^\star)$. Indeed,
 we are setting $\yy_{i_1}^{(1)}=(-id_{i_1}^{\star},1)$ for all
$i_1 \in \{1,\ldots, \ell^{\star}\}$ and $\yy_{i_1}^{(1)} =(1,0)$
for $i_1 \in \{ \ell^{\star}+1,\ldots, d \}$. This implies that if
$\id=(\id_1,\dots,\id_\ell)$ is an identity, then $\langle
\yy_{i_1},\id_{i_1}\rangle=0$ for all $i_1\in\{1,\dots,\ell\}$ if
and only if $\id$ is a prefix of $\id^\star$. As the adversary is not
allowed to make ``Reveal key'' queries for prefixes of $\id^\star$,
then the condition is satisfied. This implies that $d_1=1$ with
overwhelming probability. Also, the pre-output stage $\mathcal{P}$ in this case simply outputs $d_2=1$, always. This implies
that $d^\A_{\neg abort}$ is 1 with overwhelming probability and, as a
consequence, $\Pr[d_{\neg abort}^{\mathcal{A}}=1\ |\
\REAL_{\HF,\LHF,\omega,\zeta}^{\mathcal{A}}] $ is negligibly close to 1 and $\left(\Pr[d_{\mathcal{A}}=1\ |\
\REAL_{\HF,\LHF,\omega,\zeta}^{\mathcal{A}} \wedge d_{\neg abort}^{\mathcal{A}}=1 ] -\
\frac{1}{2}\right)$ is negligibly close to $\Pr[d_\A=1\
|\ \REAL_{\HF,\LHF,\omega}^\A]\ -\ \frac{1}{2}$. These are precisely
the requirements for $\HF$ to fulfill conditions (ii) and (iii), with $\delta=1$.\qed
\end{proof}


\subsection{Adaptive Security: Proof of Lemma \ref{lemma-lower-bound}}
\label{proof-lower-bound}

Fix the view of the adversary $\A$, which implies fixing the queried
identities $\id^{(1)},\dots,\id^{(q)},\id^{\star}$. Although we are
assuming that the adversary $\A$ makes the maximum number of
queries, with a smaller number of queries we would have the same
bounds. We abbreviate $\eta=\eta(IS,\id^{\star})$,  $\yy=\yy^{(1)}$
and also call $\ell^\star$ the depth of the challenge identity
$\id^\star$. For an integer $t$, define the event

$$E_t:\ \ \displaystyle\bigwedge_{i=1}^q\left(\bigvee_{i_1=1}^{\ell^{(i)}} \left(\langle \id_{i_1}^{(i)},\yy_{i_1}\rangle\neq0\mod t\right) \right)
\wedge \bigwedge_{i_1=1}^{\ell^{\star}} \left(\langle
\id_{i_1}^{{\star}},\yy_{i_1}\rangle= 0\mod t\right)$$

We denote
$\mathbf{Y}=\{y_1',\dots,y_d',\yy_{1},\dots,\yy_{d},\xi_1,\dots,\xi_d\}$.
For all $i_1 \in \{1,\ldots,\ell^\star \}$, we have   $$\langle
\id_{i_1},\yy_{i_1}\rangle=y_{i_1}'+\displaystyle\sum_{k=2}^\mu
\yy_{i_1}[k]\id_{i_1}[k]-2\xi_{i_1}q$$ for some $0\leq\xi_{i_1}\leq
\mu-1$. In particular, observe that $$0\leq
y_{i_1}'+\displaystyle\sum_{k=2}^\mu
\yy_{i_1}[k]\id_{i_1}[k]<2q\mu<p.$$ Let us define the value
 $\xi_{i_1}^{\star}:=\lfloor(y_{i_1}'+\displaystyle\sum_{k=2}^\mu \yy_{i_1}[k]\id^\star_{i_1}[k])/2q\rfloor$.
 If we have the two conditions
 \begin{eqnarray*}
&& \xi_{i_1}  = \xi_{i_1}^\star \qquad \qquad \qquad \qquad \qquad
\qquad \textrm{ for each } i_1 \in
 \{1,\ldots,\ell^\star\} \\
&&  \langle \id^\star_{i_1},\yy_{i_1}\rangle = 0\mod 2q,
\end{eqnarray*}
   then clearly $\langle \id^\star_{i_1},\yy_{i_1}\rangle=0\mod p$.
   Also, if $\langle \id_{i_1},\yy_{i_1}\rangle\neq 0\mod 2q$, then we also have $\langle \id_{i_1},\yy_{i_1}\rangle\neq 0\mod p$ . Using these observations, we have
\begin{align*}
\eta&\geq \Pr[\xi_{i_1}=\xi_{i_1}^{\star} \quad \forall i_1\in\{1,\dots,\ell^{\star}\}] \Pr_{\mathbf{Y}}[E_p~|~\xi_{i_1}=\xi_{i_1}^{\star} \quad \forall i_1\in\{1,\dots,\ell^{\star}\}]\\
&=\frac{1}{\mu^{\ell^{\star}}} \Pr_{\mathbf{Y}}[E_p~|~\xi_{i_1}=\xi_{i_1}^{\star}  \quad \forall i_1\in\{1,\dots,\ell^{\star}\}]\\
&\geq \frac{1}{\mu^{\ell^{\star}}} \Pr_{\mathbf{Y}}[E_{2q}~|~\xi_{i_1}=\xi_{i_1}^{\star}  \quad \forall i_1\in\{1,\dots,\ell^{\star}\}]\\
&=\frac{1}{\mu^{\ell^{\star}}}
\Pr_{\mathbf{Y}'}[E_{2q}~|~\xi_{i_1}=\xi_{i_1}^{\star} \quad \forall
i_1\in\{1,\dots,\ell^{\star}\}]
\end{align*}
where $\mathbf{Y}'$ contains $\{\yy_{1}',\dots,\yy_{d}'\}$ and
$\yy_{i_1}'=(y_{i_1}',\yy_{i_1}[2],\dots,\yy_{i_1}[\mu])$ for all
$i\in\{1,\dots,d\}$. Note that the second inequality above holds
because of the condition $\xi_{i_1}=\xi_{i_1}^\star$. If we now
define
 $\eta_{2q}=\Pr_{\mathbf{Y}'}[E_{2q}~|~\xi_{i_1}=\xi_{i_1}^{\star}
~ \forall i_1\in\{1,\dots,\ell^{\star}\}],$  we  just showed that
$\eta\geq \frac{1}{\mu^{\ell^\star}}\cdot\eta_{2q}$. \\

\indent Now, we observe some facts about  $\langle
\id_{i_1},\yy_{i_1}\rangle$. First, observe that $\langle
\id_{i_1},\yy_{i_1}\rangle$ and $\langle
\id_{i'_1},\yy_{i'_1}\rangle$ are independent for $i_1\neq i'_1$.
This is because of the way $\mathbf{Y}$ is chosen. \\
\indent Also, note that for any $\id_{i_1}$, $a\in\Z$,
$\Pr_{\mathbf{Y}'}[\langle \id_{i_1},\yy_{i_1}\rangle=a\mod
2q]=1/2q$. This is because for any choice of
$\yy_{i_1}[2],\dots,\yy_{i_1}[\mu]$, there is only one value of
$y_{i_1}'$ for which the equality holds. \\
\indent Consider $\id=(\id_1,\ldots,\id_\ell)\neq
\id'=(\id_1',\ldots,\id_{\ell'}')$ and $\id$ not being a prefix of
$\id'$ and $a,b\in\Z$. First, if $\ell
>\ell'$, for each $i_1 \in \{\ell'+1, \ldots, \ell \}$ such that $\id_{i_1} \neq \id_{i_1}'$, we have
\begin{align*}
\Pr_{\mathbf{Y}'}[\langle \id_{i_1},\yy_{i_1}\rangle=a\mod 2q|
\displaystyle \bigwedge_{k=1}^{\ell'} \langle
\id'_{k},\yy_{k}\rangle=b\mod 2q] =\Pr_{\mathbf{Y}'}[\langle
\id_{i_1},\yy_{i_1}\rangle=a\mod 2q]=1/2q.
\end{align*}
This happens because $\displaystyle\bigwedge_{k=1}^{\ell'}\langle
\id'_{k},\yy_{k}\rangle=b\mod 2q$ does not impose any condition on
$\langle \id_{i_1}, \yy_{i_1} \rangle$ and we can apply the same
arguments as previously. \\
\indent On the other hand, for all $i_1\leq \ell'$,
$\Pr_{\mathbf{Y}'}[\langle \id_{i_1},\yy_{i_1}\rangle=a\mod 2q|
\displaystyle\bigwedge_{k=1}^{\ell'}\langle
\id'_{k},\yy_{k}\rangle=b\mod 2q]$ is either $0$, if
$\id_{i_1}=\id'_{i_1}$ or $1/2q$, if $\id_{i_1}\neq \id'_{i_1}$. The
second fact is because, if $\id_{i_1}\neq \id'_{i_1}$,  there exists
an index $j$ for which $\id_{i_1}[j]=1$ and $\id_{i_1}'[j]=0$ or the
other way around. We see that, if we fix all $\yy_{i_1}[i]$ for
$i\neq j$ so that $\langle \id'_{i_1},\yy_{i_1}\rangle=b\mod 2q$,
then there is only one value for $\yy_{i_1}[j]$ so that $\langle
\id_{i_1},\yy_{i_1}\rangle=a\mod 2q$.

With all these observations, we calculate the following bound on
$\eta_{2q}$:

\begin{align}
\eta_{2q}&=\Pr_{\mathbf{Y}'}\left[\displaystyle\bigwedge_{i=1}^q\left(\bigvee_{i_1=1}^{\ell^{(i)}}
\left(\langle \id_{i_1}^{(i)},\yy_{i_1}\rangle\neq0\mod 2q\right)
\right) | \bigwedge_{i_1=1}^{\ell^{\star}} \left(\langle
\id_{i_1}^{{\star}},\yy_{i_1}\rangle= 0\mod 2q\right)\right]\cdot\nonumber\\
&\qquad\qquad\qquad\qquad\qquad\qquad\qquad\qquad\qquad\qquad\cdot
\Pr_{\mathbf{Y}'}\left[\bigwedge_{i_1=1}^{\ell^{\star}}
\left(\langle \id_{i_1}^{{\star}},\yy_{i_1}\rangle= 0\mod
2q\right)\right]\nonumber\\
&=1/(2q)^{\ell^{\star}}\Pr_{\mathbf{Y}'}\left[\displaystyle\bigwedge_{i=1}^q\left(\bigvee_{i_1=1}^{\ell^{(i)}}
\left(\langle \id_{i_1}^{(i)},\yy_{i_1}\rangle\neq0\mod 2q\right)
\right) | \bigwedge_{i_1=1}^{\ell^{\star}} \left(\langle
\id_{i_1}^{{\star}},\yy_{i_1}\rangle= 0\mod 2q\right)\right]\nonumber\\
&=1/(2q)^{\ell^{\star}}\left(1-\Pr_{\mathbf{Y}'}\left[\displaystyle\bigvee_{i=1}^q\left(\bigwedge_{i_1=1}^{\ell^{(i)}}
\left(\langle \id_{i_1}^{(i)},\yy_{i_1}\rangle=0\mod 2q\right)
\right) | \bigwedge_{i_1=1}^{\ell^{\star}} \left(\langle
\id_{i_1}^{{\star}},\yy_{i_1}\rangle= 0\mod
2q\right)\right]\right)\label{upper-bound}
\end{align}
\begin{align}
\geq
1/(2q)^{\ell^{\star}}\left(1-\displaystyle\sum_{i=1}^{q}\Pr_{\mathbf{Y}'}\left[\bigwedge_{i_1=1}^{\ell^{(i)}}\left(
\langle \id_{i_1}^{(i)},\yy_{i_1}\rangle=0\mod 2q\right) |
\bigwedge_{i_1=1}^{\ell^{\star}} \left(\langle
\id_{i_1}^{{\star}},\yy_{i_1}\rangle= 0\mod
2q\right)\right]\right)\nonumber
\end{align}

We now focus on how to bound
$$\Pr_{\mathbf{Y}'}\left[\bigwedge_{i_1=1}^{\ell^{(i)}}\left( \langle
\id_{i_1}^{(i)},\yy_{i_1}\rangle=0\mod 2q\right) |
\bigwedge_{i_1=1}^{\ell^{\star}} \left(\langle
\id_{i_1}^{{\star}},\yy_{i_1}\rangle= 0\mod 2q\right)\right].$$

First, observe that $\langle \id_{i_1}^{(i)},\yy_{i_1}\rangle=0\mod
2q$ is independent to $\langle
\id_{i_1}^{(i)},\yy_{i_1'}\rangle=0\mod 2q$ if $i_1\neq i_1'$. As a
consequence, we only need to compute $\Pr_{\mathbf{Y}'}\left[
\langle \id_{i_1}^{(i)},\yy_{i_1}\rangle=0\mod 2q |
\bigwedge_{k=1}^{\ell^{\star}} \left(\langle
\id_{k}^{{\star}},\yy_{k}\rangle= 0\mod 2q\right)\right]$.

To this end, we consider two cases: that $\ell^{(i)}>\ell^{\star}$
or $\ell^{(i)}\leq \ell^{\star}$. In the first case,  for each $i_1
 \in \{\ell^\star +1,\ldots,\ell^{(i)}\}$ such that
\begin{multline*}\Pr_{\mathbf{Y}'}\left[ \langle
\id_{i_1}^{(i)},\yy_{i_1}\rangle=0\mod 2q |
\bigwedge_{k=1}^{\ell^{\star}} \left(\langle
\id_{k}^{{\star}},\yy_{k}\rangle= 0\mod
2q\right)\right]\\=\Pr_{\mathbf{Y}'}\left[ \langle
\id_{i_1}^{(i)},\yy_{i_1}\rangle=0\mod 2q\right]=1/2q.
\end{multline*} For all indices $i_1 \in \ell^\star$, the same probability is either $1$ or $1/(2q)$.  \\ \indent  If $\ell^{(i)}\leq\ell^{\star}$,
for each $i_1 \in \{1,\ldots,\ell^{(i)}\}$, we have
\begin{multline*}\Pr_{\mathbf{Y}'}\left[ \langle
\id_{i_1}^{(i)},\yy_{i_1}\rangle=0\mod 2q |
\bigwedge_{k=1}^{\ell^{\star}} \left(\langle
\id_{k}^{{\star}},\yy_{k}\rangle= 0\mod
2q\right)\right]\\=\Pr_{\mathbf{Y}'}\left[ \langle
\id_{i_1}^{(i)},\yy_{i_1}\rangle=0\mod 2q | \langle
\id_{i_1}^{{\star}},\yy_{i_1}\rangle= 0\mod 2q\right]\end{multline*}
which is 1 if $\id_{i_1}^{(i)}=\id_{i_1}^{\star}$ and $1/(2q)$
otherwise due to the fact stated above.

Define $\chi_1^{(i)}=\max(\ell^{(i)}-\ell^{\star},0)$ and
$\chi_2^{(i)}=\#\{1\leq i_1\leq \min(\ell^{\star},\ell^{(i)})|
\id_{i_1}^{(i)}\neq\id_{i_1}^{\star}\}$. Note that, by the
restrictions imposed on $\id^{(i)}$, we have
$\chi_1^{(i)}+\chi_2^{(i)}\geq 1$ for all $i\in\{1,\dots,q\}$.
Putting it all together, we find
\begin{multline*}
 \Pr_{\mathbf{Y}'}\left[\bigwedge_{i_1=1}^{\ell^{(i)}}\left( \langle
\id_{i_1}^{(i)},\yy_{i_1}\rangle=0\mod 2q\right) |
\bigwedge_{i_1=1}^{\ell^{\star}} \left(\langle
\id_{i_1}^{{\star}},\yy_{i_1}\rangle= 0\mod 2q\right)\right]
\\=\frac{1}{(2q)^{\chi_1^{(i)}+\chi_2^{(i)}}}\leq \frac{1}{(2q)}
\end{multline*}
We can conclude that $\eta_{2q}\geq
1/(2^{\ell^{\star}+1}q^{\ell^{\star}})$. Combining this bound with
the bound on $\eta$, we get the statement of the Lemma. \qed

\subsection{Adaptive Security: Proof of Theorem \ref{theor-adapt-loss}}
\label{proof-adaptive-loss}

Regarding condition (i) of partial lossiness, let $RL_0$ and $RL_n$ be the games specified in Section \ref{security-analysis}. Let $\widehat{RL_0}$ and $\widehat{RL_n}$ be the games
which are the same as $RL_0$ and $RL_n$ except that they include the
artificial abort stage described above. First, we claim that
$$\Pr[\widehat{RL_0}\Rightarrow 1]-\Pr[\widehat{RL_n}\Rightarrow
1]\leq n \cdot \bigl( (d \cdot \mu +1) \cdot
\mathbf{Adv}^{\mathcal{P}\textrm{-}\mathrm{BDH}_1}(\B_1) + q \cdot
\mathbf{Adv}^{\mathcal{P}\textrm{-}\mathrm{DDH}_2}(\B_2) \bigr)$$

The proof for this statement is almost identical to the proof for
Lemma \ref{indist-setup}.

We now claim that both $\Pr[d_{exp}^\A=1\ |\
\REAL]=\Pr[\widehat{RL_0}\Rightarrow 1]$ and
$$|\Pr[d_{exp}^\A=1\ |\
\LOSSY]-\Pr[\widehat{RL_n}\Rightarrow 1]| \in
\mathsf{negl}(\varrho).$$
 These statements imply  condition (i) of partial lossiness. \\
\indent The proof of the last two  statements is identical to the
proof for Theorem \ref{theor-select-loss}: when using the auxiliary
input $\yy^{(0)}$, $\HFEv$ will always be injective. On the other
hand, when using $\yy^{(1)}$, $\HFEv$ will be lossy for identities
$(\id_1,\dots,\id_\ell)$ such that $\langle
\yy_{i_1}^{(1)},\id_{i_1} \rangle = 0$ for all
$i_1\in\{1,\dots,\ell\}$, with lossiness $\lambda\left( \HFEv
\big(\mathsf{pms},\mathsf{mpk},(\id_1,\ldots,\id_{\ell}),. \big)
\right) \geq n-\log{p} =\omega$, and $\HFEv$ will be injective for
the other identities.

Seeing that condition (ii) is fulfilled is straightforward: from
Lemma \ref{lemma-lower-bound}, we know that $\eta_{low}\leq\Pr
[d_1=1\ |\ \REAL]$ and, by construction, $\eta_{low}\leq \Pr[d_2=1\
|\ d_1=1\wedge\REAL]$. This gives us the lower bound $\eta_{low}^2$
on $\Pr[d_{\neg abort}=1\ |\ \REAL]$, which is non-negligible.
However, the value of $\epsilon_1$ will not be this lower bound.
Instead, we will show that another condition (which we give at the
end of this proof) is satisfied and  guarantees the existence of
$\epsilon_1$ and $\epsilon_2$.  The proof for this is actually very
similar to the proof in \cite{Wat05}. Here, we briefly outline the
details. First, due to Chernoff bounds we have the following:

$$ \Pr[d_{\neg abort}^\A=1\ |\ d_\A=1\wedge \REAL]\geq
\eta_{low}(1-\frac{1}{4}\zeta)$$ and
$$\Pr[d_{\neg abort}^\A=1\ |\ d_\A=0\wedge
\REAL]\leq \eta_{low}(1+\frac{3}{8}\zeta)
$$

We refer the reader to \cite{Wat05} for details on how to compute
these bounds.

Let us now assume the existence of a PPT adversary $\A$ such that
$$
\Pr[d_\A=1\ |\ \REAL]\ -\
\frac{1}{2}>\zeta
$$
for some non-negligible $\zeta$. This implies that we can write
$\Pr[d_\A=1\ |\ \REAL] >1/2+\zeta$ and $\Pr[d_\A=0\ |\
\REAL]<1/2-\zeta$. Combining these inequalities with the two
inequalities above (from Chernoff bounds) we obtain

\begin{multline}
\Pr[d_\A=1\ |\ \REAL]\cdot\Pr[d_{\neg abort}^\A=1\ |\ d_\A=1 \wedge
\REAL] \\
-\ \Pr[d_\A=0\ |\
\REAL]\cdot\Pr[d_{\neg abort}^\A=1\ |\
d_\A=0\wedge \REAL]>2\delta\zeta\label{eq:conseq-2}
\end{multline}
with $\delta=13\eta_{low}/16=13/\left(32\cdot(2
q\mu)^{d}\right)$. Now, let us observe that
\begin{align*}
\Pr[d_\A=1\ |\ \REAL]\ \cdot \ \Pr[d_{\neg abort}^\A=1\ |\ d_\A=1
\wedge \REAL]  =\Pr[d_{\neg abort}^\A=1\wedge d_\A=1\ |\
\REAL]\end{align*} and
\begin{align*}&\Pr[d_\A=0\ |\
\REAL]\cdot\Pr[d_{\neg abort}^\A=1\ |\
d_\A=0\wedge
\REAL]=\\&=\Pr[d_{\neg abort}^\A=1\wedge
d_\A=0\
|\ \REAL]\\
&=\Pr[d_{\neg abort}^\A=1\ |\
\REAL]\cdot\Pr[d_\A=0\ |\
d_{\neg abort}^\A=1\wedge \REAL]\\
&=\Pr[d_{\neg abort}^\A=1\ |\
\REAL]\cdot(1-\Pr[d_\A=1\ |\
d_{\neg abort}^\A=1\wedge \REAL])\\
&=\Pr[d_{\neg abort}^\A=1\ |\
\REAL]-\Pr[d_{\neg abort}^\A=1\wedge d_\A=1\
|\ \REAL]
\end{align*}

Combining these two equalities with inequality (\ref{eq:conseq-2}), we obtain
$$
2\delta\zeta < 2 \Pr[d_{\neg abort}^\A=1\wedge
d_\A=1\ |\ \REAL] - \Pr[d_{\neg abort}^\A=1\ |\
\REAL].
$$
Dividing this inequality by 2 and using the fact that $$\Pr[d_{\neg
abort}^\A=1\wedge d_\A=1\ |\ \REAL] = \Pr[d_{\neg abort}^\A=1\ |\
\REAL] \cdot  \Pr[d_\A=1\ |\ d_{\neg abort}^\A=1 \wedge \REAL],$$ we
obtain the relation
$$
\delta \zeta < \Pr[d_{\neg abort}^\A=1\ |\ \REAL] \cdot \left( \Pr[d_\A=1\ |\ d_{\neg abort}^\A=1 \wedge \REAL] - \frac{1}{2} \right)
$$

Finally, this shows that there exist $\epsilon_1$ and $\epsilon_2$ such that (ii) and (iii) are satisfied and that their product is $\delta$.

\qed

\section{Adaptive-id Secure  Deterministic (H)IBE}
\label{deterministic-scary-part}

A hierarchical identity-based deterministic encryption scheme
(HIB-DE) is a tuple of efficient algorithms $\DE=(\DESetup,\
\DEMKg,\ \DEKg,\ \DEDel, \DEEnc ,\ $ $\DEDec )$. The setup algorithm
 $\DESetup$ takes as input a security
parameter $\varrho \in \mathbb{N}$, the (constant) number of levels
in the hierarchy $d\in\N$, the length of the identities $\mu\in
\mathsf{poly}(\varrho)$ and the length of the plaintexts $s
\in\mathsf{poly}(\varrho)$, and outputs a set of global public
parameters $\mathsf{pms}$, which specifies an identity space $\IdSp$
and the necessary mathematical objects and hash functions. The
master key generation
 algorithm $\DEMKg$ takes as input $\mathsf{pms}$ and outputs a master public key $\mathsf{mpk}$ and a master secret key $\mathsf{msk}$. The
 key generation algorithm $\DEKg$ takes as input $\mathsf{pms}$, $\mathsf{msk}$ and a hierarchical
 identity $(\id_1,\ldots,\id_{\ell})\in \IdSp$, for some $\ell \geq 1$ and
 outputs a secret key $\mathbf{SK}_{(\id_1,\ldots,\id_\ell)}$. The delegation algorithm $\DEDel$ takes as
 input $\mathsf{pms}$, $\mathsf{msk}$, a hierarchical identity $(\id_1,\ldots,\id_{\ell})$, a secret key
 $\mathbf{SK}_{(\id_1,\ldots,\id_\ell)}$ for it, and an additional identity $\id_{\ell + 1}$; the output is a secret
 key $\mathbf{SK}_{(\id_1,\ldots,\id_\ell,\id_{\ell + 1})}$ for the hierarchical identity $(\id_1,\ldots,\id_\ell,\id_{\ell + 1})$ iff $(\id_1,\ldots,\id_\ell,\id_{\ell + 1})\in \IdSp$.
 The evaluation algorithm $\DEEnc$ takes as input $\mathsf{pms}$, $\mathsf{msk}$, a hierarchical identity $\id=(\id_1,\ldots,\id_{\ell})$ and a
 value $m \in \{0,1\}^s$; the result of the evaluation is denoted as $C$. Finally, the inversion algorithm $\DEDec$ takes as
 input  $\mathsf{pms}$, $\mathsf{msk}$, a hierarchical identity $\id=(\id_1,\ldots,\id_{\ell})$, a secret key $\mathbf{SK}_{\id}$
 for it and an evaluation $C$, and outputs a value $\tilde{m} \in
  \{0,1\}^s$. \\
 \indent
A HIB-DE satisfies the property of correctness if
$$
\DEDec \big(\mathsf{pms},\mathsf{mpk},\id,\mathbf{SK}_{\id}, \DEEnc
\big(\mathsf{pms},\mathsf{mpk},\id=(\id_1,\ldots,\id_{\ell}),m \big)
\big) \ =\ m ,
$$
for any $m \in \{0,1\}^s$, any
$\mathsf{pms},(\mathsf{mpk},\mathsf{msk})$ generated by  $\DESetup$
and $\DEMKg$, any hierarchical identity $(\id_1,\ldots,\id_\ell)\in
\IdSp$ and any secret key $\mathbf{SK}_{(\id_1,\ldots,\id_\ell)}$
generated either by running
$\DEKg\big(\mathsf{pms},\mathsf{msk},(\id_1,\ldots,\id_{\ell})
\big)$ or by applying   the delegation algorithm $\DEDel$ to secret
keys of shorter hierarchical identities.

%

Bellare \textit{et al.} \label{security-defs-de} \cite{BFOR08} gave
several definitions for deterministic encryption and proved them
equivalent. In the case of  {\it block sources}\footnote{Informally,
a block source is a distribution of message vectors where each
component has high min-entropy conditionally on previous ones. }
\cite{BFO08}, Boldyreva {\it et al.} \cite{BFO08} proved that
single-challenge security (called PRIV1 security in
\cite{BBO07,BFO08}) is  equivalent to multi-challenge security
(referred to as PRIV security \cite{BBO07}, where the adversary is
given a vector of challenge ciphertexts) in the sense of
indistinguishability-based   definitions. The simplified
indistinguishability-based notion, called PRIV1-IND hereafter, is
somewhat handier to work with and we thus use this one.

We define PRIV1-IND-ID security, the natural analogue of PRIV1-IND
security in the IBE scenario. The security notion is defined via the
following  experiment between a challenger and an adversary. Some
instructions depend on whether we are in the \emph{selective} or
\emph{adaptive}  security case.

%
%
%
\begin{definition} \label{HIB-DE-exp-def}
We define $\textsf{Guess}_{\DE}^\A(M)$, for a random variable $M$ as
follows:
\begin{itemize}
\item[0.] The challenger  $\mathcal{C}$ chooses  global
parameters $\mathsf{pms}$ by running $\DESetup$. The parameters
$\mathsf{pms}$ are given to $\A$, who replies by choosing a
hierarchical identity $\id^\dagger
=(\id_1^\dagger,\ldots,\id_{\ell^\dagger}^\dagger)$, for some
$\ell^\dagger \leq d$.

\item[1.] The challenger runs $(\mathsf{mpk},\mathsf{msk}) \leftarrow \DEMKg(\mathsf{pms})$ and sends $\mathsf{mpk}$ to $\A$. Also,
 two lists $QS \leftarrow \emptyset$, $IS \leftarrow \emptyset$ are initialized.

\item[2.]   $\A$ is allowed to make a number of adaptive queries for hierarchical identities $\id = (\id_1,\ldots,\id_{\ell})$ (where hierarchical identities are encoded as hierarchical
vectors $(\vec{X}_1,\ldots,\vec{X}_\ell)$).
\\ \vspace{-0.3 cm}
\begin{itemize}
\item[-] {\bf Create-key}: $\A$ provides $\id$
and the challenger $\mathcal{C}$ creates a private key $SK_\id$ by
running $\DEKg(\mathsf{pms},\mathsf{msk},\id )$. The list $QS$ is
updated as $QS=QS\cup \{\id\}$.
\item[-] {\bf Create-delegated-key}: $\A$ provides $\id=(\id_1,\ldots,\id_{\ell})$ and $\id_{\ell+1}$ such that $\id\in QS$.
The challenger $\mathcal{C}$ then computes $SK_{\id'}$ for
$\id'=(\id_1,\dots,\id_{\ell+1})$ by running the delegation
algorithm
$\DEDel\big(\mathsf{pms},\mathsf{mpk}_\beta,SK_\id,\id_{\ell+1}\big)$.
The list $QS$ is updated as $QS=QS\cup \{\id'\}$.
\item[-] {\bf Reveal-key}: $\A$ provides $\id$ with the
restriction that if $\A$ is selective, then
$\id\not\leq\id^\dagger$. $\mathcal{C}$ returns $\perp$ if
$\id\not\in QS$. Otherwise, $SK_{\id}$ is returned to $\A$ and the
list $IS$ is updated   as $IS=IS
\cup \{\id\}$.\\
\vspace{-0.3 cm}
\end{itemize}

\item[3.] The adversary $\A$ outputs a hierarchical identity
$\id^\star=(\id_1^\star,\ldots,\id_{\ell^\star}^\star)$. In the
selective setting, we impose $\ell^\star=\ell^\dagger$ and
$\id_j^\star=\id_j^\dagger$ for each $j \in
\{1,\ldots,\ell^\star\}$. In the adaptive case, no element of $IS$
is a prefix of $\id^\star$.

\item[4.] The challenger encrypts a message $m$ sampled from the given message distribution $M$. The resulting
ciphertext is sent to $\A$.

\item[5.] $\A$ outputs a bit $b'\in\{0,1\}$, which is the output of
the experiment.
\end{itemize}
\end{definition}

Recall that a random variable $X$ over $\{0,1\}^s$  is called a
$(t,s)$-source if $H_\infty(X)\geq t$, where $H_\infty(X)$ is the
min-entropy of $X$; $H_\infty(X)=-\log(\max_xP_X(x)).$ We now give
the identity-based version of IND-PRIV1 security of \cite{BFO08}.

\begin{definition}
An $s$-bit encryption scheme $\DE$ is PRIV1-IND-ID secure for
$(t,s)$-sources if for any $(t,s)$-sources $M_0$ and $M_1$ and all
polynomial time adversaries $\A$, the PRIV1-IND-ID advantage
$$Adv_{\DE}^{\mathrm{priv1}\textrm{-}\mathrm{ind}\textrm{-}\mathrm{id}}(\A,M_0,M_1)=\Pr[\textsf{Guess}_{\DE}^\A(M_0)=1]-\Pr[\textsf{Guess}_{\DE}^\A(M_1)=1]$$
of $\A$ against $\DE$ is negligible.
\end{definition}

\subsection{Universal $\HF$ implies deterministic
encryption}\label{det-construction} As it can be seen from the
definition, a hierarchical identity-based encryption scheme is very
close  to an  HIB-TDF, even syntactically.  Indeed, as shown in
\cite{BFO08}, in the public key setting a natural way to construct a
deterministic $\mathsf{DE}$ encryption scheme is by  defining the
algorithms of  $\mathsf{DE}$ as their natural counterparts in some
lossy LTDF.  Boldyreva \textit{et al.} \cite{BFO08} show that if the
lossy function is also universal, this construction is PRIV1-IND
secure. For functions not satisfying this property,  the result
follows also directly by an extension of  the Crooked Leftover Hash
Lemma of Dodis and Smith \cite{DS05} given by Boldyreva \emph{et
al.} for this purpose. Similarly, in the identity-based setting, one
can construct a hierarchical deterministic identity-based encryption
scheme $\DE$  from $\HF$, a lossy HIB-TDF, by defining the
algorithms of $\DE$ from those of  $\HF$ in the natural way (in this
case, we say that $\DE$ is defined by $\HF$).  We will show that
with the additional requirement of universality, such a construction
is secure in some natural extension of PRIV1-IND security. To do
without this extra condition, it is possible to use the same
techniques of Boldyreva \emph{et al.}.

More specifically, universality requires that if $\LHF$ is the lossy
sibling of a $(\omega,\delta)$-lossy  HIB-TDF function, then
$\HFSetup$, $\LHFMKg$ and $\HFEv(\mathsf{pms,mpk},\id^\star,\cdot)$
induce a universal hash function when we have a lossy function for
$\id^\star$, i.e., $\lambda\left( \HFEv
\big(\mathsf{pms},\mathsf{mpk}_1,(\id_1^\star,\ldots,\id_{\ell^\star}^\star),\cdot
\big) \right) \geq \omega$.  It is not difficult to see that both
\cite{BKPW} and our construction induce a universal hash in the
lossy mode.

\begin{theorem}
If $\HF$ is a universal $(\omega, \delta)$-lossy HIB-TDF function
with input $\{0,1\}^n$, then the hierarchical deterministic
encryption function $\DE$ defined by it is secure for poly-time sampleable
$(t,s)$-sources such that $t\geq n-\omega+2\log(1/\epsilon)$ for
some negligible $\epsilon$. In particular, for any two
$(t,s)$-sources $M_0$ and $M_1$ and for every PRIV1-IND-ID adversary
$\B$ against $\DE$ there exists a PPT adversary $\A$ against $\HF$
such that
$$\mathbf{Adv}_{\HF,\LHF,\mathcal{P},\omega,\zeta}^{\mathrm{lossy}}(\A) \geq \frac{1 }{3} \cdot \delta  \cdot  \mathbf{Adv}_{\DE}^{\mathrm{priv1-ind-id}}(\mathcal{B},M_0,M_1)- \nu(\varrho).$$
\end{theorem}

\begin{proof}
We use similar ideas as those of \cite{BFO08}. 
Namely, our proof goes as follows: assume that there is an adversary
$\B$ that breaks PRIV1-ID-IND security of the deterministic
encryption scheme. Namely, there exist  two $(t,s)$-sources
$M_0,M_1$ and a non-negligible  $\zeta$  such that
$$\mathbf{Adv}_{\DE}^{\mathrm{priv1}\textrm{-}\mathrm{ind}\textrm{-}\mathrm{id}}(\B,M_0,M_1)=\Pr[\textsf{Guess}_{\DE}^\B(M_0)=1]-\Pr[\textsf{Guess}_{\DE}^\B(M_1)=1]\geq \zeta. $$ Then, we build an adversary $\A$ that breaks the security of the underlying $\HF$, that is,  $\A$ will interact with a challenger according to the lossy or the real experiment and will tell the difference between both scenarios with non-negligible probability.

$\A$ forwards an identity $\id^\dagger$ to its challenger, which is
some random identity in the adaptive case.  When the challenger runs
the setup and gives the output to $\A$, $\A$ forwards this
information to $\B$. When $\B$ asks for a secret key for a
hierarchical identity $\id$, $\A$ forwards the query to the
experiment and forwards the reply to $\B$. At some point, $\B$
outputs $\id^\star$. Adversary $\A$ then forwards $\id^\star$ to its
challenger, chooses $\gamma\gets\{0,1\}$ at random and encrypts
$m_\gamma \sample M_{\gamma}$ under the identity $\id^\star$. Note
that this corresponds to an execution of
$\textsf{Guess}_{\DE}^\B(M_{\gamma})$. After some more secret key
queries, $\B$ outputs a bit $\gamma' \in \{0,1\}$ and $\A$ outputs
$d_\A=1$ if $\gamma=\gamma'$ and $d_\A=0$ otherwise.

The security analysis is very similar to the one in Section
\ref{sec:app_IND_CPA}. A simple argument shows that
$$ \Pr[\gamma'=\gamma|\ \REAL]-\frac{1}{2}=\Pr[d_\A=1|\
\REAL]-\frac{1}{2}\geq \zeta/2,$$ since $\A$ perfectly simulated the
experiment $\textsf{Guess}_{\DE}^\B(M_{\gamma})$ with $\B$. On the
 other hand, in the $\LOSSY$ setting
when $\id^\star$ is lossy, the advantage of $\B$ in guessing
$\gamma$ is negligible because the universality property means that
the output distribution of the encryption algorithm is independent
of the input distribution, and this holds regardless of whether
$d_{\neg abort}^\A=1$, therefore:
$$\Pr[d_\A=1|\ \LOSSY \wedge d_{\neg abort}^\A=1] \leq \dfrac{1}{2}+\nu,$$  for some negligible $\nu$. At this point the analysis
follows as in the analysis of the IND-ID-CPA case (proof of Theorem
\ref{cpa-construction-secure}).
\end{proof}

\subsection*{Related Work}

In a recent, independent work \cite{XXZ}, Xie, Xue and Zhang gave a
lattice-based deterministic IBE construction in the auxiliary input
setting \cite{BS11}. While their scheme does provide adaptive
security in the auxiliary input setting, it was described as a
single-level IBE. Our results are incomparable to theirs: on one
hand, we do not consider auxiliary inputs; on the other hand, we are
concerned with pairing-based schemes in a hierarchical setting and
focus on a more powerful primitive.

\section{On Hedged (H)IBE}\label{hedged-appendix}

Bellare \textit{et al.} \cite{BBN09} studied the problem of designing \emph{hedged} public key encryption schemes, which remain secure even if the randomness used for encryption is relatively bad. They consider two notions of hedged security: non-adaptive and adaptive. A non-adaptive attacker can make only a single
challenge encryption query, whereas an adaptive attacker can make up to $q$ challenge adaptive encryption queries, all of them for the same public key.

Bellare \textit{et al.} give in \cite{BBN09}  different constructions of public key encryption scheme enjoying non-adaptive hedged security. For instance, their construction $\mathsf{RtD}$ combines an IND-CPA secure encryption scheme and a deterministic PRIV1-IND secure encryption scheme to achieve non-adaptive hedged security. Then they prove a general result stating that a public key encryption scheme which enjoys both anonymity and non-adaptive hedged security, already enjoys adaptive hedged security. They show how to provide the necessary anonymity property by using universal lossy trapdoor functions.

The notion of non-adaptive hedged security, with only one challenge encryption query, can be easily extended to the identity-based setting, as well as the non-adaptively secure constructions in \cite{BBN09}. In particular, since our new notion of partial lossiness for (H)IB-TDFs implies both IND-CPA secure IBE and PRIV1-IND secure (H)IB deterministic encryption, we have that the existence of (H)IB-TDFs enjoying partial lossiness implies (through the identity-based version of $\mathsf{RtD}$) non-adaptive hedged (H)IBE secure against adversaries who choose the challenge identity in an adaptive way.

Regarding the notion of adaptive hedged security, the extension to the identity-based setting is more challenging, because the adversary could choose different identities for each of the $q$ challenge adaptive encryption queries. The partial lossiness definitions in \cite{BKPW} and in this paper do not seem to imply any positive result in that case, because these definitions consider only one challenge (lossy) identity. Maybe more general definitions for the partial lossiness notion of (H)IB-TDFs are needed to overcome this obstacle, in the line of all-but-$N$ trapdoor functions \cite{HLOV11} and all-but-many trapdoor functions \cite{Hof12}. We leave the problem of designing (H)IBE schemes with adaptive hedged security as an interesting open problem.

\end{document}